%% file: temp1merge.tex
\documentclass[11pt,onecolumn]{article}

\usepackage{amsfonts,amssymb,url,float}
\usepackage{makeidx}
\usepackage{multirow}
\usepackage{verbatim}

\usepackage{color,latexsym}
\usepackage{amsmath,amsthm}
\usepackage{xy}\xyoption{all}

\makeatletter
\newif\if@restonecol
\makeatother

\usepackage[ruled,vlined]{algorithm2e}

\usepackage{amssymb}

\usepackage[margin=1.8cm,nohead]{geometry}

\usepackage{booktabs,longtable} 
\usepackage{array}
\usepackage{multirow}

\usepackage[x11names]{xcolor} 
\usepackage{ifpdf}

\ifpdf
  \usepackage[pdftex]{graphicx}
\else
  \usepackage{graphicx}
\fi
\graphicspath{{./} {figures/}}
\usepackage{flafter} 

\usepackage{chapterbib}
\usepackage[sectionbib,super,square,sort&compress]{natbib}

\newtheorem{theorem}{Theorem}
\newtheorem{lemma}[theorem]{Lemma}

\newtheorem{definition}{Definition}

{\makeatletter
 \gdef\xxxmark{%
   \expandafter\ifx\csname @mpargs\endcsname\relax 
     \expandafter\ifx\csname @captype\endcsname\relax 
       \marginpar{xxx}
     \else
       xxx 
     \fi
   \else
     xxx 
   \fi}
 \gdef\xxx{\@ifnextchar[\xxx@lab\xxx@nolab}
 \long\gdef\xxx@lab[#1]#2{{\bf [\xxxmark #2 ---{\sc #1}]}}
 \long\gdef\xxx@nolab#1{{\bf [\xxxmark #1]}}
 \gdef\turnoffxxx{\long\gdef\xxx@lab[##1]##2{}\long\gdef\xxx@nolab##1{}}%
}

\title{\textbf{Temperature 1 Self-Assembly:}\\ Deterministic Assembly in 3D and Probabilistic Assembly in 2D}
\author{
  Matthew Cook%
    \thanks{Institute of Neuroinformatics, University of Zurich and ETH Zurich, Switzerland,
      \protect\url{cook@ini.phys.ethz.ch}}
\and
  Yunhui Fu%
    \thanks{Department of Computer Science, University of Texas - Pan American,
      \protect\url{fuyunhui@gmail.com}}
\and
  Robert Schweller%
    \thanks{Department of Computer Science, University of Texas - Pan American,
      \protect\url{schwellerr@cs.panam.edu}}
}
\date{}

\begin{document}

\maketitle

\begin{abstract}
We investigate the power of the Wang tile self-assembly model at temperature 1, a threshold value that permits attachment between any two tiles that share even a single bond.  When restricted to deterministic assembly in the plane, no temperature 1 assembly system has been shown to build a shape with a tile complexity smaller than the diameter of the shape.  In contrast, we show that temperature 1 self-assembly in 3 dimensions, even when growth is restricted to at most 1 step into the third dimension, is capable of simulating a large class of temperature 2 systems, in turn permitting the simulation of arbitrary Turing machines and the assembly of $n\times n$ squares in near optimal $O(\log n)$ tile complexity.  Further, we consider temperature 1 probabilistic assembly in 2D, and show that with a logarithmic scale up of tile complexity and shape scale, the same general class of temperature $\tau=2$ systems can be simulated with high probability, yielding Turing machine simulation and $O(\log^2 n)$ assembly of $n\times n$ squares with high probability.  Our results show a sharp contrast in achievable tile complexity at temperature 1 if either growth into the third dimension or a small probability of error are permitted.  Motivated by applications in nanotechnology and molecular computing, and the plausibility of implementing 3 dimensional self-assembly systems, our techniques may provide the needed power of temperature 2 systems, while at the same time avoiding the experimental challenges faced by those systems.
\end{abstract}

\newpage

\begin{table*}[t]
\caption{ In this table we summarize the state of the art in achievable tile complexities and computational power for tile self-assembly in terms of temperature 1 versus temperature 2 assembly, 2-dimensional versus 3-dimensional assembly, and deterministic versus probabilistic assembly. Our contributions are contained in rows 2 and 3.} \small
\begin{center}
\begin{tabular}{||c||c|c|c||}
\hline
\hline
& \multicolumn{2}{|c|}{\textbf{$n\times n$ Squares}} & \textbf{Computational} \\
& \multicolumn{1}{|c}{LB} & \multicolumn{1}{c|}{UB} & \textbf{Power}\\
\hline
&\multicolumn{2}{|c|}{}&\\
Temperature 2, $2D$ & \multicolumn{2}{|c|}{$\Theta(\frac{\log n}{\log\log n})$} & Universal\\
Deterministic&\multicolumn{1}{|c}{(see \cite{Rothemund:2000:PSC})} & \multicolumn{1}{c|}{(see \cite{Adleman:2001:RTP})}& (see\cite{Winfree:1998:ASA})\\
\hline
&&&\\
\textbf{Temperature 1, $3D$} & $\Omega(\frac{\log n}{\log\log n})$ & $O(\log n)$ & Universal\\
\textbf{Deterministic}&\multicolumn{1}{|c|}{(see \cite{Rothemund:2000:PSC})} & \multicolumn{1}{|c|}{(Thm.\ref{thm:square})}& (Thm.\ref{theorem:3dturing})\\
\hline
&&&Time Bounded\\
\textbf{Temperature 1, $2D$} & $\Omega(\frac{\log n}{\log\log n})$ & $O(\log^2 n)$ & Turing Simulation\\
\textbf{Probabilistic}&\multicolumn{1}{|c|}{(Thm.\ref{thm:kolmog})} & \multicolumn{1}{|c|}{(Thm.\ref{thm:square2d})}& (Thm.\ref{theorem:2dturning})\\
\hline
&&&\\
Temperature 1, $2D$ & $\Omega(\frac{\log n}{\log\log n})$ & $2n-1$ & Unknown\\
Deterministic &\multicolumn{1}{|c|}{(see \cite{Rothemund:2000:PSC})} & \multicolumn{1}{|c|}{(see \cite{Rothemund:2000:PSC})}& \\
\hline
\hline

\end{tabular}
\end{center}

\label{table:summary}
\end{table*}

\section{Introduction}
Self-assembly is the process by which simple objects autonomously assemble into an organized structure.  This phenomenon is the driving force for the creation of complex biological organisms, and is emerging as a powerful tool for bottom up fabrication of complex nanostructures.  One of the most fruitful classes of self-assembly systems is DNA self-assembly.  The ability to synthesize DNA strands with specific base sequences permits a highly reliable technique for programming strands to assemble into specific structures.  In particular, molecular building blocks or \emph{tiles} can be assembled with distinct bonding domains~\cite{Fu:1993:DDC,LaBean:2000:CAL}.  These DNA tiles can be designed to simulate the theoretical bonding behavior of the \emph{Tile Assembly Model}\cite{Rothemund:2000:PSC}.

In the Tile Assembly Model, the particles of a self-assembly system are modeled by four sided Wang tiles for 2D assembly, or 6 sided Wang cubes in 3D.  Each side of a tile represents a distinct binding domain and has a specific glue type associated with it.  Starting from an initial \emph{seed} tile, assembly takes place by attaching copies of different tile types in the system to the growing seed assembly one by one.  The attachment is driven by the affinities of the glue types.  In particular, a tile type may attach to a growing seed assembly if the total bonding strength from all glues abutting against the seed assembly exceeds some given parameter called the \emph{temperature}.  If the assembly process reaches a point when no more attachments are possible, the produced assembly is denoted terminal and is considered the output assembly of the system.

 Motivated by bottom up nanofabrication of complex devices and molecular computing, a number of fundamental problems in the tile assembly model have been considered\cite{RoPaWi04,BarSchRotWin09,CheSchGoeWin07,MaoLabReiSee00,Winfree:1996:DSA,LiuShaSee99,MaoSunSee99}.  A few problems are as follow: (1) \emph{shape fabrication}, given a target shape $\Upsilon$, design a system of tile types that will uniquely assemble into shape $\Upsilon$ that uses as few distinct tile types as possible; 2) \emph{molecular computing}~\cite{Winfree:1998:ASA, Brun:2008:SNP},  given some assembled input assembly that encodes a description of a computational problem, design a tile system that will read this input and assemble a structure that encodes the solution to the computational problem; (3) \emph{shape replication}\cite{Abel:2010:SRT,SW05}, given a single copy of a preassembled input shape or pattern, efficiently create a number of replicas of the input shape or pattern.


 While a great body of work has emerged in recent years considering problems in the tile assembly model, almost all of this work has focussed on temperature 2 assembly in which tiles require 2 separate positive strength glue bonds to attach to the growing seed assembly.  This is in contrast to the simpler temperature 1 model which permits attachment given any positive strength bond.  It seems that some fundamental increase in computational power and efficiency is achieved by making the step from temperature 1 to temperature 2.  In fact, there is no known 2D construction to deterministically assemble a width $n$ shape in fewer that $n$ distinct tiles types at temperature 1.  This is in contrast to efficient temperature 2 systems which assemble large classes of shapes efficiently, including $n\times n$ squares in optimal $\theta(\frac{\log n}{\log\log n})$ tile types.  In fact, the ability to simulate universal computation and assemble arbitrary shapes in a number of tile types close to the Kolmogorov complexity of a given shape at temperature 2 in 2D~\cite{Soloveichik:2005:CSA} has resulted in limited interest in exploring 3D assembly, as it would appear no substantial power would be gained in the extra dimension.

 While temperature 2 assembly yields very efficient, powerful constructions, it is not without its drawbacks.  One of the main hurdles preventing large scale implementation of complex self-assembly systems is high error rates.  The primary cause of these errors in DNA self-assembly systems stems from the problem of insufficient attachments of tile types~\cite{Ashish04errorfree,Winfree:2004:PTS,chen:2008:DAC,majumder:2008:CRP}.  That is, in practice tiles often attach with less than strength 2 bonding despite carefully specified lab settings meant to prevent such insufficient bonds.  Inherently, this is a problem specific only to temperature 2 and higher systems. For this reason, development of temperature 1 theory may prove to be of great practical interest.

 Because of a perceived lack of power, temperature 1 assembly has received little attention compared to the more powerful temperature 2 assembly.  In addition, directions such as 3D assembly have not received substantial attention stemming from a perceived lack of ability to increase the functionality of the already powerful temperature 2 systems.  Interestingly, we find that both directions are fruitful when considered together; temperature 1 assembly systems in 3D are nearly as powerful as temperature 2 systems, suggesting that both the perception of limited temperature 1 power and the perception of limited value in 3D are not completely accurate.

\paragraph{Our Results.}
  In this paper we show that temperature 1 deterministic tile assembly systems in 3D can simulate a large class of temperature 2 systems we call \emph{zig zag systems}. We further show that this simulation grants both: (1) near optimal $O(\log n)$ tile type efficiency for the assembly of $n\times n$ squares and (2) universal computational power.  Further, in the case of 2D probabilistic assembly, we show similar results hold by achieving $O(\log^2 n)$ efficient square assembly and the ability to efficiently simulate any time bounded Turing machine with arbitrarily small chance of error.  The key technique utilized in our constructions is a method of limiting glue attachment by create geometrical barriers of growth that prevent undesired attachments from propagating.  This technique is well known in the field of chemistry as \emph{steric hindrance} or \emph{steric protection}~\cite{HellerPugh1,HellerPugh2,GotEtAl00,WadeOrganicChemistry91} where a chemical reaction is slowed or stopped by the arrangement of atoms in a molecule.
 These results show that temperature 1 assembly is not as limited as it appears at first consideration, and perhaps such assemblies warrant more consideration in light of the potential practical advantages of temperature 1 self-assembly in DNA implementations.

\paragraph{Practical Drawbacks of Temperature 1 Self-Assembly.}
While temperature 1 assembly avoids many of the practical hurdles limiting temperature 2 assembly, temperature 1 assembly also introduces new issues.  In particular, the problem of multiple nucleation, in which tiles begin to grow without the presence of the seed tile, is a more substantial problem at temperature 1.  However, further research into temperature 1 assembly may suggest and motivate new techniques to limit such errors.  In the specific case of multiple nucleation, we discuss in this paper as future work a new design technique for temperature 1 self-assembly to limit such errors, even in a pure temperature 1 assembly model.  By fully exploring techniques such as this, and any new techniques that may emerge, temperature 1 self-assembly may emerge as a practical alternative to temperature 2 assembly.

\paragraph{Related Work}  Some recent work has been done in the area of proving lower bounds for temperature 1 self-assembly.  Doty et. al~\cite{doty:2008:LSA} show a limit to the computational power of temperature 1 self-assembly for \emph{pumpable} systems.  Munich et. al~\cite{manuch:2009:TLB} show that temperature 1 assembly of a shape requires at least as many tile types as the diameter of the assembled shape if no mismatched glues are permitted.  In terms of positive results, Chandran et. al~\cite{Chandran:2009:TCL} consider the probabilistic assembly of lines with expected length $n$ (at temperature 1) and achieve $O(\log n)$ tile complexity.  Kao and Schweller~\cite{Kao:2008:RSA} and Doty~\cite{Doty:2009:RSA} use a variant of probabilistic self-assembly (at temperature 2) to reduce distinct tile type complexity.  Demaine et. al~\cite{Demaine:2007:SSA} and Abel et. al~\cite{Abel:2010:SRT} utilize steric hindrance to assist in the assembly and replication of shapes over a number of stages.

\paragraph{Paper Layout.}  In Section~\ref{sec:basics} we define the Tile Assembly Model, in Section~\ref{sec:zigzag3D} we describe an algorithm to convert a temperature 2 \emph{zig zag} system into an equivalent temperature 1 3D system, or a probabilistic 2D system, in Section~\ref{sec:square} we show how temperature 1 systems can efficiently assemble $n\times n$ squares, in Section~\ref{sec:turing} we show that temperature 1 systems can simulate arbitrary Turing machines, in Section~\ref{sec:sim} we discuss preliminary experimental simulations, and in Section~\ref{sec:future} we discuss further research directions.

\section{Basics}\label{sec:basics}

\subsection{Definitions: the Abstract Tile Assembly Model in 2 Dimensions}
To describe the tile self-assembly model, we make the following
definitions.  A tile type $t$ in the model is a four sided Wang tile denoted by the ordered quadruple $(\textrm{north}(t),
\textrm{east}(t), \textrm{south}(t), \textrm{west}(t))$. The entries of the quadruples are glue types taken from an alphabet $\Sigma$ representing the north, east, south, and west edges of the Wang
tile, respectively.  Each pair of glue types are assigned a non-negative integer bonding strength (0,1, or 2 in this paper) by the glue function ${\Sigma}^2$ to $ \{0,1,\ldots\}$.  It is assumed that
$G(x,y) = G(y,x)$, and there exists a $\tt{null}$ in $\Sigma$ such
that $\forall x \in \Sigma$, $G(\tt{null},x) = 0$. In this paper we
assume the glue function is such that $G(x,y) = 0$ when $x\neq y$
and denote $G(x,x)$ by $G(x)$.

A \textit{tile system} is an ordered triple $\langle T, s, \tau\rangle$ where $T$ is a set of tiles
called the $\textit{tileset}$ of the system, $\tau$ is a positive
integer called the \emph{temperature} of the system and $s \in T$ is a
single tile called the $\textit{seed}$ tile.  $|T|$ is referred to as the \emph{tile complexity} of the system.  In this paper we only consider
temperature $\tau=1$ and $\tau=2$ systems.

Define a $\textit{configuration}$ to be a mapping from
${\mathbb{Z}}^{2}$ to $T$ $\bigcup$ $\{\tt{empty}\}$, where
$\tt{empty}$ is a special tile that has the $\tt{null}$ glue on each
of its four edges.  The \emph{shape} of a configuration is defined
as the set of positions $(i,j)$ that do not map to the empty tile.
For a configuration~$C$, a tile $t \in T$ is said to be
\textit{attachable} at the position $(i,j)$ if $C(i,j) = \tt{empty}$
and $ G(\textrm{north}(t) , \textrm{south}(C(i,j+1)))+
G(\textrm{east}(t) , \textrm{west}(C(i+1, j)))+ G(\textrm{south}(t)
, \textrm{north}(C(i, j-1)))+ G(\textrm{west}(t) ,
\textrm{east}(C(i-1, j))) \geq \tau$.  For configurations $C$ and
$C'$ such that $C(x,y) = \tt{empty}$, $C'(i,j) = C(i,j)$ for all
$(i,j) \neq (x,y)$, and $C'(x,y) = t$ for some $t\in T$, define the
act of \textit{attaching} tile $t$ to $C$ at position $(x,y)$ as the
transformation from configuration $C$ to $C'$.  For a given tile
system \textbf{T}, if a supertile $B$ can be obtained from a
supertile $A$ by the addition of a single tile we write
$A\rightarrow_T B$.  Further, we denote $A\rightarrow_T$ as the set
of all $B$ such that $A\rightarrow_T B$ and $\rightarrow_T^*$ as the
transitive closure of $\rightarrow_T$.

Define the \textit{adjacency graph} of a configuration $C$ as
follows. Let the set of vertices be the set of coordinates $(i,j)$
such that $C(i,j)$ is not empty.  Let there be an edge between
vertices $(x_1, y_1)$ and $(x_2, y_2)$ iff $|x_1 - x_2| + |y_1 -
y_2| = 1$. We refer to a configuration whose adjacency graph is
finite and connected as a $\textit{supertile}$.  For a supertile
$S$, denote the number of non-empty positions (tiles) in the
supertile by $\textrm{size}(S)$.  We also note that each tile $t \in
T$ can be thought of as denoting the unique supertile that maps
position $(0,0)$ to $t$ and all other positions to $\tt{empty}$.
Throughout this paper we will informally refer to tiles as being
supertiles.

\subsection{The Assembly Process}
\paragraph{Deterministic Assembly}
Assembly takes place by \textit{growing} a supertile starting with
tile $s$ at position $(0,0)$.  Any $t \in T$ that is attachable at
some position $(i,j)$ may attach and thus increase the size of the
supertile.  For a given tile system, any supertile that can be
obtained by starting with the seed and attaching arbitrary
attachable tiles is said to be \emph{produced}. If this process
comes to a point at which no tiles in $T$ can be added, the
resultant supertile is said to be \textit{terminally} produced. For
a given shape $\Upsilon$, a tile system $\Gamma$ \textit{uniquely
produces} shape $\Upsilon$ if for each produced supertile $A$, there
exists some terminally produced supertile $A'$ of shape $\Upsilon$
such that $A\rightarrow_T^* A'$.  That is, each produced supertile
can be grown into a supertile of shape $\Upsilon$.  The $\textit{tile complexity}$ of a shape $\Upsilon$ is the minimum tile set size required to uniquely assemble $\Upsilon$.  For an assembly system which uniquely assembles one supertile, the system is said to be \emph{deterministic}.

\paragraph{Probabilistic Assembly}
For non-deterministic assembly systems, we define the probabilistic assembly model to place probability distributions on which tiles attach throughout the assembly process.  To study this model we can think of the
assembly process as a Markov chain where each producible supertile
is a state and transitions occur with non-zero probability from
supertile $A$ to each $B\in A\rightarrow_T$.  For each $B\in
A\rightarrow_T$, let $t_B$ denote the tile added to $A$ to get $B$.
The transition probability from $A$ to $B$ is defined to be

$$\mathrm{TRANS}(A,B)=\frac{1}{|A\rightarrow_T|}.$$

The probability that a tile system $T$ terminally assembles a
supertile $A$ is thus defined to be the probability that the Markov
chain ends in state $A$.  Further, the probability that a system
terminally assembles a shape $\Upsilon$ is the probability the chain
ends in a supertile state of shape $\Upsilon$.

\subsection{Extension to 3 Dimensions}
Extending the model, we can consider assembly in 3 dimensions by considering Wang cubes with added ``up'' and ``down'' glues ($\textrm{up}(t)$  and $\textrm{down}(t)$), and configurations of tile types as mappings from $Z^3$ to a tile set $T$.  In this model, assembly begins as before with an initial seed cube (informally we will refer to cubes as tiles) at the origin, with cubes attaching to the north, west, south, east, top, or bottom of already placed tiles if the total strength of attachment from the glue function meets or exceeds the temperature threshold $\tau$.

%
\begin{figure}[htbp]
\centering
\includegraphics[width=0.25\linewidth]{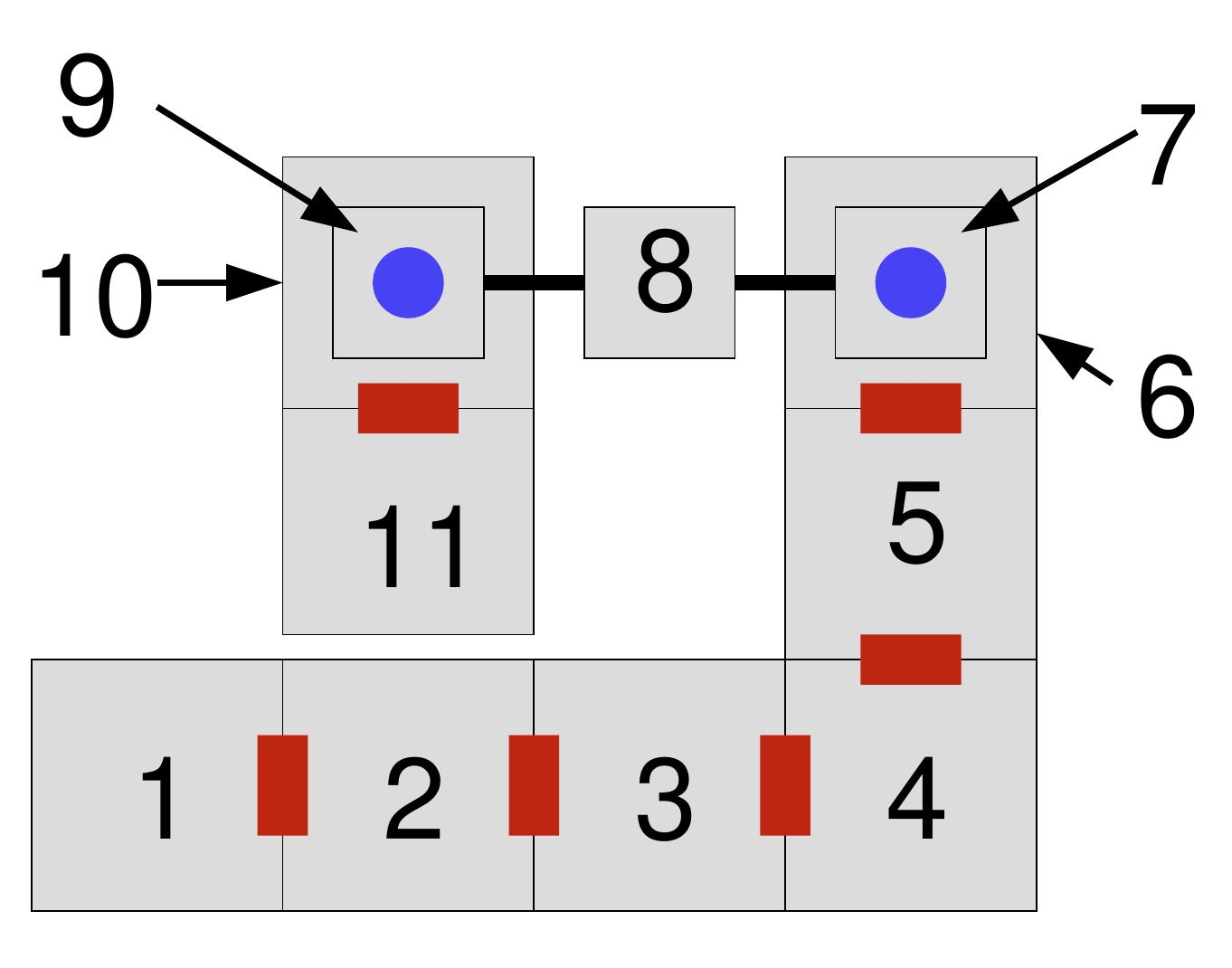}
\caption{For 3 dimensional tile assemblies, we create figures with the notation depicted above. First, large tiles denote tiles placed in the $z=0$ plane, while the smaller squares denote tiles placed in the $z=1$ plane.  The red connectors between bottom tiles denotes some unique glue shared by the tiles in the figure, as does the thin black line connecting tiles in the top plane.  Blue circles denote a unique glue connecting the bottom of the top tile with the top of the tile below it.  In this example, the tiles are numbered showing the implied order of attachment of tiles assuming the '1' tile is a seed tile.}
\label{fig:notation}
\end{figure}
%

In this paper we only consider temperature $\tau = 1$ assembly systems in 3D.  Further, we only consider systems that assemble tiles within the $z=0$ or $z=1$ plane.  To depict 3 dimensional tile sets and 3 dimensional assemblies, we introduce some new notation in Figure~\ref{fig:notation}.

\subsection{Scaled Simulations}
Our constructions in this paper relate to algorithms that generate  temperature 1 tile systems that simulate the assembly of a given temperature 2 tile system.  Formally, we define the notion of one tile system simulating another system.

 \begin{definition}
A tile system $\Upsilon = (T',s',\tau')$ simulates a deterministic tile system $\Gamma=(T,s,\tau)$ at horizontal scale factor $scale_x$, vertical scale factor $scale_y$, and tile complexity factor $C$ if:
\begin{enumerate}
\item there exists an onto function $\textrm{SIM}: R\subseteq T' \Rightarrow T$ for some subset $R\subseteq T'$ such that for any non empty tile type $\Gamma(i,j)$ at position $(i,j)$ in the terminal assembly of $\Gamma$, there exists exactly one tile type $r\in R$ such that $\Upsilon(k,\ell) = r$ for some $k$ and $\ell$ such that $i\cdot scale_x \leq k\leq (i+1)scale_x$ and $j\cdot scale_y \leq \ell\leq (j+1)scale_y$, and $\textrm{SIM}(r) = \Gamma(i,j)$.
\item for each position $(i,j)$ such that $\Gamma(i,j)=empty$, it is the case that for each $k$ and $\ell$ such that $ i\cdot scale_x \leq k\leq (i+1)scale_x$ and $j\cdot scale_y \leq \ell \leq (j+1)scale_y$, then $\Upsilon(k,\ell) = empty$.
\item $|T'| \leq C|T|$
\end{enumerate}

 \end{definition}

 Informally, a system $\Upsilon$ simulates a system $\Gamma$ if for each tile type $t$ placed by system $\Gamma$, a tile type representing $t$ is place by $\Upsilon$ at the same position but scaled up by the vertical and horizontal scale factors.  In the case that the simulating system is a 3D system, all tile positions from the 3D system are projected on the $z=0$ plane to apply the definition.  In the case that $\Upsilon$ is a probabilistic assembly system, then the tile placed at a given position by $\Upsilon$ is a random variable, and we are interested in the probability that each tile placed by $\Gamma$ is correctly simulated by $\Upsilon$ according to some given assignment of representative tile types.

\section{Simulation of Temperature $\tau=2$ Zig-Zag Systems at Temperature $\tau=1$}\label{sec:zigzag3D}

\subsection{Zig-Zag Tile Systems}
%
\begin{figure}[htbp]
\centering
\includegraphics[width=0.30\linewidth]{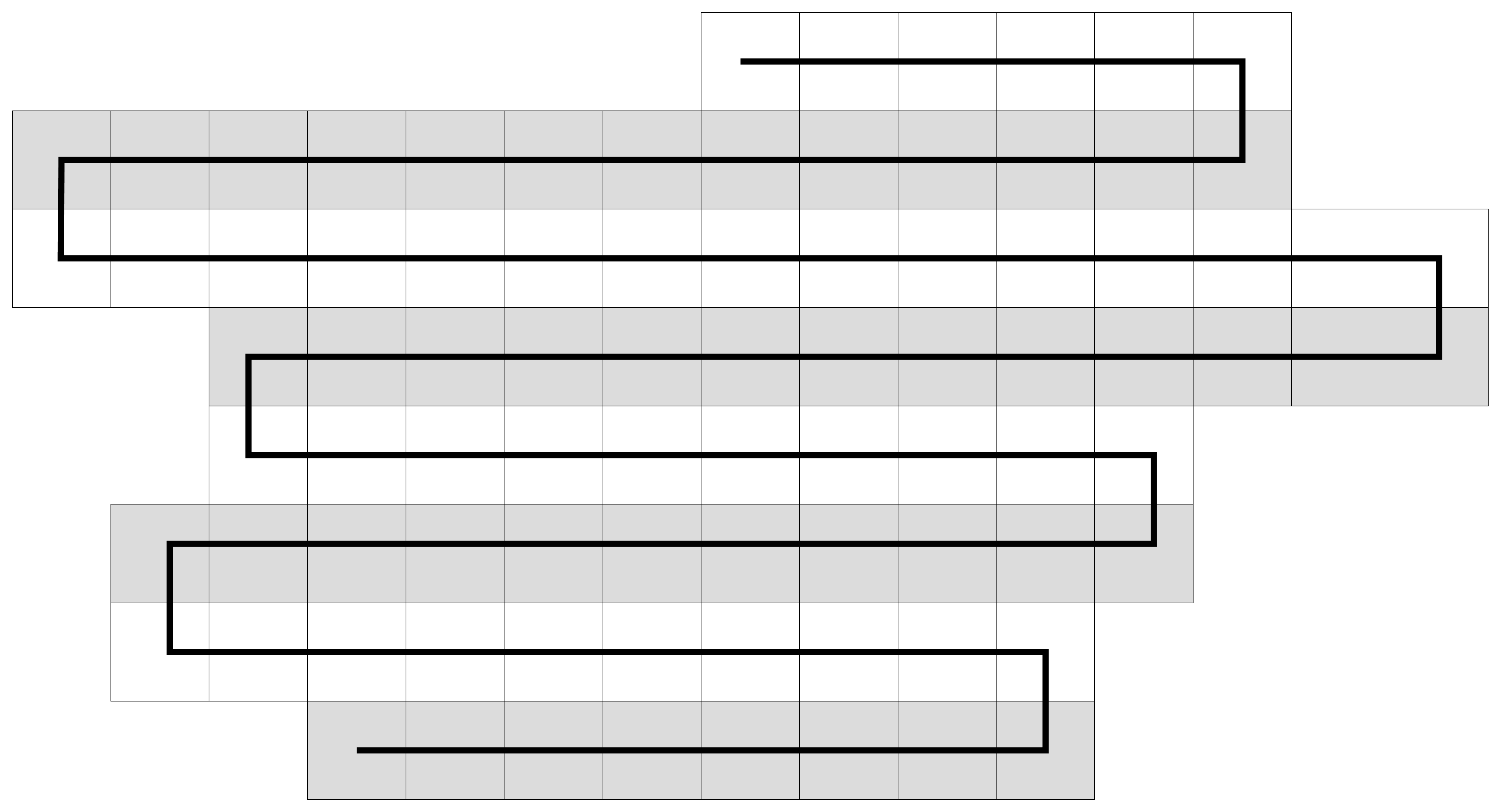}
\caption{A zig-zag tile system alternates growth from left to right at each vertical level.}
\label{fig:zigdiagram}
\end{figure}
%

Our first result is a construction that shows a large class of temperature 2 self-assembly systems (in 2 dimensions) can be simulated at temperature $\tau =1$ if growth is permitted in the third dimension.

\paragraph{Zig-zag assembly systems.}  An assembly system is defined to be a zig-zag system if it has the following two properties:
\begin{itemize}
\item The assembly sequence, which specifies the order in which tile types are attached to the assembly along with their position of attachment, is unique.  For such systems denote the type of the $i^{th}$ tile to be attached to the assembly as $type(i)$, and denote the coordinate location of attachment of the $i^{th}$ attached tile by $(x(i), y(i))$.
\item For each $i$ in the assembly sequence of a zig-zag system, if $y(i-1)$ is even, then either $x(i) = x(i-1) + 1$ and $y(i-1) = y(i)$,  or $y(i) = y(i-1) + 1$.  For odd $y(i-1)$, either $x(i) = x(i-1) - 1$ and $y(i-1)$, or $y(i) = y(i-1) + 1$.
\end{itemize}

Intuitively, a zig-zag assembly is an assembly which must place tiles from left to right up to the point at which a first tile is place into the next row north.  At this point growth must grow from right to left, until the next growth north once again.  Thus even rows grow from left to right, odd rows grow from right to left, as shown in Figure~\ref{fig:zigdiagram}.

The key property zig-zag systems have is that it is not possible for a west or east strength 1 glue to be exposed without the position directly south of the open slot already being tiled, or, if this is possible, the glue is redundant and can be removed to achieve the same assembly sequence.

\subsection{3D Temperature 1 Simulation of Zig-Zag Tile Systems}

In this section we show that any temperature 2 zig-zag tile system can be simulated efficiently by a temperature 1 3D tile system.  The scale and efficiency of the simulation is dependant on the number of distinct strength 1 glues that occur on either the north or south edge of a tile type in the input tile system's tileset $T$.  For a tileset $T$, let $\sigma_T$ denote the set of strength 1 glue types that occur on either a north or south face of some tile type in $T$.  The following simulation is possible.

\begin{theorem}\label{thm:zigzag}
For any 2 dimensional temperature 2 zig-zag tile system $\Gamma=\langle T, s,2 \rangle$ with $\sigma_T$ denoting the set of distinct strength 1 north/south glue types occuring in $T$, there exists a 3D temperature 1 tile system that simulates $\Gamma$ at vertical scale = 4, horizontal scale = $O(\log |\sigma_T|)$, and (total) tile complexity $O(|T|\log |\sigma_T|)$, which constitutes a $O(\log |\sigma_T|)$ tile complexity scale factor.
\end{theorem}

\begin{proof}
To simulate a temperature 2 zig-zag tileset with a temperature 1 tileset, we replace each tile type in the original system with a collection of tiles designed to assemble a larger \emph{macro tile} at temperature 1.  The goal will be to ``fake'' cooperative temperature 2 attachment of a given tile from the original system by utilizing geometry to force an otherwise non-deterministic growth of tiles to \emph{read} the north input glue of a given macro tile.

Consider an arbitrary temperature 2 zig-zag tile system $\Gamma=\langle T, s, 2\rangle$ where $\sigma_T$ denotes the alphabet of strength 1 glue types that occur on either the north or south face of some tile type in $T$.  Index each glue $g\in \sigma_T$ from $0$ to $|\sigma_T|-1$.  For a given $g\in \sigma_T$, let $b(g)$ denote the index value of $g$ written in binary.  We will refer to each glue $g$ by the new name $c_{b(g)}$ when describing our construction.

To prove the theorem, we construct a tileset with $O(\log |\sigma_T|)$ distinct tile types for each tile $t\in T$.  In particular, the number of tile types will be at most $|T|(12*\log |\sigma_T| + 12) = O(|T|\log|\sigma_T|)$.

To assign tile types to the new temperature 1 assembly system, take all west growing tile types that have an east glue of type 'x' for some strength 1 glue 'x'.  If there happens to be both east growing and west growing tile types that have glue 'x' as an east glue, first split all such tile types in to two separate tiles, one for each direction.

For the set of west growing tile types with east glue $x$, a collection of tile types are added to the simulation set as a function of $x$ and the subset of $\sigma_T$ glues that appear on the south face of the collected tile types. The tile types added are depicted in Figure~\ref{fig:zigzag}.  In the example from Figure~\ref{fig:zigzag} there are 4 distinct tile types that share an east $x$ glue type.  As these tiles are west growing tile types, the temperature 2 simulation places each of these tiles using the cooperative bonding of glue $x$ and glue $c_{111}$ in the case of the right most tile type.  The east and south glue types of a west growing tile can be thought of as \emph{input glues}, which uniquely specify which tile type is placed, in turn specifying two \emph{output glues}, glue type $a$ to the west and glue type $c_{111}$ to the north in the case of the right most tile type.  At temperature 1, we cannot directly implement this double input function by cooperative bonding as even a single glue type is sufficient to place a tile. Instead, we use glue type to encode the east input, and geometry of previously assembly tiles to encode the south input.

In more detail, the tile types specified in Figure~\ref{fig:zigzag} (b) constitute a nondeterministic chain of tiles whose possible assembly paths for a binary tree of depth $\log$ of the number of distinct south glue types observed in the tile set being simulated.  In the given example, the tree starts with an input glue $(x, -)$.  This glue knows the tile to its east has a west glue of type $x$, but has no encoding of what glue type is to the south of the macro tile to be placed.  This chain of tiles nondeterministically places either a $0$ or a $1$ tile, which in turn continues growth along two separate possible paths, one denoted by glue type $(x,0)$, and the other by glue type $(x,1)$.  By explicitly encoding all paths of a binary prefix tree ending with leaves for each of the south glues of the input tile types, the decoding tiles nondeterministically pick exactly one south glue type to pair with the east glue type $x$, and output this value as a glue specifying which of the 4 tile types should be simulated at this position.

Now, to eliminate the non-determinism in the decoding tiles, we ensure that the geometry of previous place tiles in the $z=1$ plane is such that at each possible branching point in the binary tree chain, exactly 1 path is blocked, thus removing the non-determinism in the assembly as depicted in Figure~\ref{fig:zigzag} (d).  This prebuilt geometry is guaranteed to be in place by the correct placement of the simulated macro tile placed south of the current macro tile.  One the proper tile type to be simulated is decode, the 2 output values, $a$ and $c_{111}$ in the case of the rightmost tile type of Figure~\ref{fig:zigzag} (a), must be propagated west and north respectively.  This is accomplished by the collection of tile types depicted in Figure~\ref{fig:zigzag} (c).  Now that the north and west output glues have been decode, this macro tile will assemble a geometry of \emph{blocking tiles} to ensure that a tile using this north glue as a south glue input will deterministically decode the correct glue binary string.  In particular, pairs of tiles are placed in the plane $z=2$ for each bit of the output binary string.  The pair is place to locations vertically higher for $1$ bits than for $0$ bits.  The next row of macro tiles will then be able to decode this glue type encoded in geometry by applying a binary tree of decoder tiles similar to those shown in Figure~\ref{fig:zigzag} (b).

The complete conversion algorithm from a temperature 2 zig-zag system to a temperature 1 3D system has a large number of special cases.  However, the example worked out in this proof sketch gets at the heart of the idea.  The fully detailed conversion algorithm, with all cases detailed is described in the Appendix in Section~\ref{sec:detailconvertzigzag}.

%
\begin{figure}[htbp]
\centering
\includegraphics[width=\linewidth]{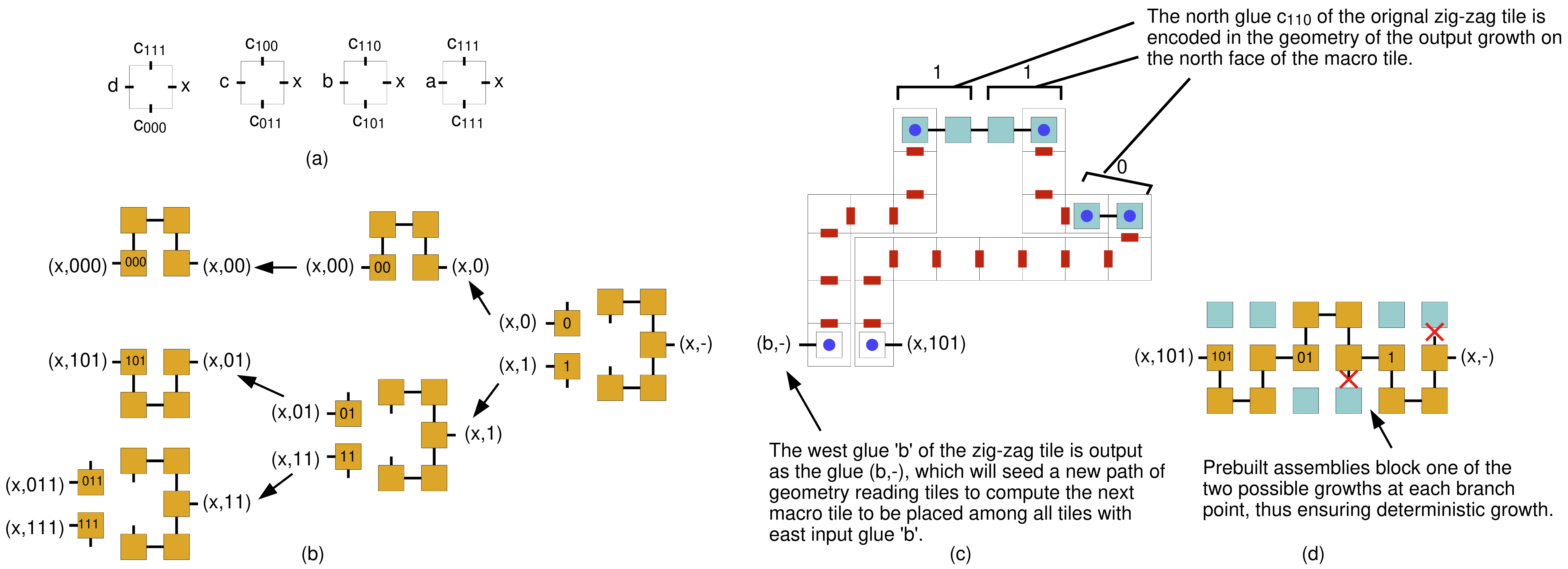}
\caption{This tile set depicts how one tile from a collection of input tile types all sharing the same east glue $x$ are simulated by combining strength 1 glues bonds with geometrical blocking to simulate cooperative binding and correctly place the correct macro tile.}
\label{fig:zigzag}
\end{figure}
%
\end{proof}

\subsection{An Example}
To further illustrate the construction for Theorem~\ref{thm:zigzag} we present a simple zig-zag tile set and corresponding 3D, temperature $\tau=1$ simulation tile set, described in the Appendix in Figures~\ref{fig:simpleexample} and \ref{fig:simpleexampleassembled}.

\subsection{2D Probabilistic Temperature 1 Simulation of Zig-Zag Tile Systems}\label{sec:2dprobtemp1simulation}
In this section we sketch out the simulation of any temperature 2 zig-zag tile system by a 2D probabilistic tile system.  The basic idea is similar to the 3D conversion.  Each tile type for a given east glue $x$ is used to generate a binary tree of nondeterministic chains of tiles to decode a southern input glue encoded by geometry.  However, in 2D, it is not possible to block both branches of the binary tree due to planarity.  Therefore, we only block one side or block neither side.  The basic idea behind the decoding is depicted in Figure~\ref{fig:zigzag2D}.

The key idea is to buffer the length of the geometric blocks encoding bits by a buffing factor of $k$.  In the example from Figure~\ref{fig:zigzag2D}, a 2-bit string is encoded by two length $2k$ stretches of tiles (k=4 in this example), where the binary bit value 1 is represented by a dent that is one vertical position higher than the encoding for the binary bit value 0.  Whereas in the 3D encoding each bit value is encoded with a fixed two horizontal tile position for encoding, we now have an encoding region of size $2k$ for each bit value.

As in the 3D case, we utilize a set of decoder tiles (the orange tiles in Figure~\ref{fig:zigzag2D})
to read the geometry of the north face of the macro tile to the south.
Due to planarity, only one possible growth branch is blocked.
However, in the case of a 0 bit, the $k$ repetitions of the non-deterministic
placement of the branching tile makes it highly likely that
at least one downward growth will beat its competing westward growth over
$k$ independent trials, if $k$ is large.
In the case of 1 bits, the decoder is guaranteed to get the correct answer.

Therefore, we can analyze the probability that a given 0 bit is correctly decoded by bounding the probability of flipping a (biased) coin $k$ times and never getting a single tail.  The coin is biased because the south branch of growth must place 3 consecutive tiles before the west branch places a single tile to successfully read the 0 bit, which happens with probability 1/8\footnote{To make the choice non-biased, multiple copies of the south branching tile types could be included in the tile set, making the 3 consecutive placements happen with equal, or even greater probability than the single tile placement of the west branch.}.  Thus, by making $k$ sufficiently large, we can bound the probability that a given block will make an error by incorrectly interpreting a 0-bit as a 1-bit.  For a zig-zag system that makes $r$ tile attachments, we can therefore set $k$ to be large enough such that with high probability a temperature 1 simulation will make all $r$ attachments without error, yielding the following theorem.

\begin{theorem}\label{thm:zigzag2d}
For any $\epsilon > 0$ and  zig-zag tile system $\Gamma = \langle T, s, 2 \rangle$ whose terminal assembly has size $r$,  there exists a temperature 1, 2D probabilistic tile system that simulates $\Gamma$ without error with probability at least $1-\epsilon$.  The $scale_x$ of this system is $O(\log\sigma_T\log r + \log\sigma_T\log\frac{1}{\epsilon})$, the $scale_y$ is 4, and the tile complexity is $O(|T|\log\sigma_T\log r + |T|\log\sigma_T\log\frac{1}{\epsilon})$.
\end{theorem}
\begin{proof}
We first note that the tile complexity and scale factor of the simulation are determined by the values $|T|$ and $|\sigma_T|$, as in the 3D transformation, as well as the choice of the parameter $k$.  In particular, the tile complexity of the system is $O(|T|k\log\sigma_T)$, the horizontal scale factor is $scale_x = O(k\log\sigma_T)$, and the vertical scale factor is $scale_y = 4$.  Therefore, to show the result we need to determine what parameter choice for $k$ is needed to guarantee that all $r$ block placements are error free with probability at least $1-\epsilon$.

For a given $\epsilon >0$, set $k$ equal to
$$
k = \log_{8/7} r + \log_{8/7}\frac{1}{\epsilon}
$$
As any given block will make an error with probability at most $(\frac{7}{8})^k$, the expected number of total errors over $r$ block placements is
$$
E[errors] = r(\frac{7}{8})^k = \epsilon
$$
Therefore, for $k = \log_{8/7} r + \log_{8/7}\frac{1}{\epsilon}$ the probability of making 1 or more errors is bounded by $\epsilon$, and this value of $k$ achieves the theorem's stated bounds.
\end{proof}


%
\begin{figure}[htbp]
\centering
\includegraphics[width=\linewidth]{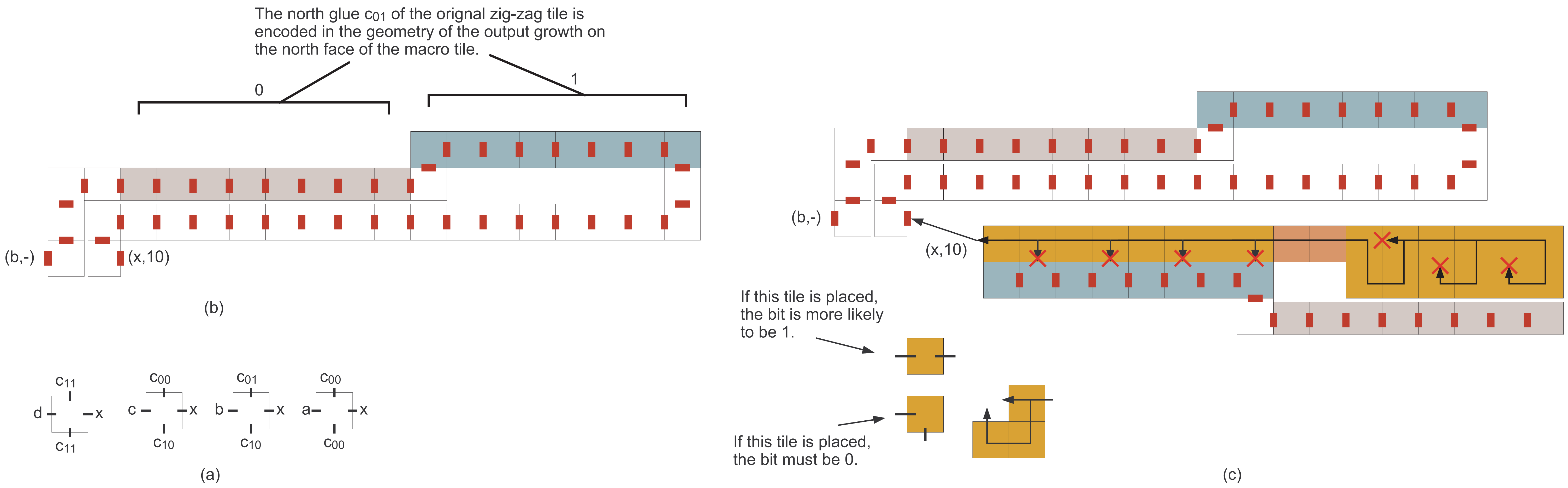}
\caption{This tile set depicts the macro tiles used to transform a zig-zag system into a probabilistic 2D assembly system.}
\label{fig:zigzag2D}
\end{figure}
%

\section{Efficient Assembly of $n\times n$ Squares at Temperature 1}\label{sec:square}
In this section we show how to apply the zig-zag simulation approach from Section~\ref{sec:zigzag3D} to achieve tile complexity efficient systems that uniquely assemble $n\times n$ squares at temperature $\tau=1$ either deterministically in 3D or probabilistically in 2D.

\subsection{Deterministic Temperature 1 assembly of $n\times n$ Squares in 3D}
 We show that the assembly of an $n\times n$ square with a 3D tile system requires at most $O(\log n)$ tile complexity at temperature 1.

\begin{lemma}\label{lemma:zigzagcounter}
For any $n$, there exists a 2D, temperature $\tau = 2$ zig-zag tile system, with north/south glueset $\sigma$ of constant bounded size, that uniquely assembles a $\log n \times n$ rectangle.  Further, this rectangle can be designed so that a unique, unused glue appears on the east side of the northeast placed tile (this allows the completed rectangle to seed a new assembly, as utilized in Theorem~\ref{thm:square}).
\end{lemma}

\begin{proof}
The binary counter depicted by Rothemund and Winfree~\cite{Rothemund:2000:PSC} is a zig-zag tile set and can easily be modified to have a special glue at the needed position.
\end{proof}

\begin{theorem}\label{thm:square}
For any $n\geq 1$, there exists a 3D temperature $\tau=1$ tile system with tile complexity $O(\log n)$ that uniquely assembles an $n\times n$ square.
\end{theorem}
\begin{proof}
  The basis of our proof is a construction for a binary counter that is a zig-zag system with a $O(1)$ size north/south glueset.  Such a tile system can assemble a $\log n$ by $n$ rectangle using $O(\log n)$ tile types.  By Theorem~\ref{thm:zigzag}, this system can be simulated at temperature 1 in 3D to make 1 side of an $n\times n$ square.  By seeding a similar version of this tile system upon the completion of the initial assembly, a second and third side of the square can be completed, permitting a final filler tile to fill in the body of the square.  As the north/south glue set for the counter system is of constant size, the same $O(\log n)$ tile complexity is achieved for the construction of the square at temperature 1 in 3D by combining Lemma~\ref{lemma:zigzagcounter} with Theorem~\ref{thm:zigzag}.
\end{proof}

\subsection{Probabilistic Temperature 1 assembly of $n\times n$ Squares in 2D}
In this section we show that for any value $\epsilon$ and positive integer $n$, there exists a 2D tile system with $O(\log^2 n  + \log n \log\frac{1}{\epsilon})$ tile types that will assemble an $n\times n$ square with probability at least $1-\epsilon$.  For a fixed constant $\epsilon$, this simplifies to a $O(\log^2 n)$ tile complexity.

\begin{theorem}\label{thm:square2d}
For any $\epsilon >0$ and positive integer $n$, there exists a 2D tile system with $O(\log^2 n  + \log n \log\frac{1}{\epsilon})$ tile types that will assemble an $n\times n$ square with probability at least $1-\epsilon$.
\end{theorem}
\begin{proof}
The basic approach of this theorem is the same as Theorem~\ref{thm:square}.  Consider a zig-zag binary counter construction with $O(1)$ size north/south glueset designed to assemble a $\log n$ by $n$ rectangle using $O(\log n)$ tile types.  We can apply Theorem~\ref{thm:zigzag2d} to this system to get a temperature 1 probabilistic simulation of the counter.  Further, as the counter places at most $O(n\log n)$ tile types before terminating, Theorem~\ref{thm:zigzag2d} yields a tile complexity bound of $O(\log^2 n + \log n \log\frac{1}{\epsilon})$ to correctly simulate the counter with probability $1-\epsilon$.  As with the 3D construction, 3 rectangles can be assembled in sequence to build 3 of the borders of the goal $n\times n$ square, and finally filled with a fill tile type.
\end{proof}

\subsection{Kolmogorov Lower Bound for Probabilistic Assembly of $n\times n$ Squares}
Rothemund and Winfree~\cite{Rothemund:2000:PSC} showed that the minimum tile complexity for the deterministic self-assembly of an $n\times n$ square is lower bounded by $\Omega(\frac{\log n}{\log\log n})$ for almost all positive integers $n$.  This applies to both 2D and 3D systems.  In this section we show that this lower bound also holds for probabilistic systems that assemble an $n\times n$ square with probability greater than $1/2$.

\begin{theorem}\label{thm:kolmog}
For almost all positive integers $n$, the minimum tile complexity required to assemble an $n\times n$ square with probability greater than 1/2 is $\Omega(\frac{\log n}{\log\log n})$.
\end{theorem}
\begin{proof}
The Kolmogorov complexity of an integer $n$ with respect to a
universal Turing machine $U$ is $K_U (n) = $ min$|p|$ s.t $U(p) =
b_n$ where $b_n$ is the binary representation of $n$.  A
straightforward application of the pigeonhole principle yields that
$K_U (n) \geq \lceil \log n \rceil - \Delta$ for at least $1-
(\frac{1}{2})^{\Delta}$ of all $n$ (see~\cite{Li:1997:IKC} for
results on Kolmogorov complexity).  Thus, for any $\epsilon > 0$,
$K_U (n) \geq (1- \epsilon)\log n = \Omega (\log n)$ for almost all
$n$.

If there exists a fixed size Turing
machine that takes as input a tile system and outputs the maximum
extent (i.e. width or length) of any shape that is terminally produced by the system with probability more than 1/2, then the $\Omega(\log N)$ bound applies to the number of
bits required to represent any such tile system that assembles an $n\times n$ square
with probability greater than 1/2, as such a system combined with the fixed size Turing machine would constitute a machine that outputs $n$ (note that the unique output of $n$ depends on the probability of assembly exceeding 1/2 as only one such highly probable assembly can exist in this case).  Further, we can represent any tileset
$\Gamma=\langle T, s, \tau \rangle$ with $O(|T|\log |n|)$ bits assuming a constant bounded temperature (there are at most $4|T|$ glues in the system, thus all glues can be labeled uniquely with $O(|T|\log|T|)$ bits).  Thus, for constants $d$ and $c$ we have that $d\log n \leq c|T|\log|T|$. If we assume $T$ is the smallest possible tile set to assemble an $n\times n$ square with probability at least 1/2, we know that $c|T|\log|T| \leq |T|\log\log n$ as there is a known (deterministic) construction to build any $n\times n$ square in $O(\log n)$ tile types~\cite{Rothemund:2000:PSC}.  Thus, we have that $d\log n \leq |T|\log\log n$ which yields that the smallest number of tile types to assemble a square is $|T| = \Omega(\frac{\log n}{\log\log n})$.

The above argument requires the existence of a fixed size Turing machine/algorithm to output the dimension of any terminally produced shape assembled with probability greater than 1/2.  The following algorithm  satisfies this criteria.\\

OutputDimension( Tile system $\Gamma=\langle T, s, \tau \rangle$ )
\begin{itemize}
    \item For k from 1 to infinity do the following
        \begin{itemize}
            \item iterate over all producible $k$-tile assemblies of $\Gamma$
                \begin{itemize}
                    \item If the current assembly $X$ is terminal, compute the probability that $\Gamma$ assembles $X$.  If the probability exceeds 1/2, output the dimension of $X$ and halt.
                \end{itemize}
        \end{itemize}
\end{itemize}

\end{proof}

\subsection{Why not $O(\frac{\log n}{\log\log n})$ square building systems at Temperature $\tau=1$?}
Temperature $\tau=2$ assembly systems are able to assemble $n\times n$ squares in $O(\frac{\log n}{\log\log n})$ tile complexity, which meets a corresponding lower bound from Kolmogorov complexity that holds for almost all values $n$~\cite{Rothemund:2000:PSC}.  Further, this can be achieved with zig-zag tile systems.  So, why can't we get $O(\frac{\log n}{\log\log n})$ tile complexity for temperature $\tau =1$ assembly in 3D?  One answer is, we don't know that we can't.  However, the direct application of the zig-zag transformation to the temperature 2, $O(\frac{\log n}{\log\log n})$ tile complexity system does not work.  In particular, the temperature $\tau=2$ result is achieved by picking an optimal base for a counter, rather than counting in binary.  In general, a $k \times n$ rectangle can be assembled by a zig-zag tile system in $O(k + n^{\frac{1}{k}})$ tile complexity by implementing a $k$ digit, base $n^\frac{1}{k}$ counter~\cite{Aggarwal:2005:CGM}.  However, inherent in counting base $n^{\frac{1}{k}}$, the counter uses at least $n^{\frac{1}{k}}$ north/south glue types.  Thus, the zig-zag transformation yields a tile complexity of
$$
(k + n^{\frac{1}{k}})\log (n^{\frac{1}{k}}) = \Omega(\log n)
$$
 Therefore, regardless of what base we choose for the counter, we get at least a $\Omega(\log n)$ tile complexity.  The same problem exists for the base conversion scheme used in~\cite{Adleman:2001:RTP}.  It is an interesting open question whether or not it is possible to achieve the optimal $O(\frac{\log n}{\log\log n})$ tile complexity bound at temperature $\tau = 1$.

\section{Turing Machine Simulation for General Computation}
\label{sec:turing}  By extending non zig-zag Turing machines from the literature~\cite{Soloveichik:2005:CSA} we show how to implement a zig-zag tile set to simulate the behavior of an arbitrary input Turing machine.  From Theorem~\ref{thm:zigzag}, we can simulate this system at temperature 1.  The details of the construction are given in the Appendix in Section~\ref{sec:appturing}.

In this section we discuss our results for temperature 1 tile assembly systems capable of simulating Turing machines.  Our results are based on existing temperature 2 self-assembly systems capable of simulating Turing machines~\cite{Soloveichik:2005:CSA}.  With straightforward modifications, these constructions can be modified into zig-zag systems, and thus can be simulated at temperature 1 in 3D or in 2D probabilistically.  Our first Lemma states that such temperature 2 zig-zag systems exist.  For simplicity of bounds, we assume the Turing machine to be simulated is of a constant bounded size (constant bound number of states, alphabet size) and refer the reader to the appendix for tile complexities in terms of these variables.

\begin{lemma}\label{lemma:zigzagturing}  For a given Turing machine $T$ there exists a temperature 2 zig-zag tile assembly system that simulates T.  More precisely, there exists a zig-zag system that, given a seed assembly consisting of a horizontal line with north glues denoting an input string, the assembly will place a unique $\emph{accept}$ tile type if and only if T accepts the input string, and a unique $\emph{reject}$ tile type if and only if $T$ rejects the input string.
\end{lemma}

By combining this Lemma with the zig-zag simulation lemma, we get the following results.

\begin{theorem}\label{theorem:3dturing}  For a given Turing machine $T$ there exists a temperature 1, 3D tile assembly system that simulates T, even when assembly is limited to one step into the third dimension.  Thus, 3D temperature 1 self-assembly is universal.
\end{theorem}
\begin{proof}
This theorem follows from Lemma~\ref{lemma:zigzagturing} and Theorem~\ref{thm:zigzag}.  The exact tile complexity and scale of the simulation in terms of the number of states and alphabet size of the Turing machine are derived in the Appendix in Section~\ref{sec:appturing}.
\end{proof}

\begin{theorem}\label{theorem:2dturning}  For a given Turing machine $T$ there exists a temperature 1, 2D probabilistic tile assembly system that simulates T.  In particular, for any $\epsilon>0$ and machine $T$ that halts after $N$ steps, there exists a temperature 1, 2D system with $O(\log N + \log\frac{1}{\epsilon})$ tile types in terms of $N$ and $\epsilon$ that simulates $T$ and guarantees correctness with probability at least $1-\epsilon$.
\end{theorem}
This theorem follows from combining Lemma~\ref{lemma:zigzagturing} and the construction for zig-zag simulation used in Theorem~\ref{thm:zigzag2d} with parameter $k=\log\frac{N}{\epsilon}$.  As discussed in Section~\ref{sec:2dprobtemp1simulation} the probability of an single error in simulation of a given zig-zag set is $1/2^k$.  Thus, the expected number of errors for $k=\log\frac{N}{\epsilon}$ becomes $\frac{\epsilon}{N}$, implying that the probability of 1 ore more errors is bounded by at most $\epsilon$.

  The exact tile complexity and scale of the simulation in terms of the number of states and alphabet size of the Turing machine are derived in the Appendix in Section~\ref{sec:appturing}.

\section{Simulation Results}\label{sec:sim}
We have developed automatic conversion software to generate temperature 1 3D or 2D systems given any input zig-zag tile set.  We have further verified the correctness of our algorithms and converter software by implementing a number of tile systems such as Sierpinski tile sets and binary counters.
We have run the simulations on Caltech's Xgrow simulator\cite{xgrow} and ours, the system performs well.

\begin{figure}[h]\centering
 \includegraphics[width=\linewidth]{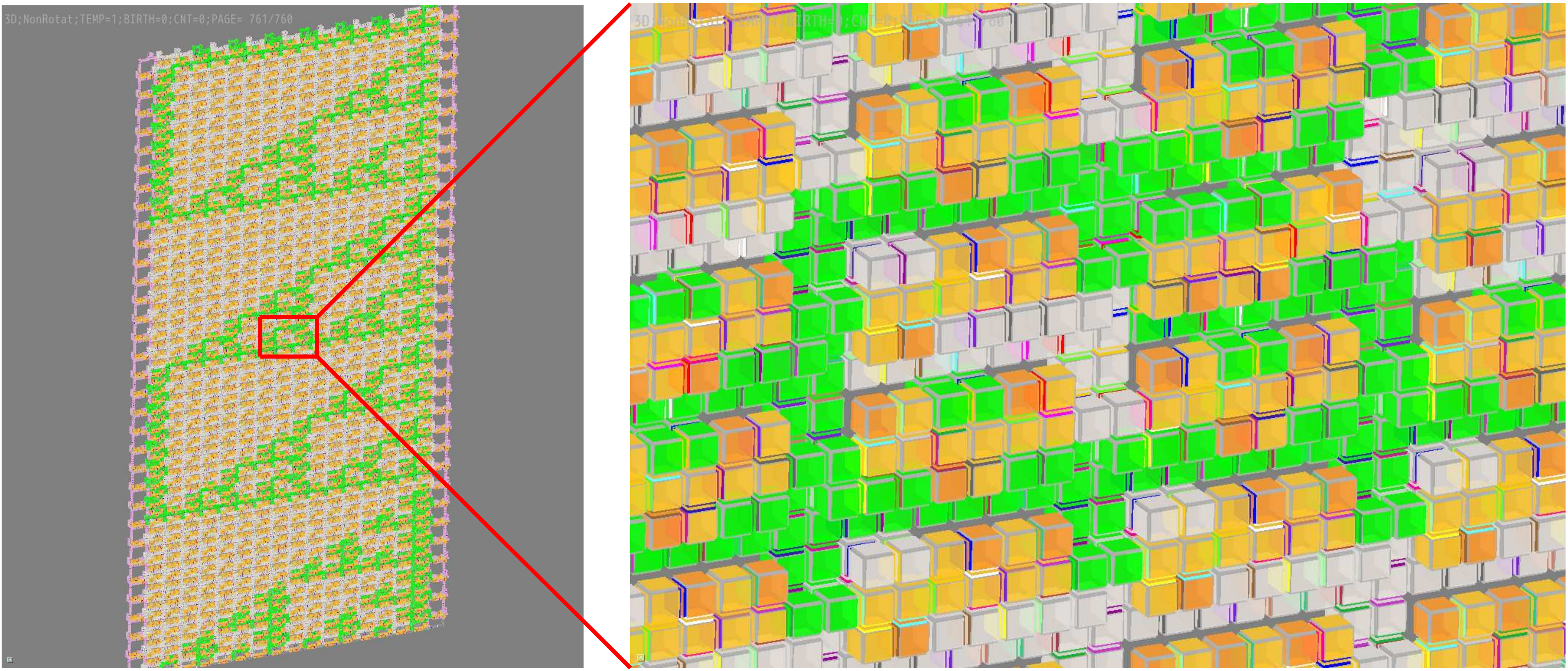}\\
 \caption{Snapshots of Sierpinski in 3D Zig-Zag Tile System}\label{fig:snapshot3d}
\end{figure}

We are interested in experimentally comparing the fault tolerance of
our temperature 1 probabilistic system with other fault tolerance schemes such as snaked proofreading systems.
We have implemented both systems with similar tile complexity constraints to see
under what settings either system would begin to make errors.
We have done preliminary simulations of the 2D systems in the kTAM with Caltech's Xgrow.
Our test case was a simple parity checking tile system\cite{Ashish04errorfree},
and our metric was whether or not the correct parity was computed.
Our preliminary tests used the input string "1000".
The block size of the snaked system is $ 6 \times 6 $,
and the length of each bit($K$) of probabilistic assembly tiles is 5.
The glue strength of the probabilistic assembly tiles are changed to 2.

We used $G_{mc}=13.6$, $G_{se}=7.0$. Both tile systems outputted correct results. But as we increased the $G_{se}$, the performances of the two systems began to differ. The snaked system became unstable, facet errors occurred and the tile mismatches increased. In contrast, the probabilistic assembly tile set grew faster and the output of the system was still correct. We used the $\tau=\frac{G_{mc}}{G_{se}}$ ranging from $2$ to $0.5$, and the results show that the 2D probabilistic assembly system is more stable than the snaked system within the parameters we tested.

The comparison of the probabilistic and snaked systems is very preliminary because the two system use different temperatures. We adjusted the glue strength of the probabilistic system to make it comparable with the snaked system. To more accurately compare the two systems, we need to investigate other methods and test sets, both for probabilistic and snaked systems.

\section{Further Research Directions}\label{sec:future}
There are a number of further directions relating to the work in this paper.  The most glaringly open problems are related to the power of deterministic temperature 1 self-assembly in 2D.  For example, it is conjectured that the optimal tile complexity for the assembly of an $n\times n$ square at temperature 1 in 2D is $2n-1$ tile types.  While it is straightforward to achieve this value, the highest lower bound known is no better than the small $O(\frac{\log n}{\log\log n})$ lower bound from Kolmogorov complexity.  Similarly, little is known about the computational power of this model.  It seems that it cannot perform sophisticated computation, but no one has yet been able to prove this.

Another direction is the concern that temperature 1 systems are prone to multiple nucleation errors in which many different growths of tiles spontaneously attach without requiring a seed.  One potential solution is to design growths such that assemblies that grow apart from the seed have a high probability of \emph{self-destructing} by yielding growths that block all further growth paths for the assembly.  In contrast, properly seeded growths by design could have the self-destruct branches blocked by the seed base they are attached to.

\newpage

\bibliographystyle{plain}

\bibliography{SA}

\input{appendix}

\end{document}

%% file: appendix.tex
\SetAlFnt{\small\sf}


\appendix
\section{A Simple 3D Example}
See Figure \ref{fig:simpleexample} and \ref{fig:simpleexampleassembled}.
%
\begin{figure}[htbp]
\centering
\includegraphics[width=\linewidth]{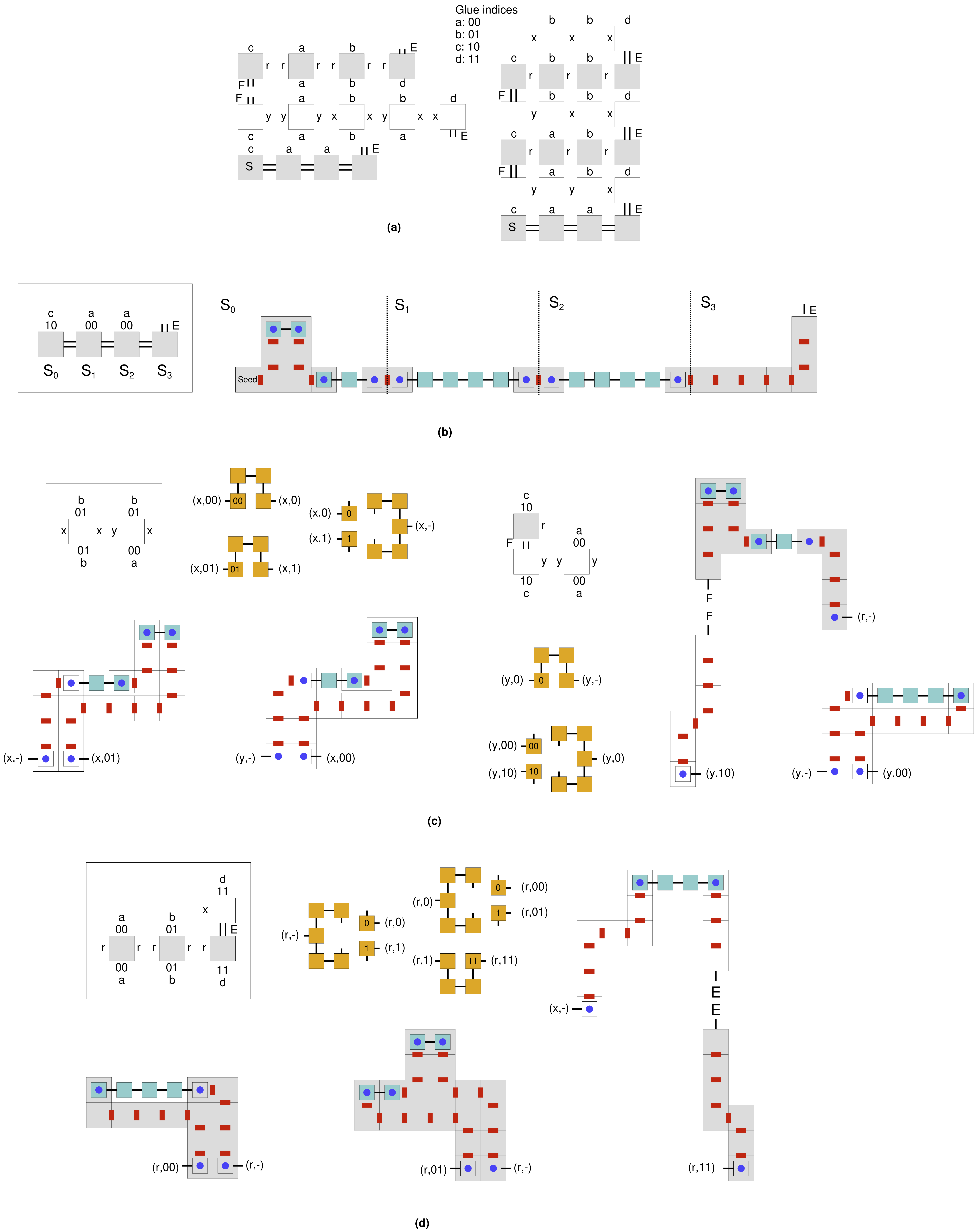}
\caption{The tileset in (a) is a zig-zag tileset that swaps the right most a to b every other line.  In (b), (c), and (d), a collection of macro tiles are provided, according to Theorem~\ref{thm:zigzag}.  This set will simulate the given tileset system using 3 dimensions at temperature $\tau =1$.}
\label{fig:simpleexample}
\end{figure}
%

%
\begin{figure}[htbp]
\centering
\includegraphics[width=0.9\linewidth]{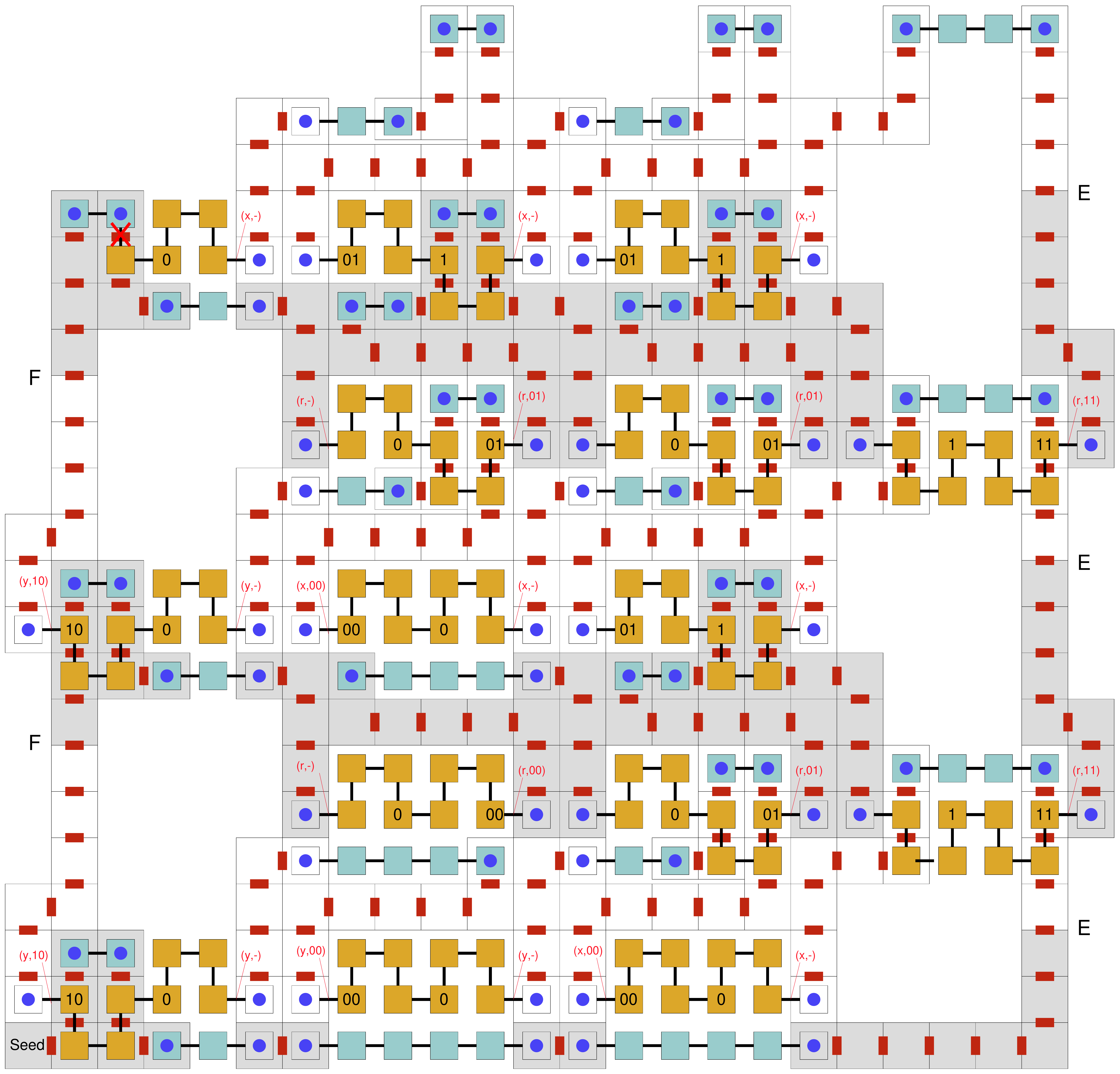}
\caption{This is the final assembly from the tileset given in Figure~\ref{fig:simpleexample}}.
\label{fig:simpleexampleassembled}
\end{figure}
%

\clearpage

\section{Constructing a Turing Machine with Zig-Zag Tile Sets at $\tau=2$ in 2D}\label{sec:appturing}

\subsection{Converting}

A {\bf \emph{Turing machine}} is a 7-tuple\cite{Sipser05introductionto}, $(Q, \Sigma , \Gamma , \delta , q_0, q_{accept}, q_{reject})$, where
$ Q, \Sigma , \Gamma $ are all finite sets and
\begin{enumerate}
  \item $Q$ is the set of states,
  \item $ \Sigma $ is the input alphabet not containing the {\bf \emph{blank symbol}} $ \textvisiblespace $,
  \item $ \Gamma $ is the tape alphabet, where $\textvisiblespace \in \Gamma $ and $\Sigma \subseteq \Gamma $,
  \item $ \delta : Q \times \Gamma \longrightarrow Q \times \Gamma \times \{L, R\}$ is the transition function,
  \item $ q_0 \in Q $ is the start state,
  \item $ q_{accept} \in Q $ is the accept state, and
  \item $ q_{reject} \in Q $ is the reject state, where $q_{accept} \neq q_{reject} $.
\end{enumerate}

All of the glues at the side north or south of the tile regard as the alphabet in set $ \Gamma $. The $g_n$ of tile is the output and $g_s$ is the input in our turing machine, so the action of placing one tile with difference $g_n$ and $g_s$ on the supertile regard as one write operation, placing one tile with the same $g_n$ and $g_s$ regard as one read operation.

Basically, the tiles on the west of growth direction will be used to the main calculation, and the tiles on the direcition of east will be used to copy the glues from south to north.

\begin{figure}[h]\centering
 \includegraphics[scale=0.9]{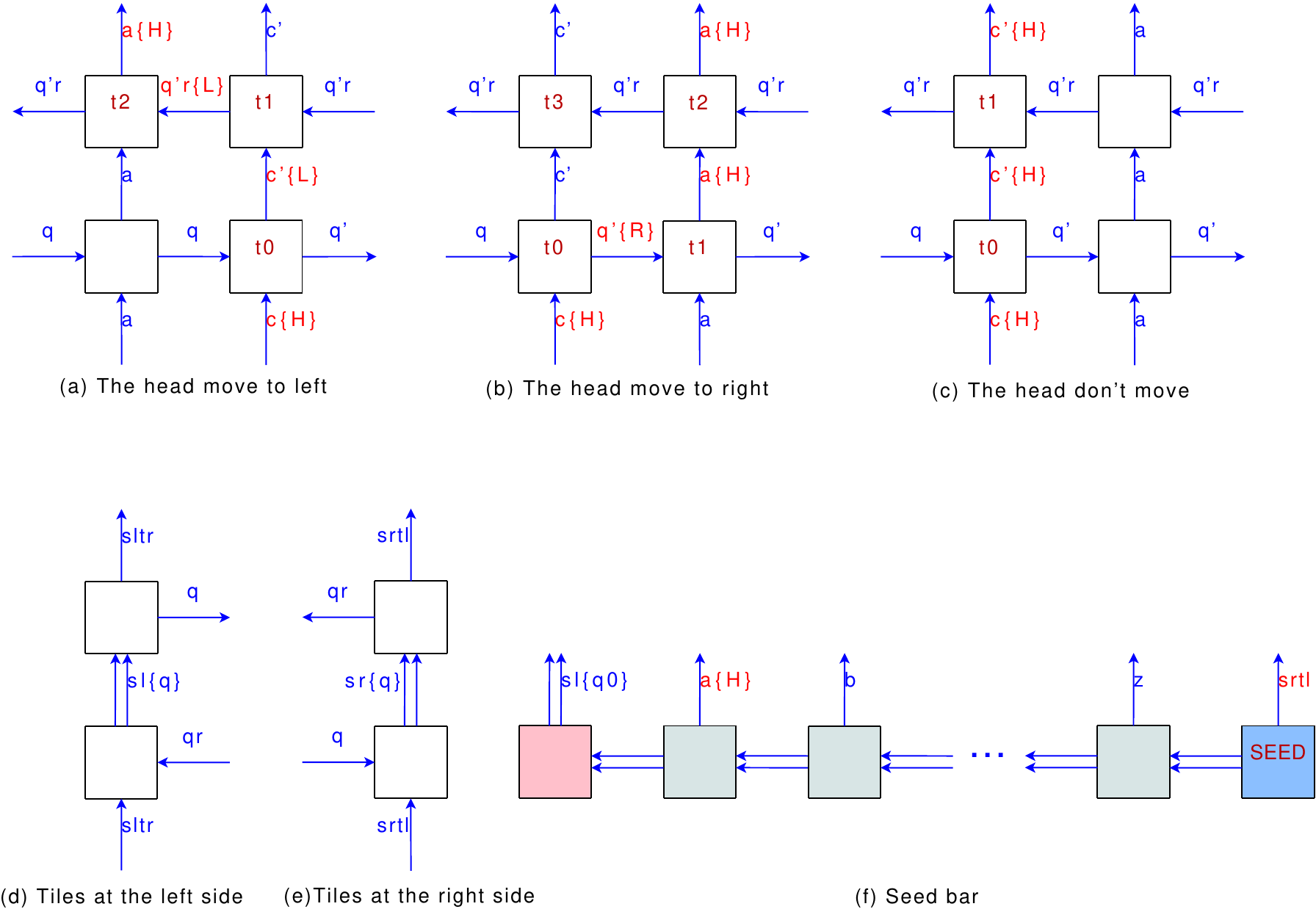}
 \caption{{\bf Zig-Zag Tile Structure for Contructing Turing Machine}. {\bf (a)} $ \delta (q, c) = (q', c', L); ${\bf (b)} $ \delta (q, c) = (q', c', R) $; {\bf (c)} $ \delta (q, c) = (q', c', \phi) $; {\bf (d)} Left side; {\bf (e)} Right side; {\bf (f)} Seed bar: $ q_0 $ and input sequence $a, b, \dots , z$}\label{fig:zz2tm}
\end{figure}

First, let's consider the equation $ \delta (q, c) = (q', c', R) $. The head of TM is located at the $c$ position and the current state is $q$. Then the record of the current position of tape is changed to $c'$, the current state is changed to $q'$, and head move to right (R). Three types of the tiles, which are indicated by $t_0$, $t_1$, and $t_2$ in Figure \ref{fig:zz2tm}b, are used to simulate this process. The glue notation at the south of the tile $t_0$ means that the head is at the column of tile $t_0$ and input alphabet is $c$. The glue notation ($q$) at the west of the tile $t_0$ means that the current state is $q$. The glues at the north and east indicate the output and next state seperately. The east glue ($q'\{R\}$) of the tile $t_0$ indicates the next state ($q'$) and carring the extra information to notify the next tile that the head is coming.

The equation $ \delta (q, c) = (q', c', L) $ indicate that the head will move to left. Figure \ref{fig:zz2tm}a shows how to move the head to left. Again, the head is at the column of tile $t_0$, the current state is $q$, and input is $c$. Then the tile $t_0$ output north glue $c'\{L\}$, and east glue $q'$. The tile $t_1$ will translate the glue information ($c'\{L\}$) to state information ($q'r\{L\}$) which contains the command of moving left. At last, the tile $t_2$ add the head notation to the glue of the current column.

If the head don't move after read one character from tape, the corresponding tile is $t_0$ showed in Figure \ref{fig:zz2tm}c.
Figure \ref{fig:zz2tm}d and Figure \ref{fig:zz2tm}e show the tiles at the left and right sides of zig-zag structure seperatly.
Figure \ref{fig:zz2tm}f shows the seed bar, which initializes the turing machine, includes the start state $q_0$, the input sequence ($a, b, \dots , z$) in the type etc.


\subsection{Complexity Analysis}
In the case of $\tau=2$ 2D zig-zag tile system, the number of the tile types for state transfer ($t_0$ in the Figure \ref{fig:zz2tm}) is $ | \delta | $. And the number of the auxiliary tile types for state transfer ($t_1$ and $t_2$ in Figure \ref{fig:zz2tm}a, \ref{fig:zz2tm}b ) is $|\Sigma|\times |Q|\times 2$ (move left) and $|\Sigma|\times |Q|\times 2$ (move right). The number of tile $t_1$ in Figure \ref{fig:zz2tm}c does not need be included because it is counted in the auxiliary tile types ($t_2$ in Figure \ref{fig:zz2tm}b ).
The other tile types are the tiles for transfer alphabet information, the number of it is $|Q|\times |\Sigma|\times 2$ (2 indicate the east direction and west direction of the growth).

There's also need some special types of tile in both sides of the zig-zag structure to change the growth direction, it will cost $4|Q|$ tiles. If the length of the tape is $n$, then the number of the tile types for constructing seed bar will cost $n$ + 2 tiles (See Figure \ref{fig:zz2tm}f).

So the total tile types is
$T=| \delta | + |\Sigma|\times |Q|\times 2 + |\Sigma| \times |Q|\times 2 + |Q|\times |\Sigma|\times 2 + 4|Q| + n + 2
=| \delta | + 6 |Q| |\Sigma| + 4|Q| + n + 2$. The upper bound of tile types is $O(| \delta | + |Q| |\Sigma| + n) $. Considering that $|Q| \leq | \delta | \leq |Q|^2 | \Sigma |^2$ and $n \leq | \Sigma |$, the complexity of tile types is $O(| \delta | + |Q| |\Sigma| + n) = O(| \delta |)$.

Each of the state transfer cost two lines of the zig-zag structure (east direction line and west direction line above). Given a input tape and start state, if the number of the state transfers to be passed before it stop is $r$, then the lines of the tile structure will be $2r$. So the space used in calculation is $O(nr)$.

In the case of $\tau=1$ 3D zig-zag tile system, the length of each mapped tile depend on the length of encoding code of glues. If the number of glues of the 2D zig-zag tiles is $G$, then $G=O(| \delta |)$. The complexity of tile types in 3D is $O(| \delta | \log G)=O(| \delta | \log | \delta |)$, the complexity of space is $O(nr\log G)=O(nr \log | \delta |)$.

In the case of $\tau=1$ 2D probabilistic zig-zag tile system, the length of each mapped tile depend on the parameter $K$ and the length encoding code of glues.
The parameter $K$ of the $\tau=1$ 2D probabilistic zig-zag tile system, is a constant depend on the probability of correction of result. $K=O(\log r + \log \frac{1}{\epsilon})$.
The space complexity is $O(nrK\log G) = O(nr\log | \delta |)$, and the complexity of tile types is $O(| \delta | \log | \delta |)$.

\begin{table}[h]\caption{The Complexity of Zig-Zag Turing Machine}
\centering
\label{tab:zigzagcomplexity}
 \begin{tabular}{|c|c|c|c|}
 \hline
           &{\bf $\tau=2$ 2D}   &{\bf $\tau=1$ 3D}  &{\bf $\tau=1$ 2D Prob.}\\
 \hline
{\bf Space}&$O(nr)$             &$O(nr\log | \delta |)$      &$O(nr\log | \delta |)$\\
 \hline
{\bf Tiles}&$O(| \delta |) $&$O(| \delta |\log | \delta |)$&$O(| \delta |\log | \delta |)$\\
 \hline
\end{tabular}\\
$| \delta |$: the number of state functions;\\
$n$: the length of the tape (number of alphabet);\\
$r$: the number of the state transfers to be passed before it stop;\\
$K$: the parameter of probabilistic zig-zag tile system;\\

\end{table}

\subsection{Algorithms}
The detailed algorithms for transfering any turing machine to zig-zag tile set at temperature $\tau=2$ in 2D is listed in Algorithm \ref{alg:tmstate2zigzag}.

\newcommand{\tmstatetozigzag}{\ensuremath{\mbox{\sc Turing-machine-to-zigzag}}}

\begin{algorithm}
\caption{\tmstatetozigzag(): Convert TM to Zig-Zag Tile Set at $\tau=2$ in 2D}
\label{alg:tmstate2zigzag}
\linesnumberedhidden
\SetKwData{tileout}{$T_{2d,2t}$}

\KwIn{$q_0$, the start state\\
$str$, the input string at the tape\\
$ \delta $, the state transition functions
}

\KwOut{\tileout, zig-zag tile set at temperature 2 in 2D}

    $Q \gets \phi $ \tcc*[r]{All of the states of the $ \delta $.}
    $\Sigma \gets \phi $ \tcc*[r]{All of the input/output of the $ \delta $ and the alphabet at the tape.}

    \ForEach{ $ f \in \delta $ }{
        $Q \gets Q \cup \{$ current state of $f\}$ \;
        $Q \gets Q \cup \{$ next state of $f\}$ \;

        $\Sigma \gets \Sigma \cup \{$ input of $f\}$ \;
        $\Sigma \gets \Sigma \cup \{$ output of $f\}$ \;
    }
    \ForEach{ $ c \in str $ }{
        $\Sigma \gets \Sigma \cup \{c\}$ \;
    }
    \ForEach{ $ q \in Q $ }{
        \tcc{A four sided Wang tile denoted by the ordered quadruple $(\textrm{north}(t), \textrm{east}(t), \textrm{south}(t), \textrm{west}(t))$}
        $\tileout \gets \tileout \cup \{tile (sl\{q\}, q_r, sltr, \phi)\}$ \tcc*[r]{left, down}
        $\tileout \gets \tileout \cup \{tile (sltr, q, sl\{q\}, \phi)\}$ \tcc*[r]{left, up}
        $\tileout \gets \tileout \cup \{tile (sr\{q\}, \phi , srtl, q)\}$ \tcc*[r]{right, down}
        $\tileout \gets \tileout \cup \{tile (srtl, \phi , sr\{q\}, q_r)\}$ \tcc*[r]{right, up}
        \ForEach{ $ c \in \Sigma $ }{
            $\tileout \gets \tileout \cup \{tile (c, q, c, q)\}$ \tcc*[r]{tiles for coping}
            $\tileout \gets \tileout \cup \{tile (c, q_r, c, q_r)\}$ \tcc*[r]{tiles for coping, return to left sides.}
            $\tileout \gets \tileout \cup \{tile (c\{H\}, q_r, c\{H\}, q_r)\}$ \tcc*[r]{Right moving auxiliary tiles: $t_2$}
            $\tileout \gets \tileout \cup \{tile (c\{H\}, q, c, q\{R\})\}$ \tcc*[r]{Right moving auxiliary tiles: $t_1$}
            $\tileout \gets \tileout \cup \{tile (c\{H\}, q_r\{L\}, c, q_r)\}$ \tcc*[r]{Left moving auxiliary tiles: $t_2$}
        }
    }
    \ForEach{ $ f \in \delta $ }{
        \Switch {head moving}{
            \showln \Case {LEFT}{
                $\tileout \gets \tileout \cup \{tile (f.output\{L\}, f.state_{out}, f.input\{H\}, f.state_{in})\}$ \tcc*[r]{$t_0$}
                $\tileout \gets \tileout \cup \{tile (f.output, f.state_{out}r, f.input\{L\}, f.state_{out}r\{L\})\}$ \tcc*[r]{$t_1$}
            }
            \showln \Case {RIGHT}{
                $\tileout \gets \tileout \cup \{tile (f.output, f.state_{out}\{R\}, f.input\{H\}, f.state_{in})\}$ \tcc*[r]{$t_0$}
            }
            \showln \Case {No moving}{
                $\tileout \gets \tileout \cup \{tile (f.output\{H\}, f.state_{out}, f.input\{H\}, f.state_{in})\}$ \tcc*[r]{$t_0$}
            }
        }
    }
    \tcc{Seed Bar}
    $\tileout \gets \tileout \cup \{tile (srtl, \phi , \phi , gr)\}$ \tcc*[r]{SEED}
    $\tileout \gets \tileout \cup \{tile (sl\{q_0\}, gl , \phi , \phi)\}$ \tcc*[r]{SEED\_L}
    \For{$i \gets 0 ~ \KwTo ~ (\vert str \vert - 1)$}{
        $c \gets i$th alphabet of str \;
        \If{$ i = 0 $}{
            $\tileout \gets \tileout \cup \{tile (c\{H\}, s_i , \phi , gl)\}$ \;
        }\ElseIf{$ i = (\vert str \vert - 1) $}{
            $\tileout \gets \tileout \cup \{tile (c, gr , \phi , s_{i-1})\}$ \;
        }\Else{
            $\tileout \gets \tileout \cup \{tile (c, s_i , \phi , s_{i-1})\}$ \;
        }
    }

    \KwRet{$\tileout$}\;
\end{algorithm}


\section{Convert Arbitrary Zig-Zag Tile Set from $\tau=2$ to $\tau=1$}\label{sec:detailconvertzigzag}

\subsection{Notation}

A tile $t$ is a four sided Wang tile denoted by the quadruple $(g_n, g_e, g_s, g_w)$, $g_n, g_e, g_s, g_w$
denote the glue type of the four sides: North, East, South, and West.

\subsection{The Converting Table}
When the Zig-Zag tiles in 2D are categorized into some types, the tile types can be mapped to 3D
and 2D probabilistic tile sets directly.
Table \ref{tab:zigzagmapping} list all of the relationships between the 2D tiles and 3D tile sets.

The tiles (see the 3D figures in Table \ref{tab:zigzagmapping}) are {\bf \emph{different}} from each other.
The lines between two adjacent tiles denote the glues which strength are 1
and they are also {\bf \emph{different}} from each other except adjacent glues.
Large squares denote the tiles in the plane $z=0$, and small squares denote the tiles in the plane $z=1$.
The encoded code($e_n$) showed in the figures in Table \ref{tab:zigzagmapping} is two bits long(maxbits=2).
The parameter $K$ of Probabilistic Zig-Zag is 2 in Table \ref{tab:zigzagmapping}.

\begin{longtable}[h]{>{\centering\arraybackslash}m{0.26\textwidth}m{0.28\textwidth}m{0.36\textwidth}}
\caption{Zig-Zag Tile Set Mapping} \label{tab:zigzagmapping}\\
\toprule
{\bf Zig-Zag in 2D $\tau=2$} & {\bf Zig-Zag in 3D $\tau=1$} & {\bf Prob. Zig-Zag in 2D $\tau=1$}\\
\midrule
\endfirsthead  
\midrule
{\bf Zig-Zag in 2D $\tau=2$} & {\bf Zig-Zag in 3D $\tau=1$} & {\bf Prob. Zig-Zag in 2D $\tau=1$}\\
\midrule
\endhead 
\multicolumn{3}{r}{Continued on next page \dots} \\
\midrule

\endfoot 
\endlastfoot

\begin{minipage}[b]{0.26\textwidth}\centering
\includegraphics[scale=0.22]{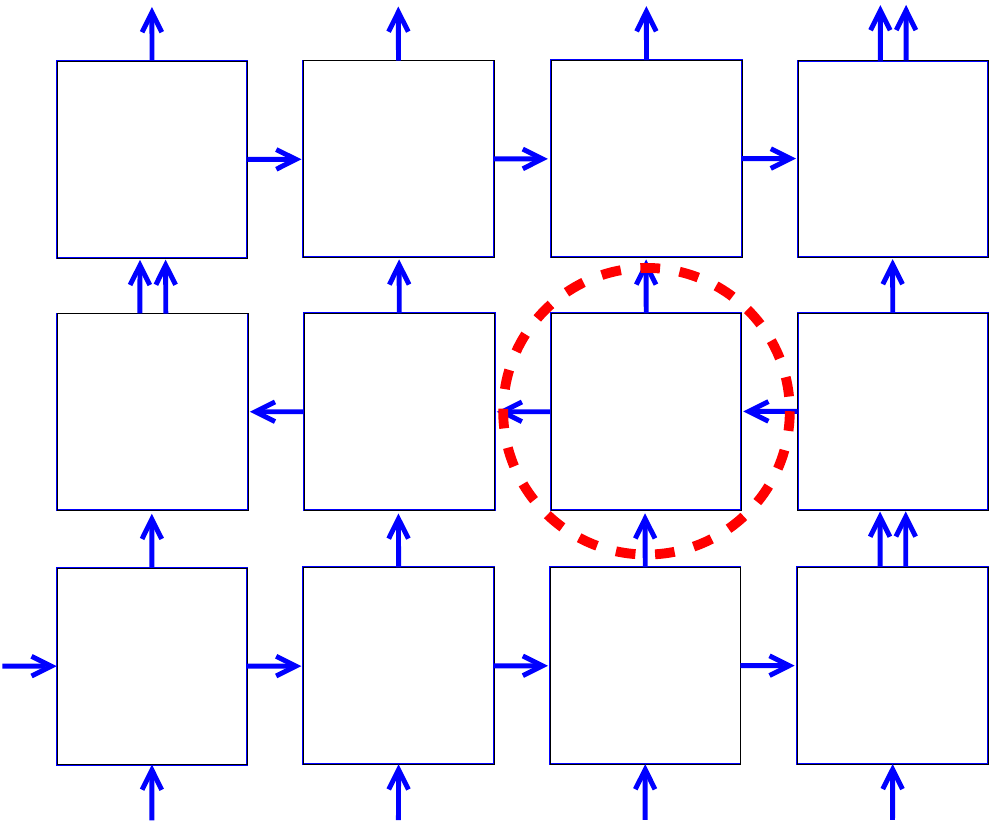}
\includegraphics[scale=0.46]{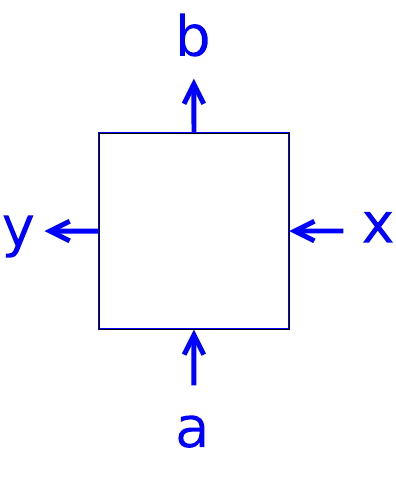}
\small{Direction West, West 1 (DWW1)}
\end{minipage}&
\includegraphics[scale=0.23]{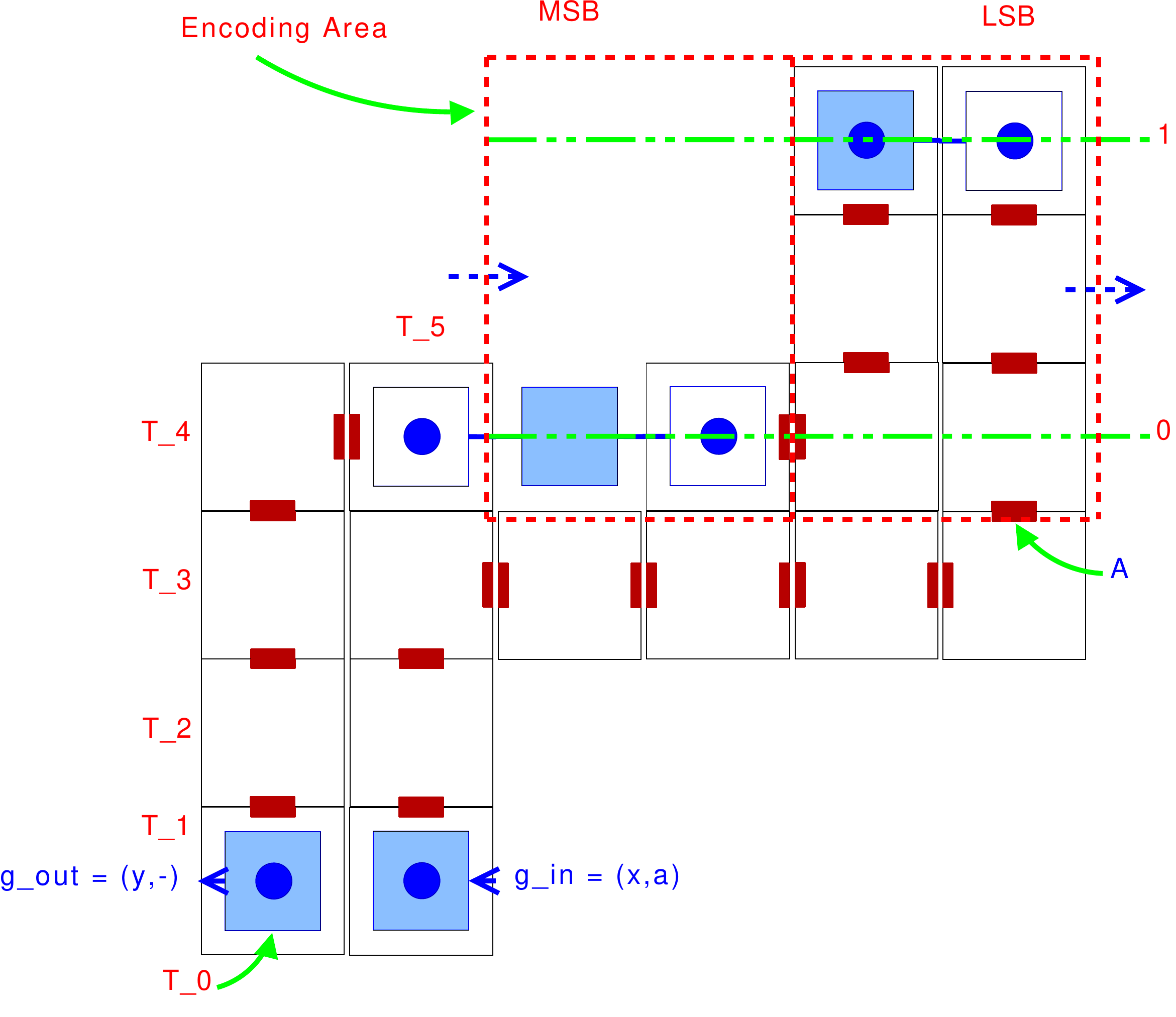}&
\includegraphics[scale=0.23]{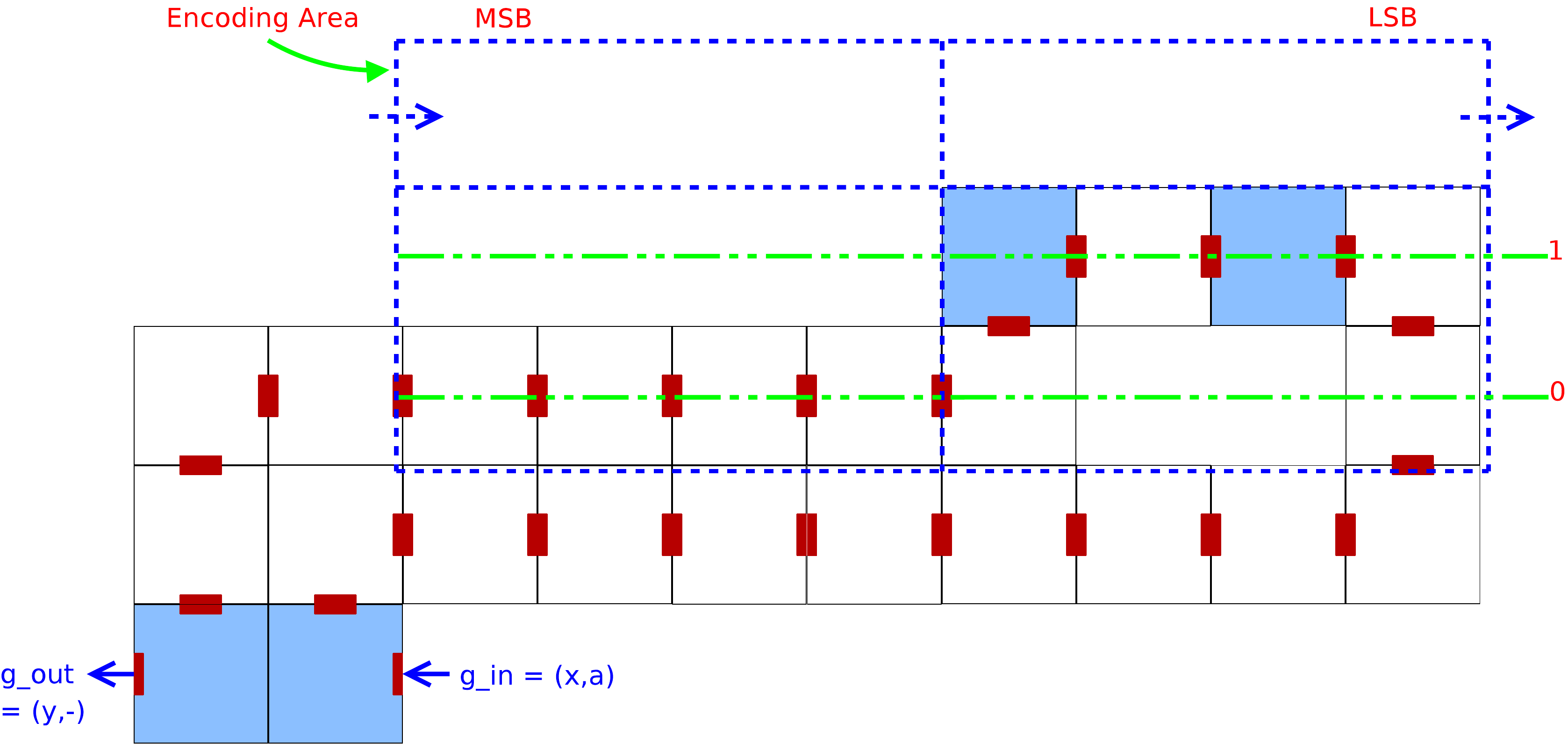}\\

\midrule
\begin{minipage}[b]{0.28\textwidth}\centering
\includegraphics[scale=0.22]{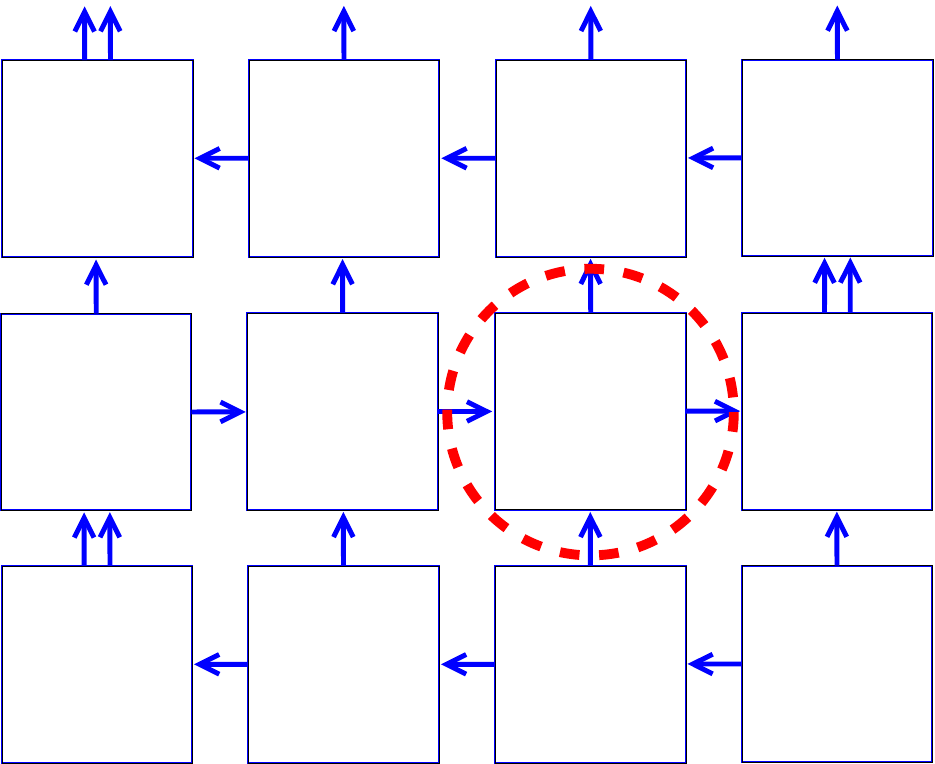}
\includegraphics[scale=0.46]{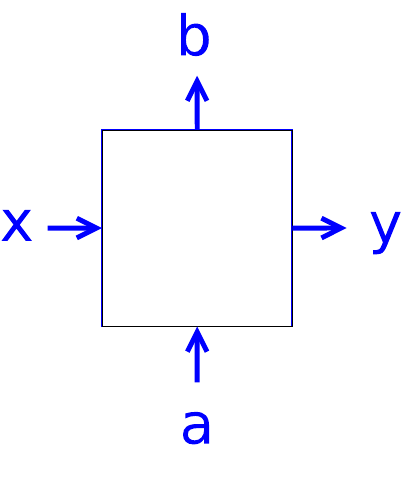}
\hspace{0.04\textwidth}
\small{Direction East, East 1 (DEE1)}
\end{minipage}&
\includegraphics[scale=0.23]{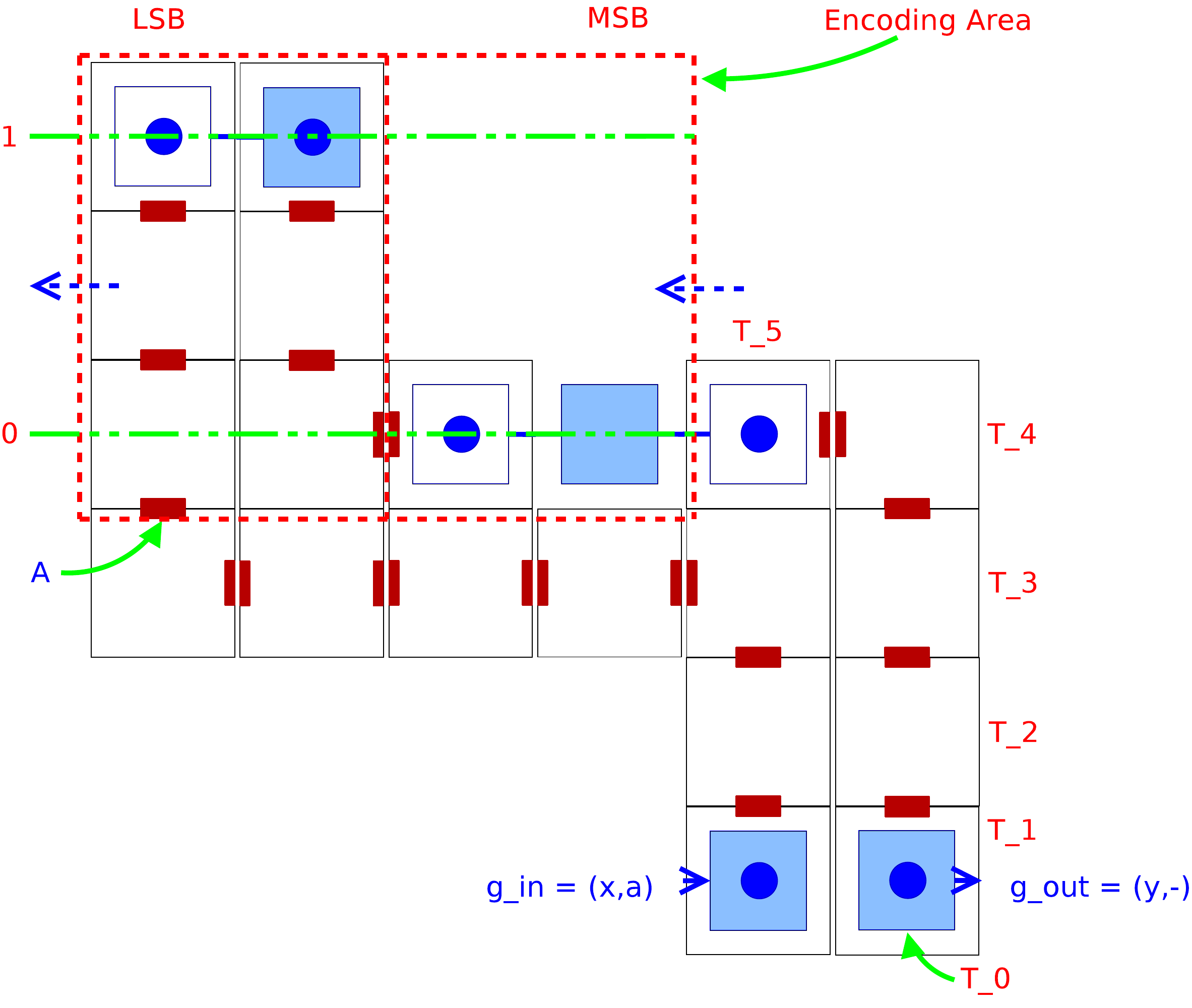}&
\includegraphics[scale=0.23]{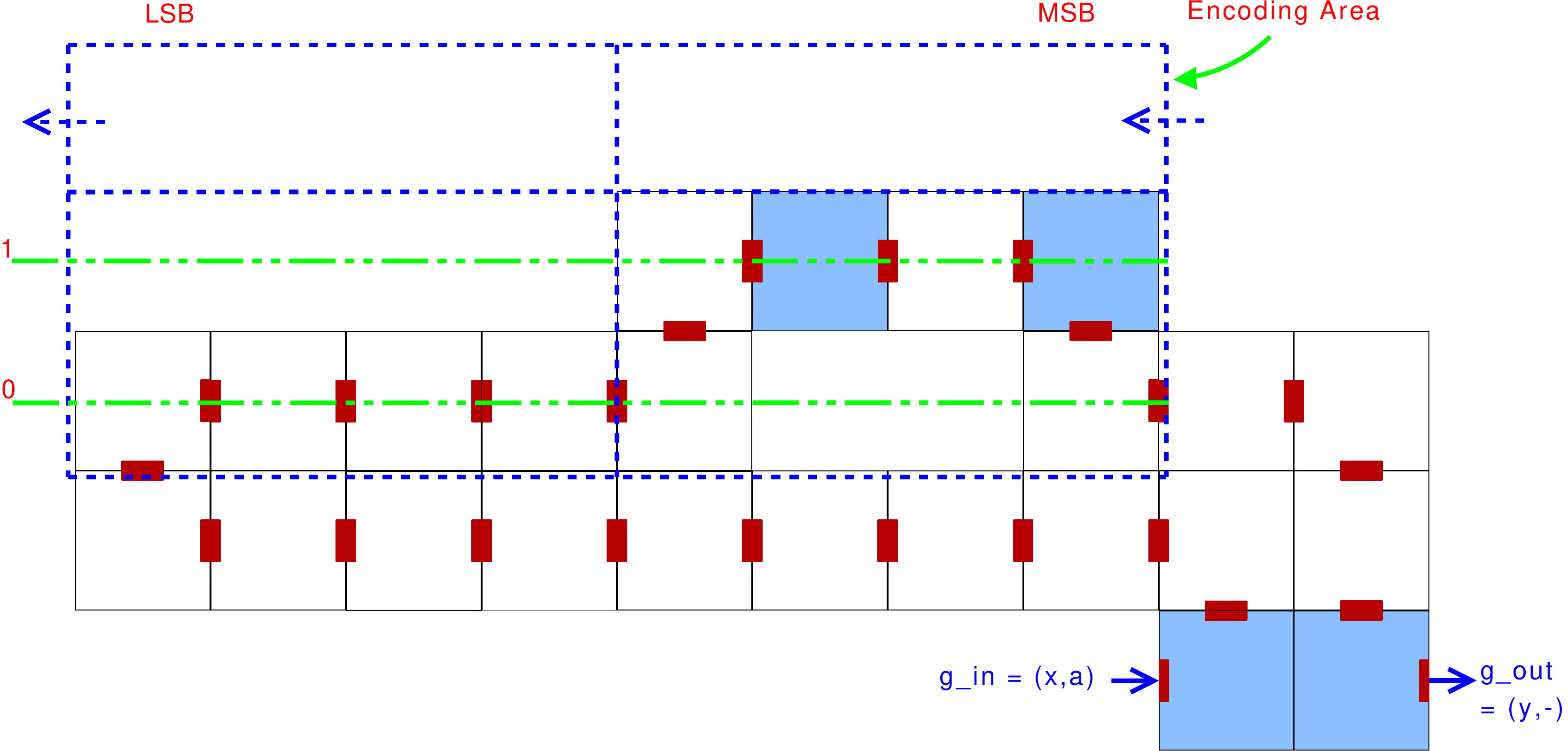}\\

\midrule
\begin{minipage}[b]{0.28\textwidth}\centering
\includegraphics[scale=0.22]{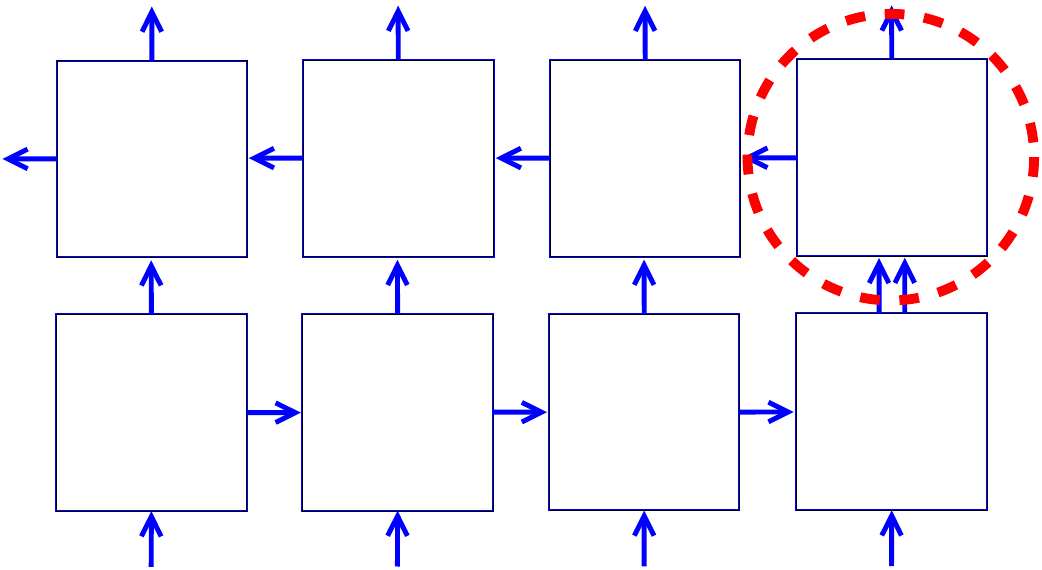}
\includegraphics[scale=0.46]{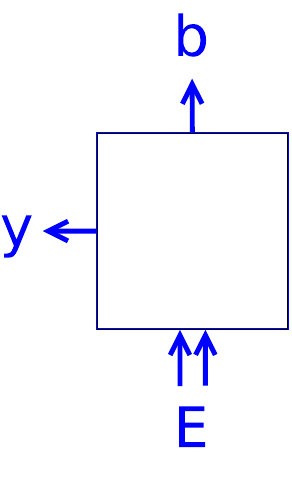}
\hspace{0.04\textwidth}
\small{Turn at East, West 1 (TEW1)}
\end{minipage}&
\includegraphics[scale=0.23]{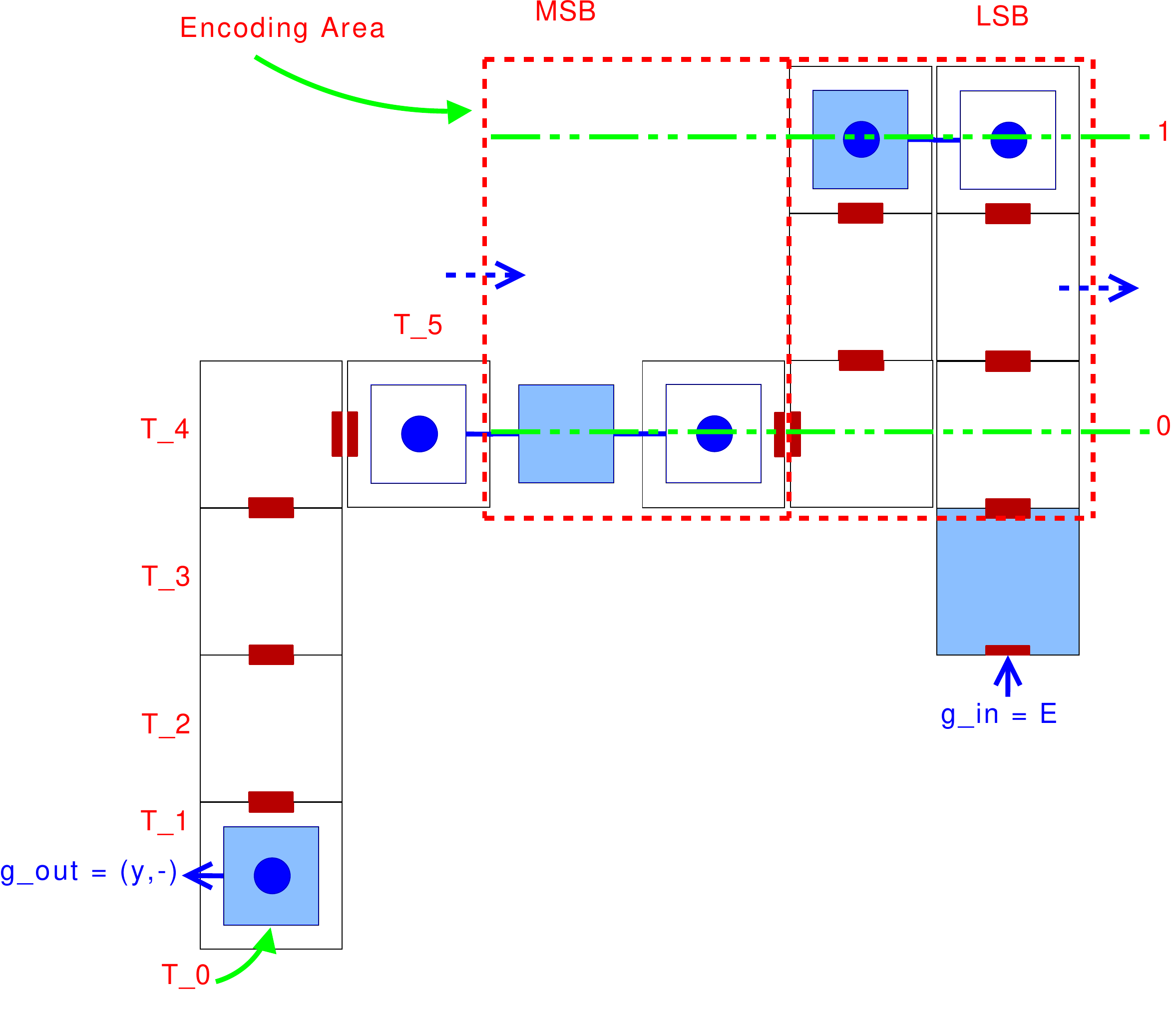}&
\includegraphics[scale=0.23]{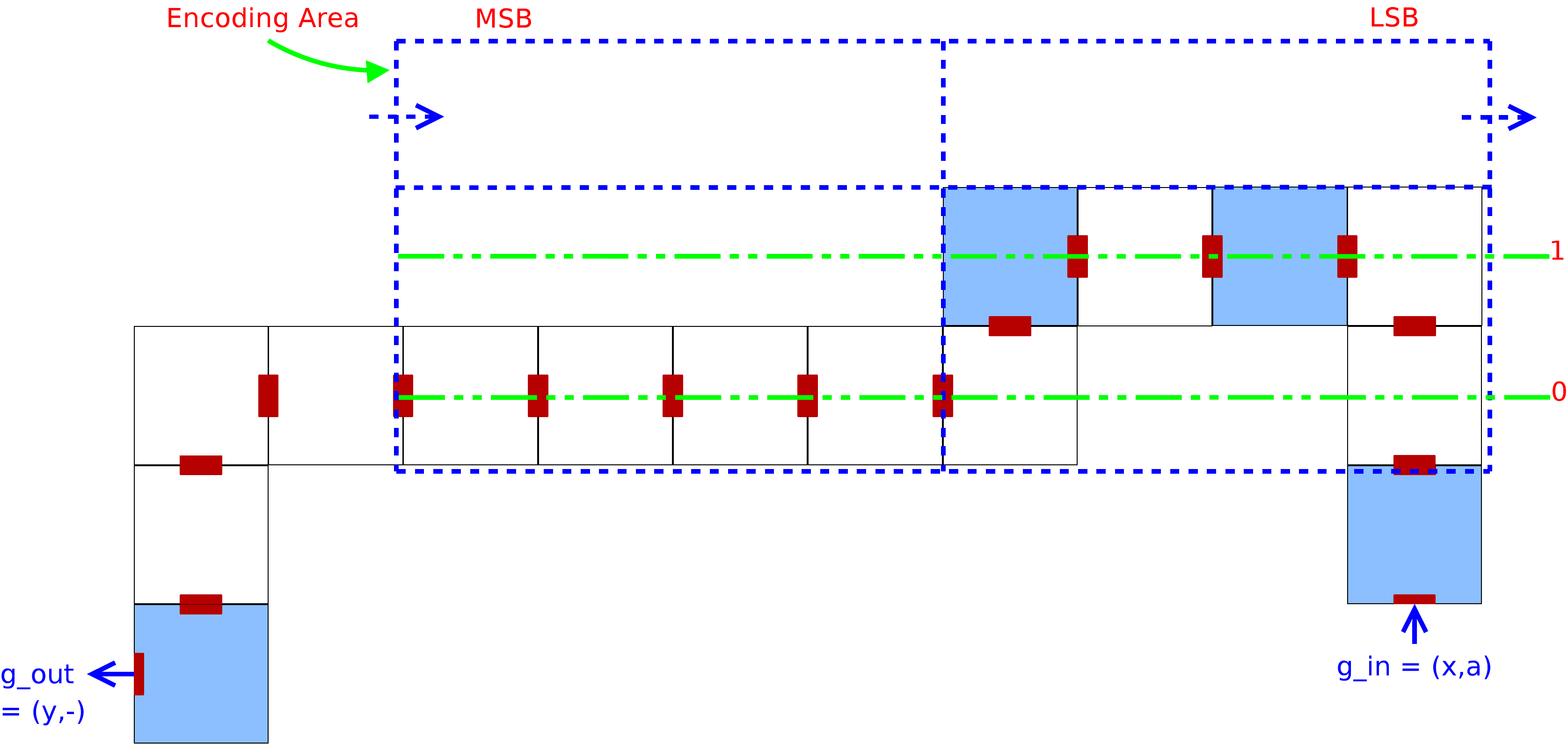}\\

\midrule
\begin{minipage}[b]{0.28\textwidth}\centering
\includegraphics[scale=0.22]{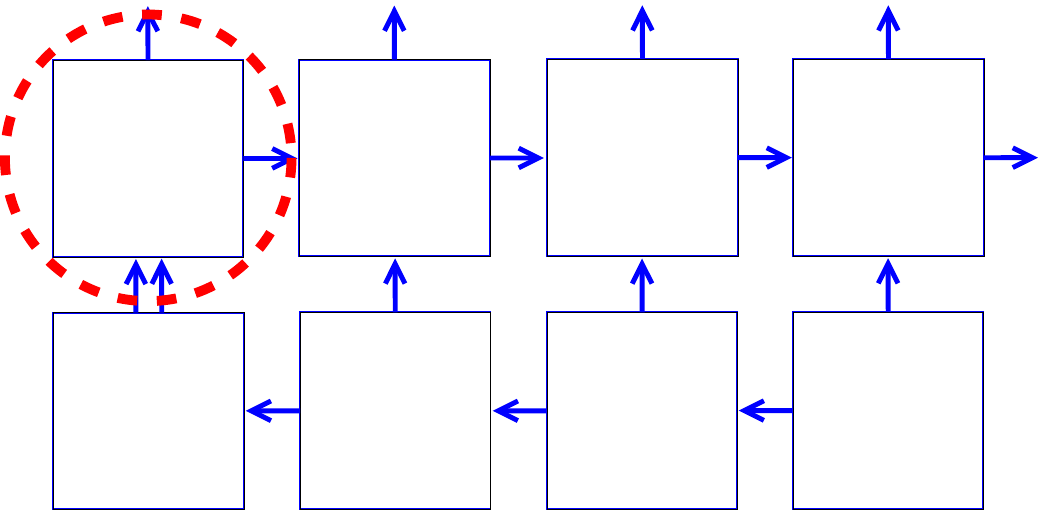}
\includegraphics[scale=0.46]{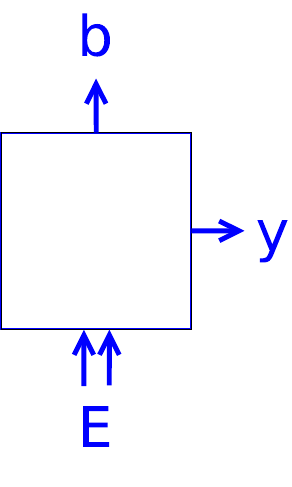}
\hspace{0.04\textwidth}
\small{Turn at West, East 1 (TWE1)}
\end{minipage}&
\includegraphics[scale=0.23]{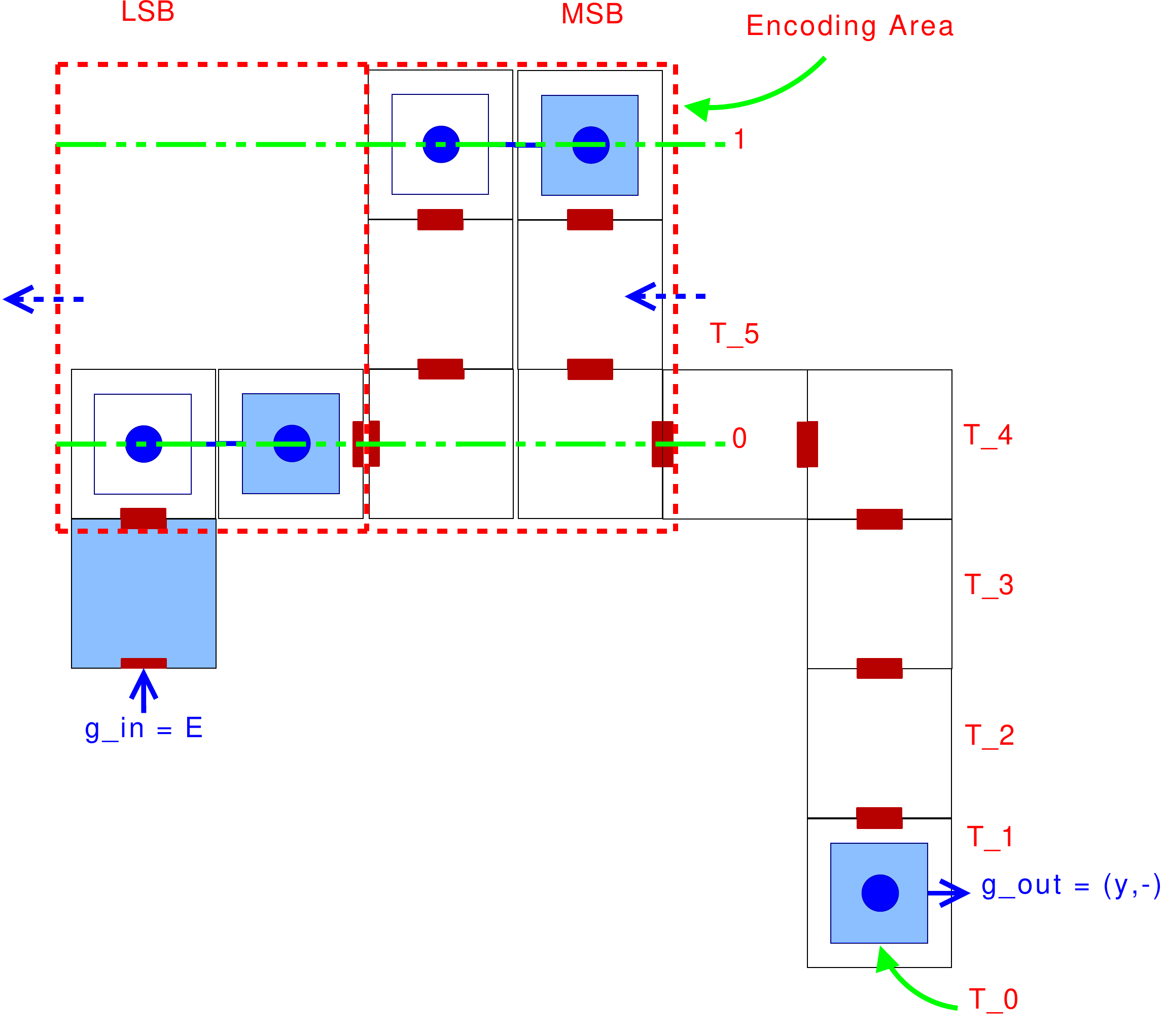}&
\includegraphics[scale=0.23]{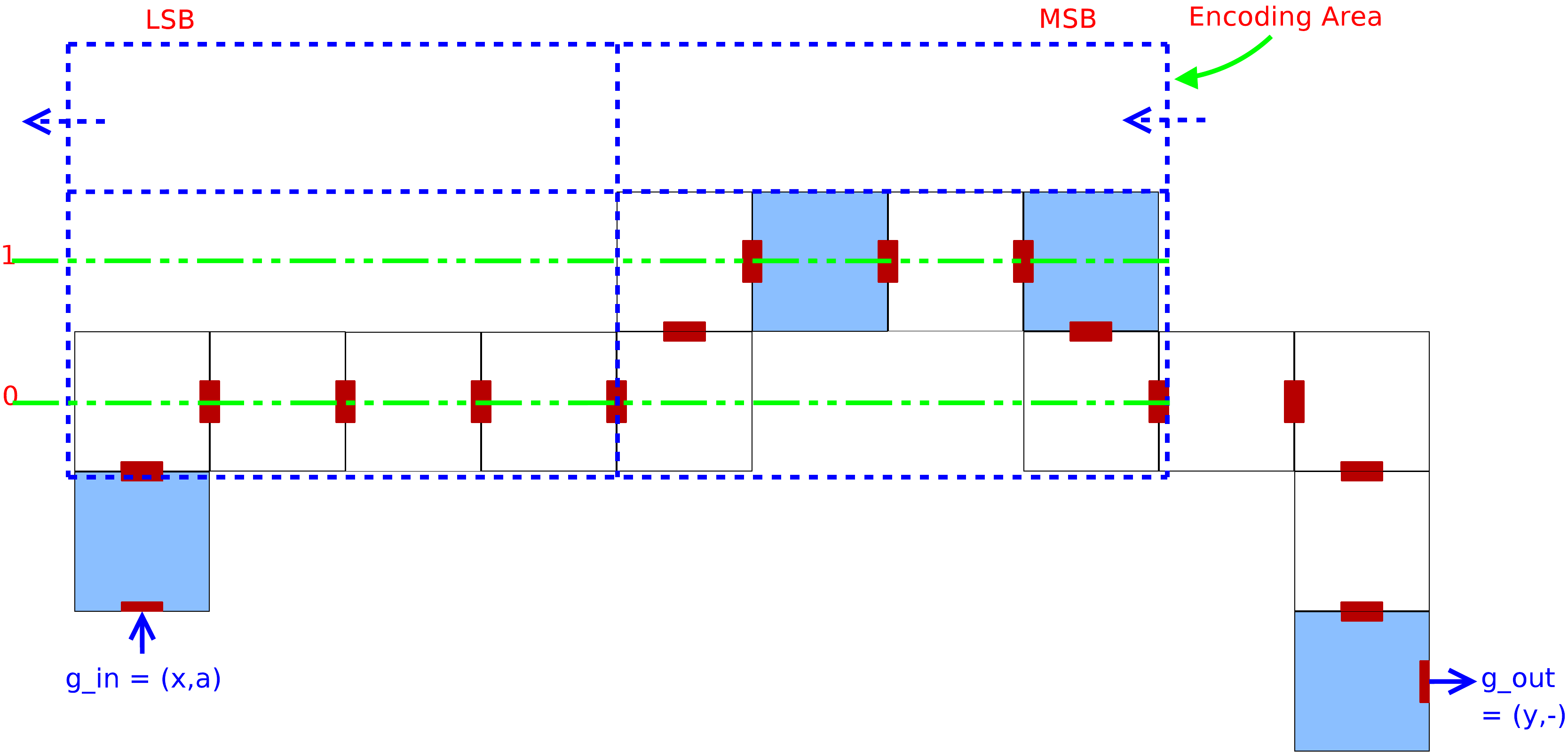}\\

\midrule
\begin{minipage}[b]{0.28\textwidth}\centering
\includegraphics[scale=0.22]{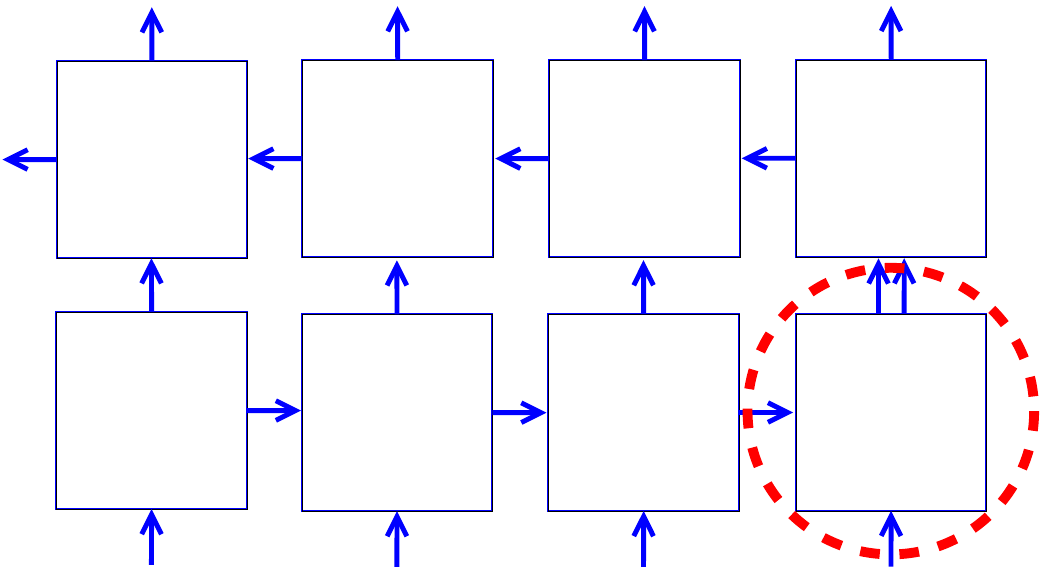}
\includegraphics[scale=0.46]{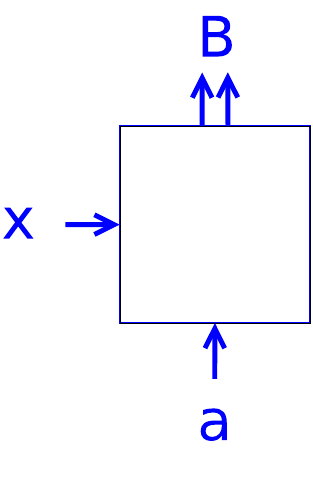}
\hspace{0.04\textwidth}
\small{Direction East, North 2 (DEN2)}
\end{minipage}&
\includegraphics[scale=0.23]{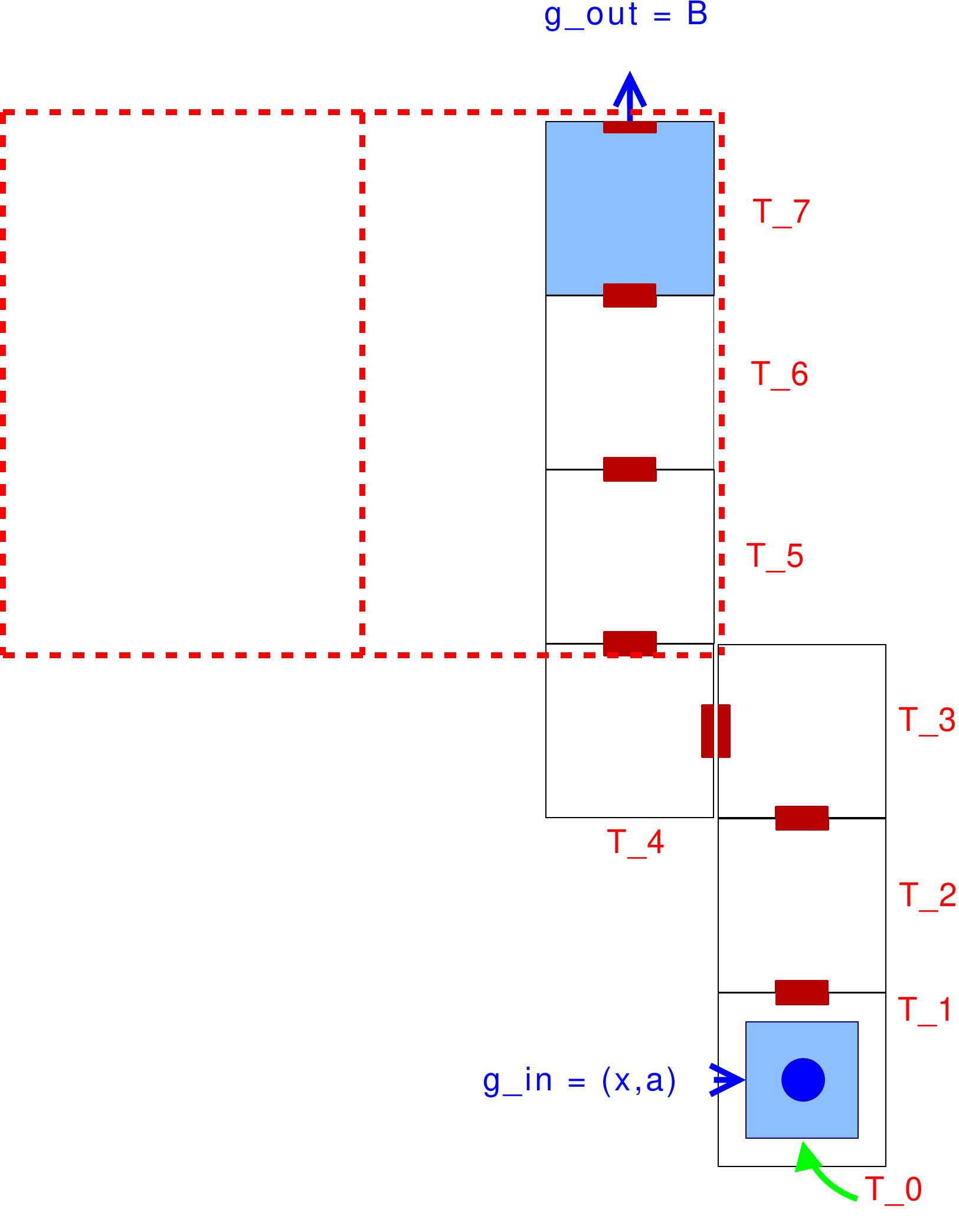}&
\includegraphics[scale=0.23]{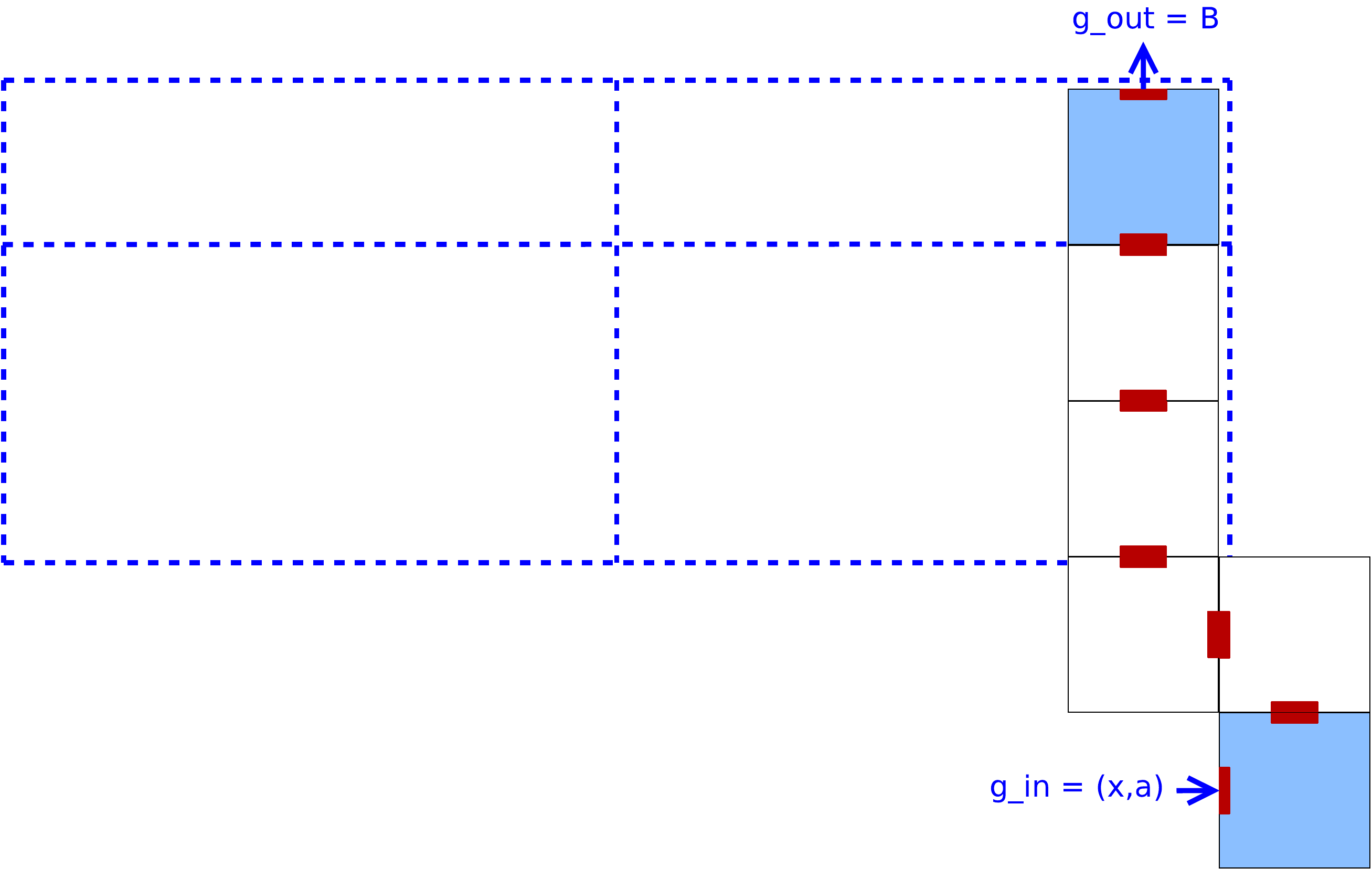}\\

\midrule
\begin{minipage}[b]{0.28\textwidth}\centering
\includegraphics[scale=0.22]{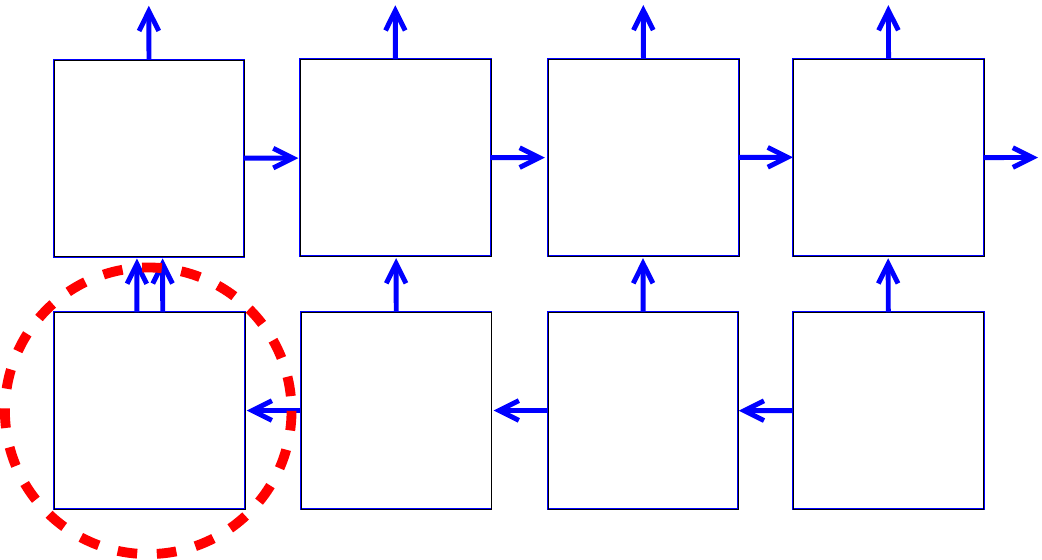}
\includegraphics[scale=0.46]{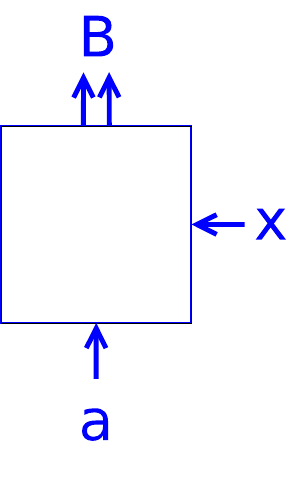}
\hspace{0.04\textwidth}
\small{Direction West, North 2 (DWN2)}
\end{minipage}&
\includegraphics[scale=0.23]{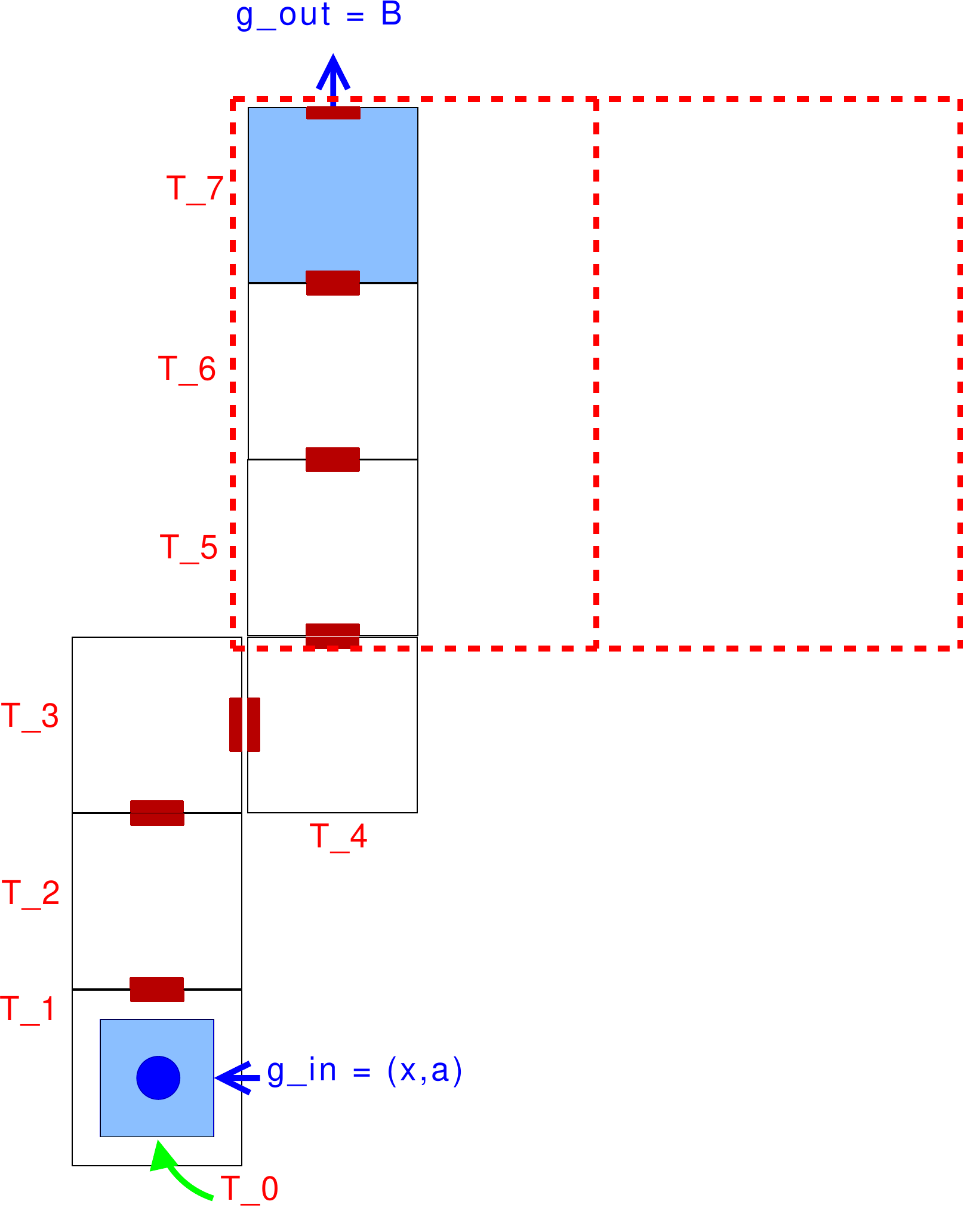}&
\includegraphics[scale=0.23]{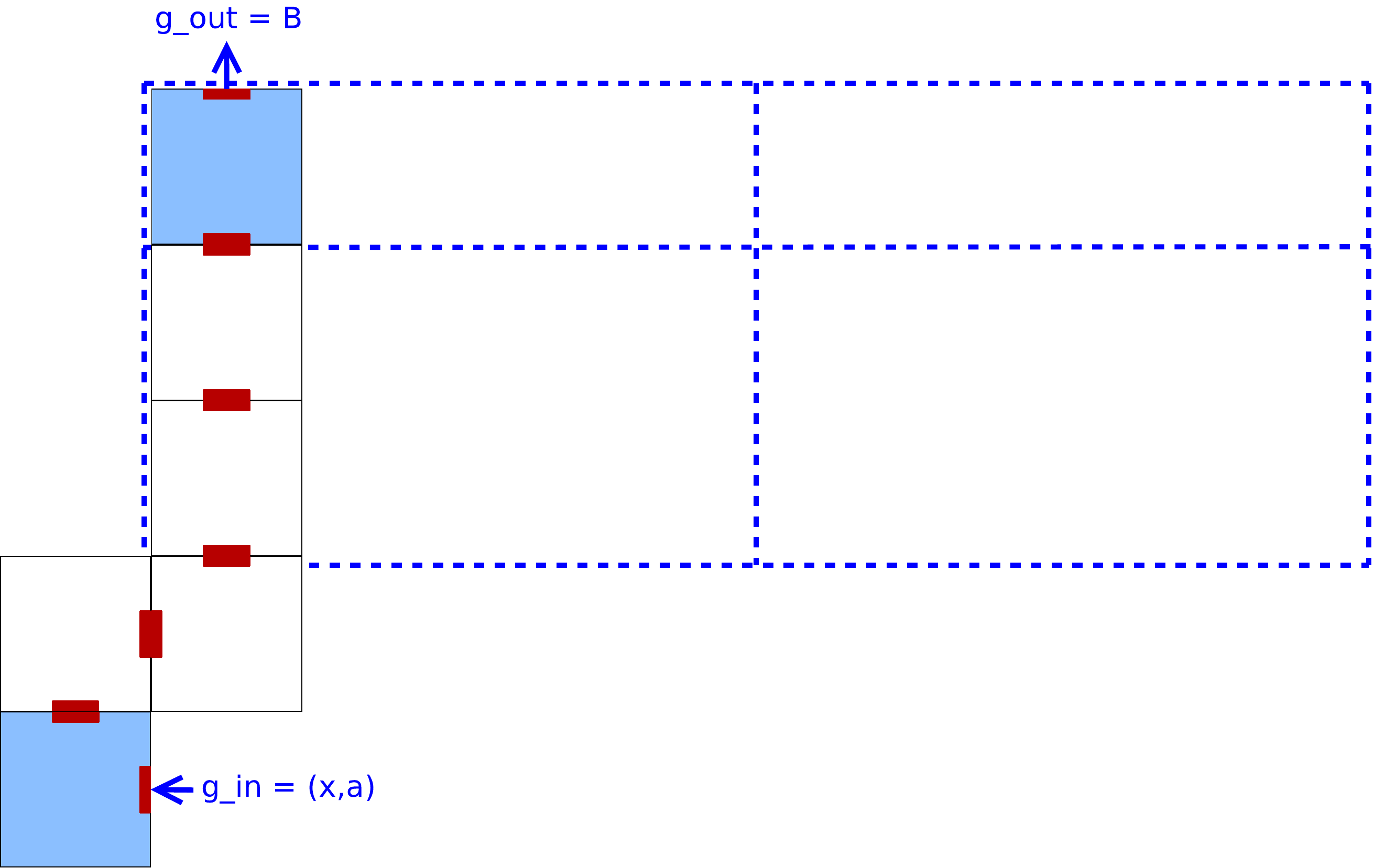}\\

\midrule
\begin{minipage}[b]{0.28\textwidth}\centering
\includegraphics[scale=0.22]{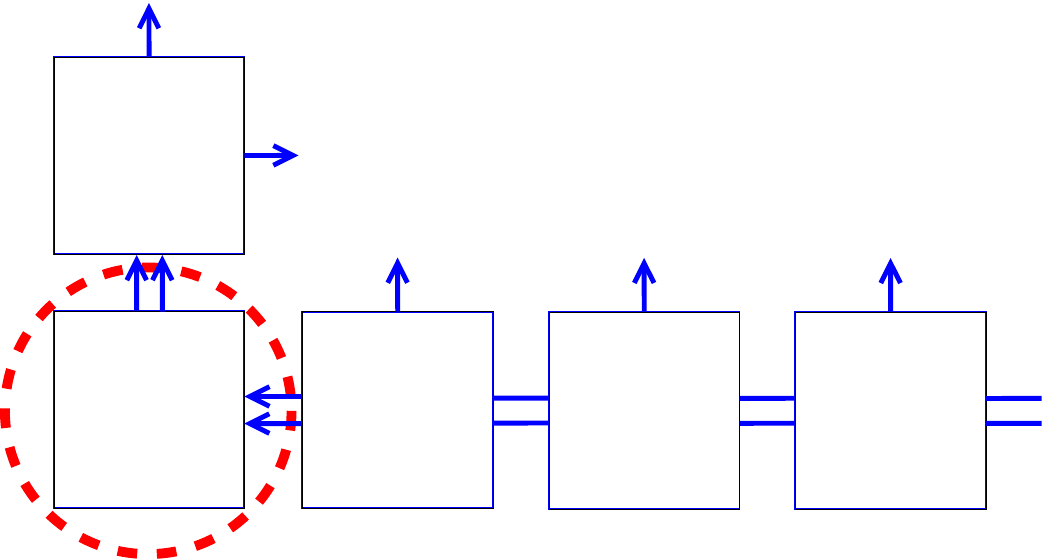}
\includegraphics[scale=0.46]{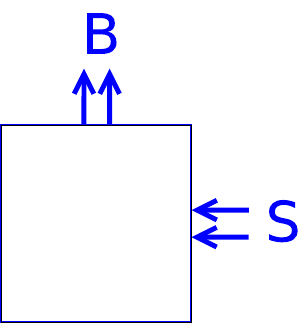}
\hspace{0.04\textwidth}
\small{Side West, North 2 (SWN2)}
\end{minipage}&
\includegraphics[scale=0.23]{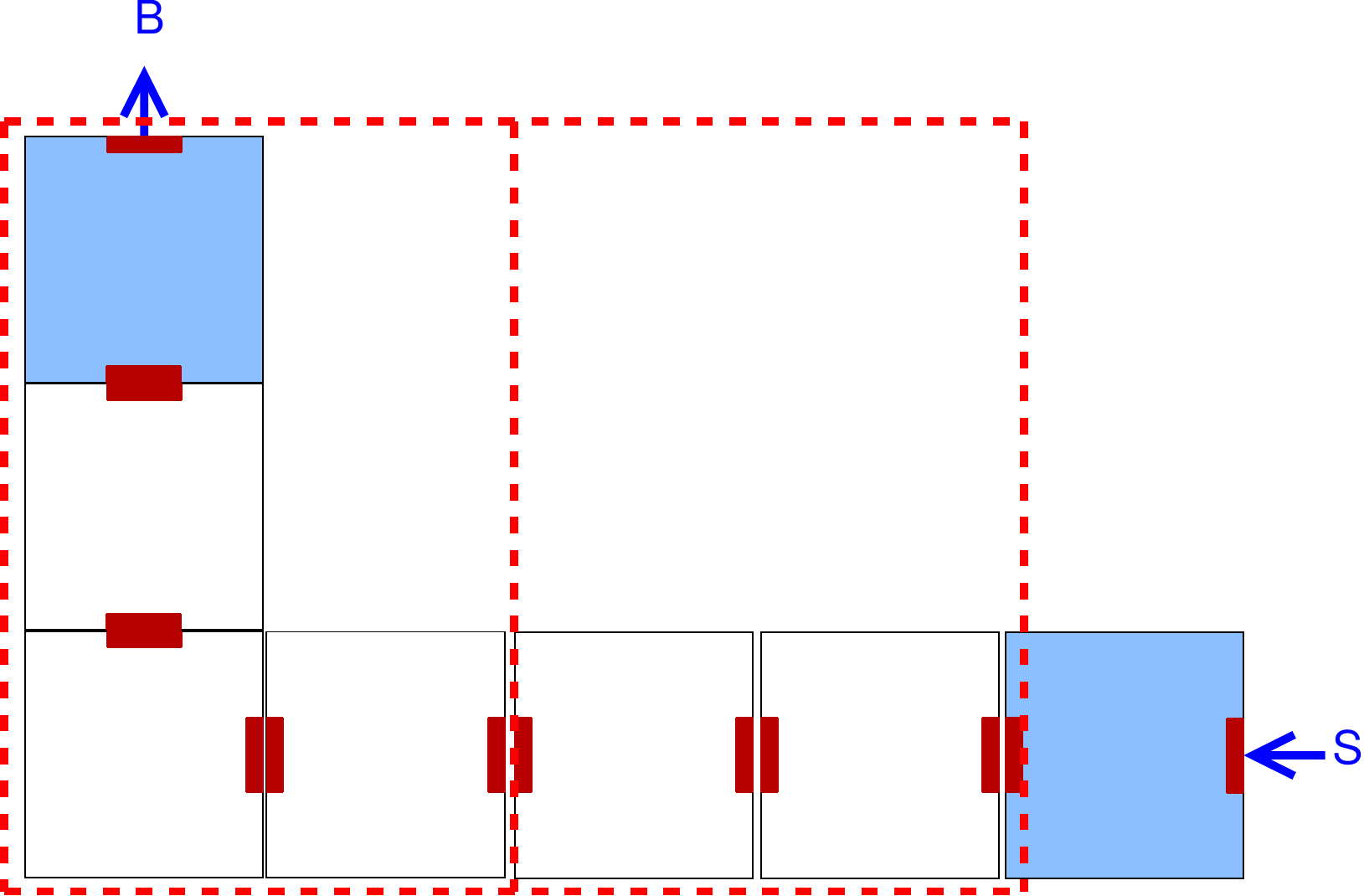}&
\includegraphics[scale=0.23]{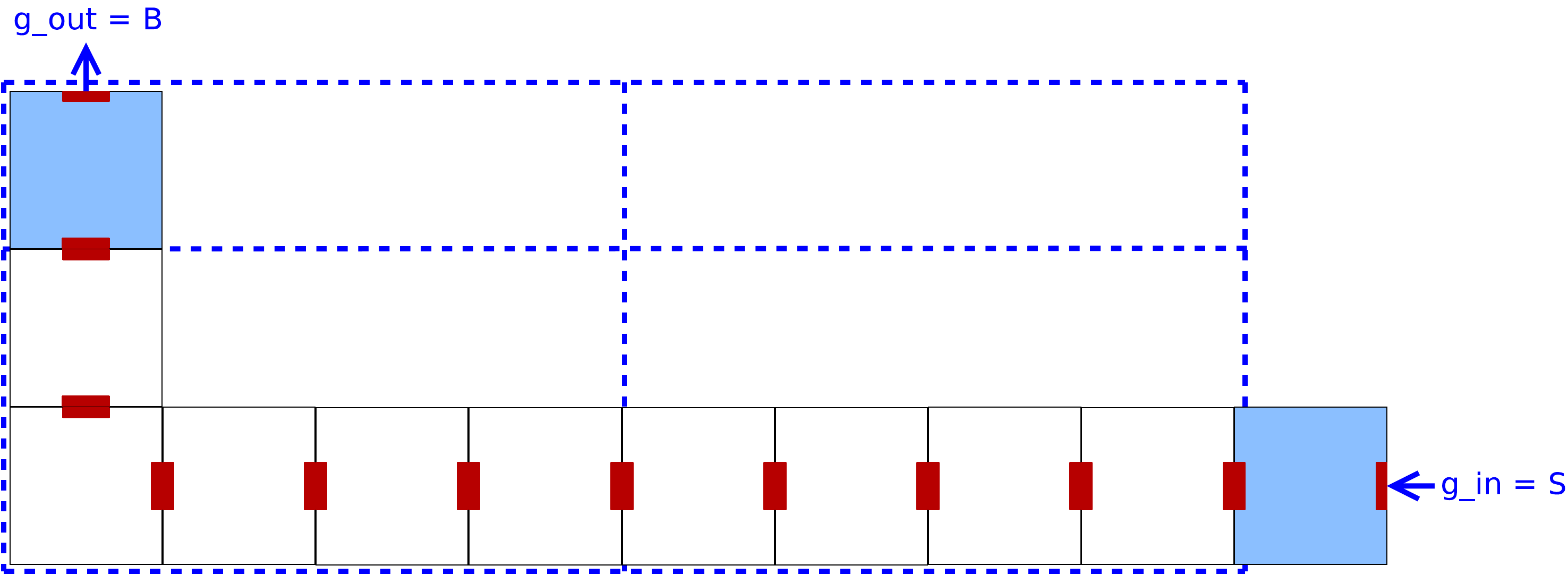}\\

\midrule
\begin{minipage}[b]{0.28\textwidth}\centering
\includegraphics[scale=0.22]{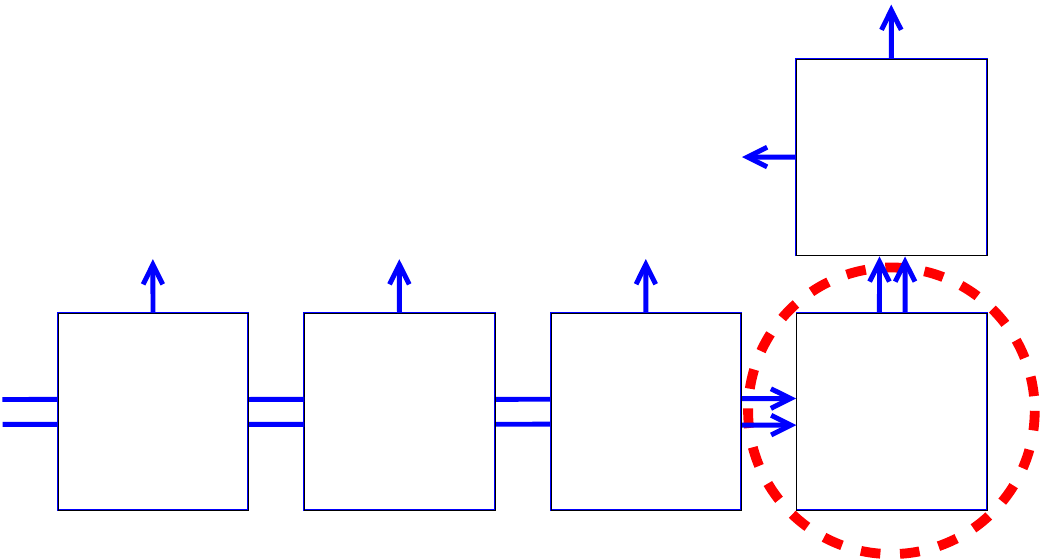}
\includegraphics[scale=0.46]{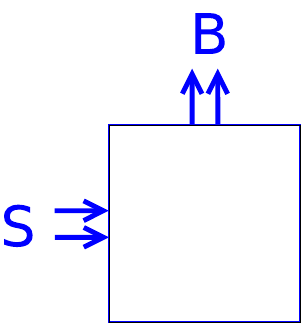}
\hspace{0.04\textwidth}
\small{Side East, North 2 (SEN2)}
\end{minipage}&
\includegraphics[scale=0.23]{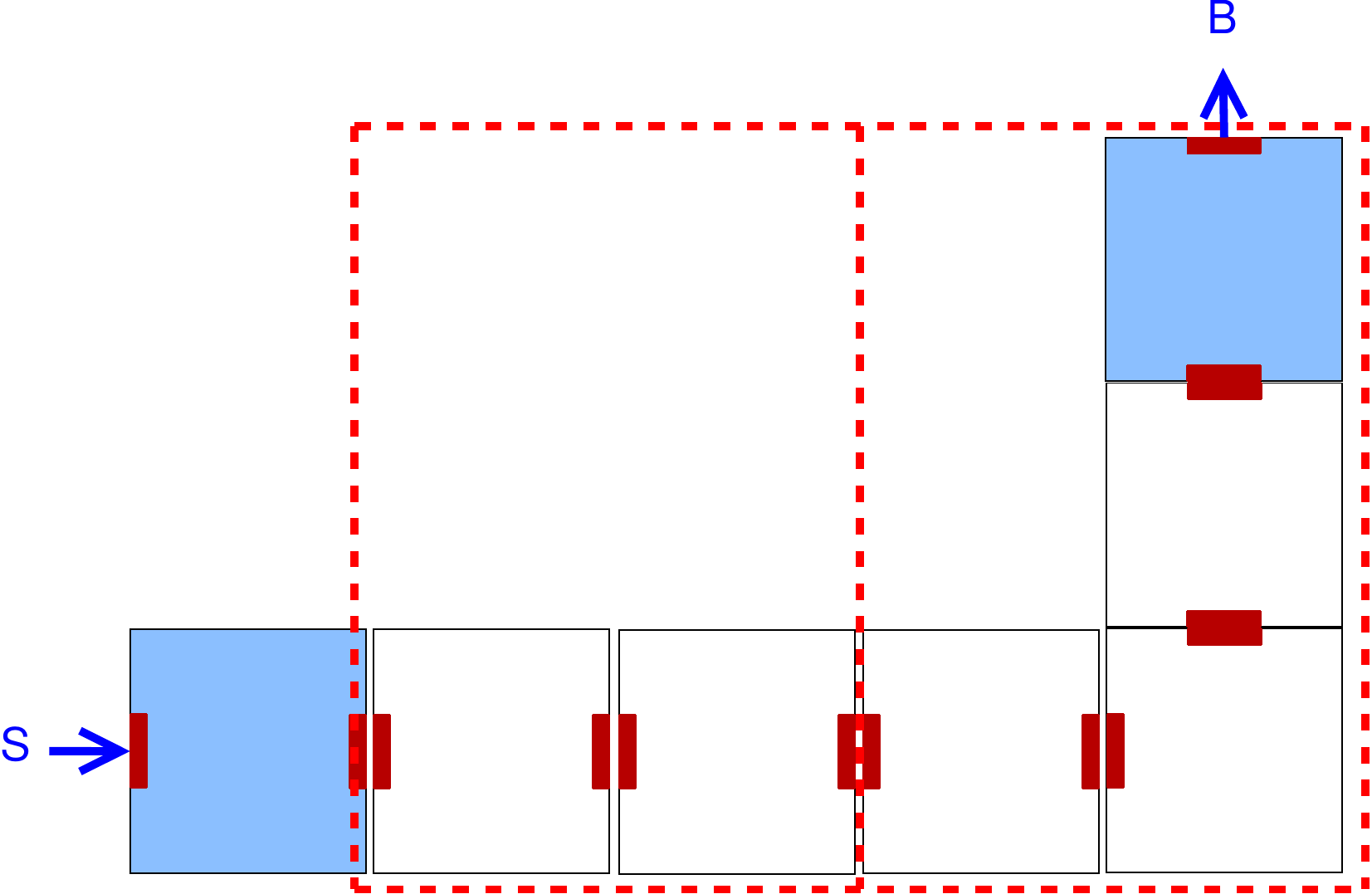}&
\includegraphics[scale=0.23]{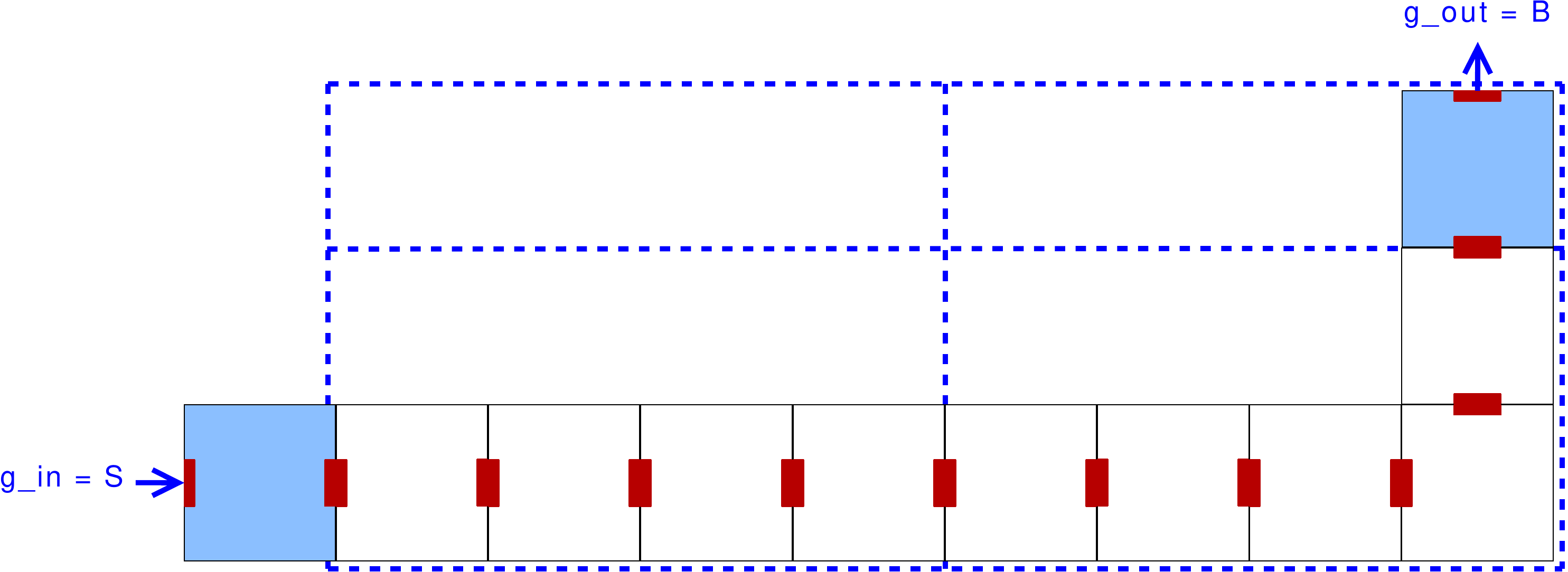}\\

\midrule
\begin{minipage}[b]{0.28\textwidth}\centering
\includegraphics[scale=0.22]{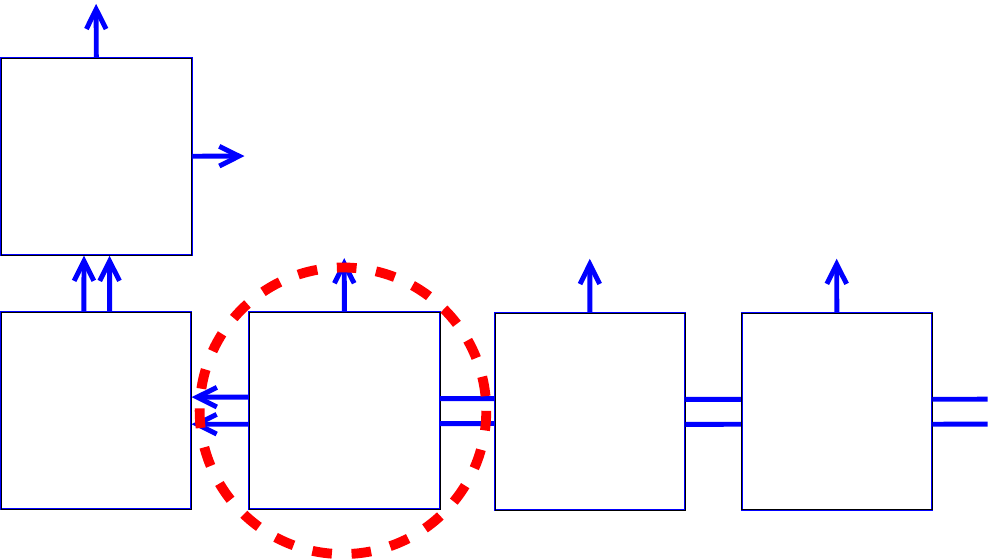}
\includegraphics[scale=0.46]{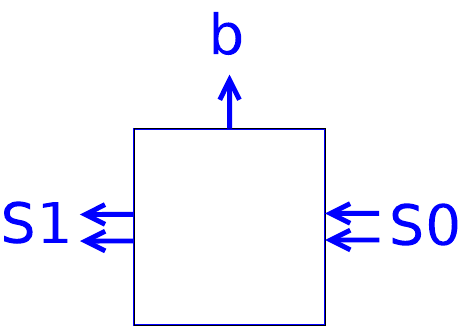}
\hspace{0.04\textwidth}
\small{Fixed Location, Direction West (FW)}
\end{minipage}&
\includegraphics[scale=0.23]{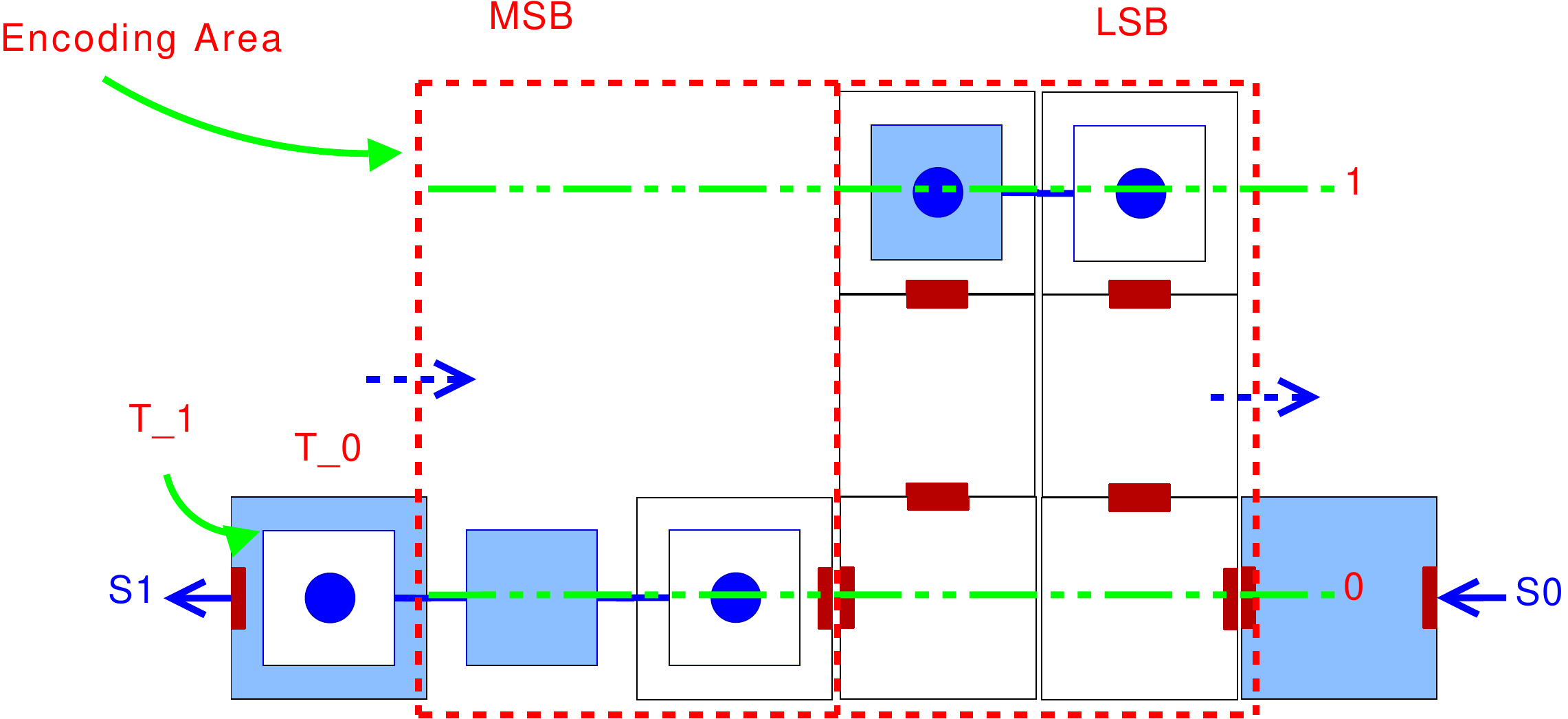}&
\includegraphics[scale=0.23]{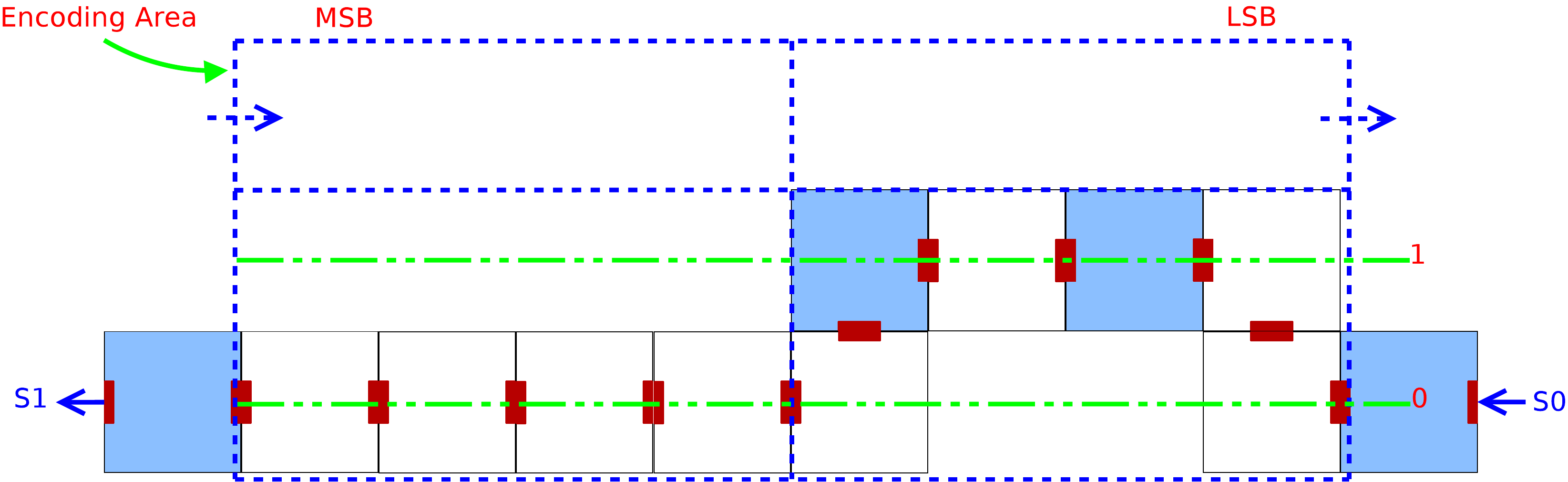}\\

\midrule
\begin{minipage}[b]{0.28\textwidth}\centering
\includegraphics[scale=0.22]{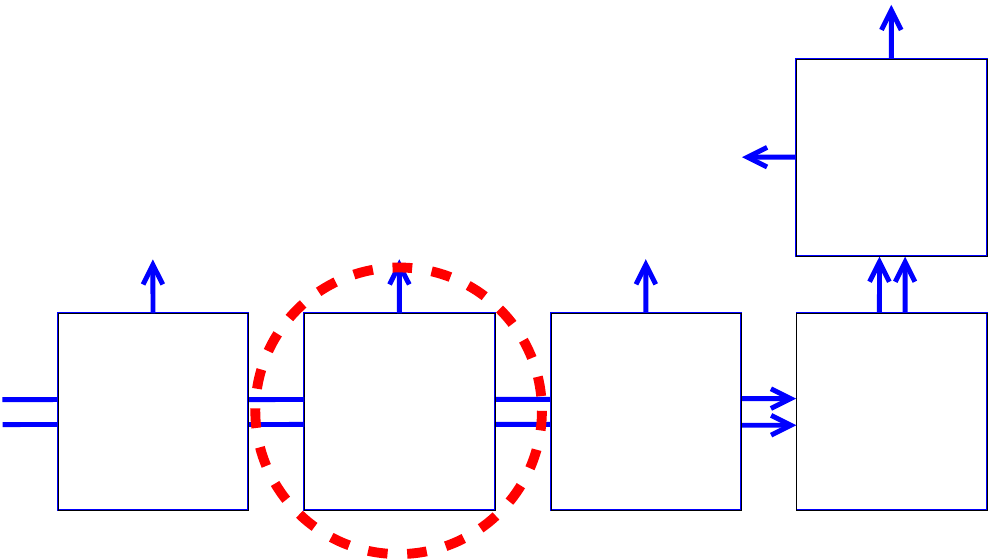}
\includegraphics[scale=0.46]{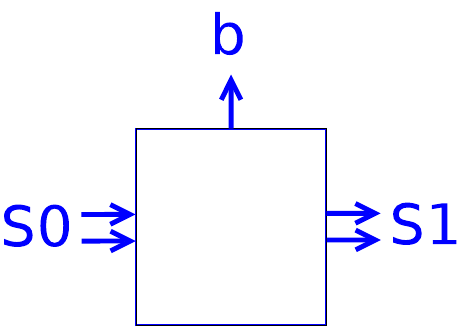}
\hspace{0.04\textwidth}
\small{Fixed Location, Direction East (FE)}
\end{minipage}&
\includegraphics[scale=0.23]{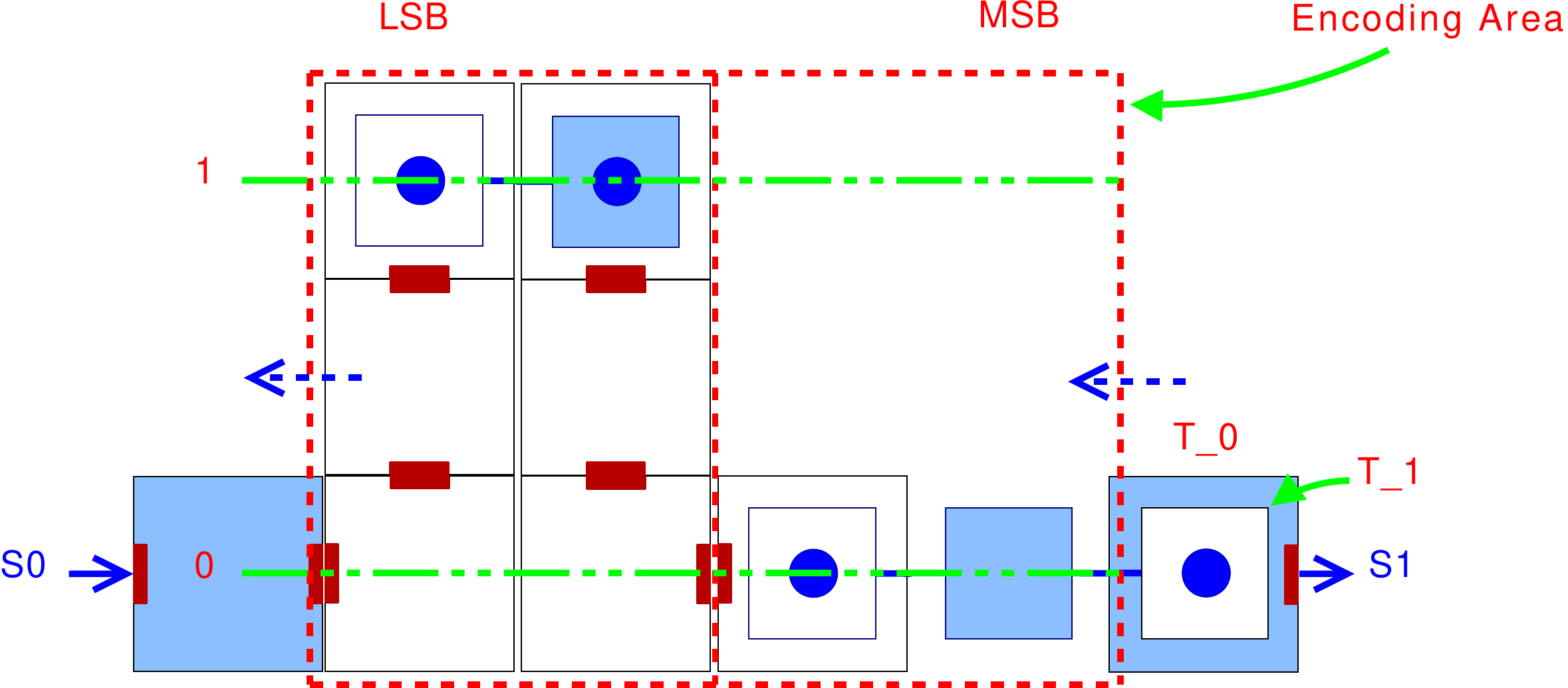}&
\includegraphics[scale=0.23]{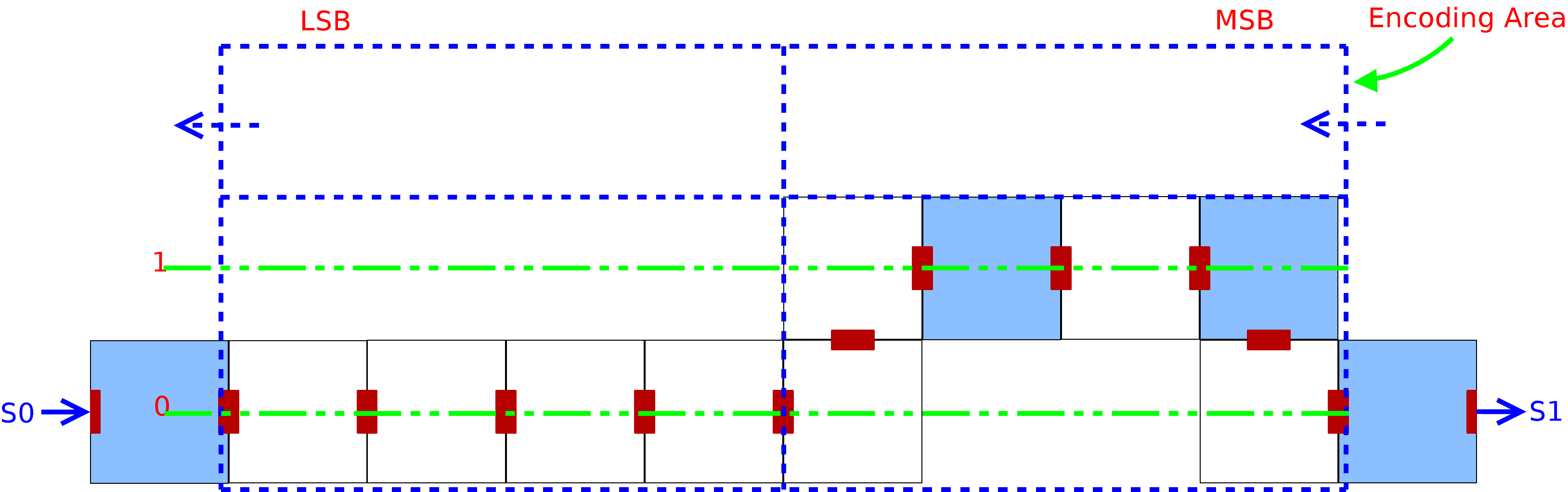}\\

\midrule
\begin{minipage}[b]{0.28\textwidth}\centering
\includegraphics[scale=0.22]{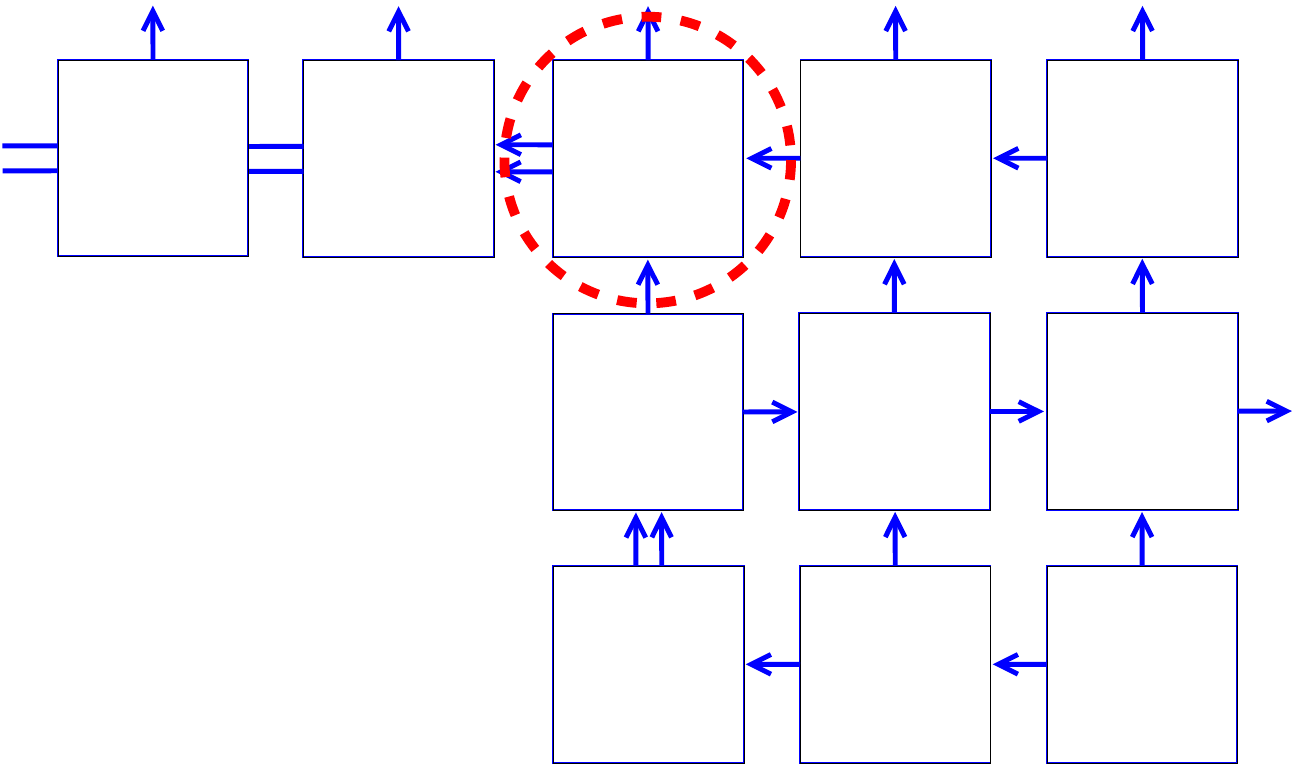}
\includegraphics[scale=0.46]{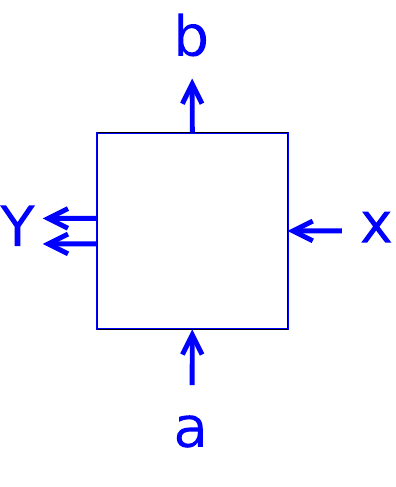}
\hspace{0.04\textwidth}
\small{Direction West, West 2 (DWW2)}
\end{minipage}&
\includegraphics[scale=0.23]{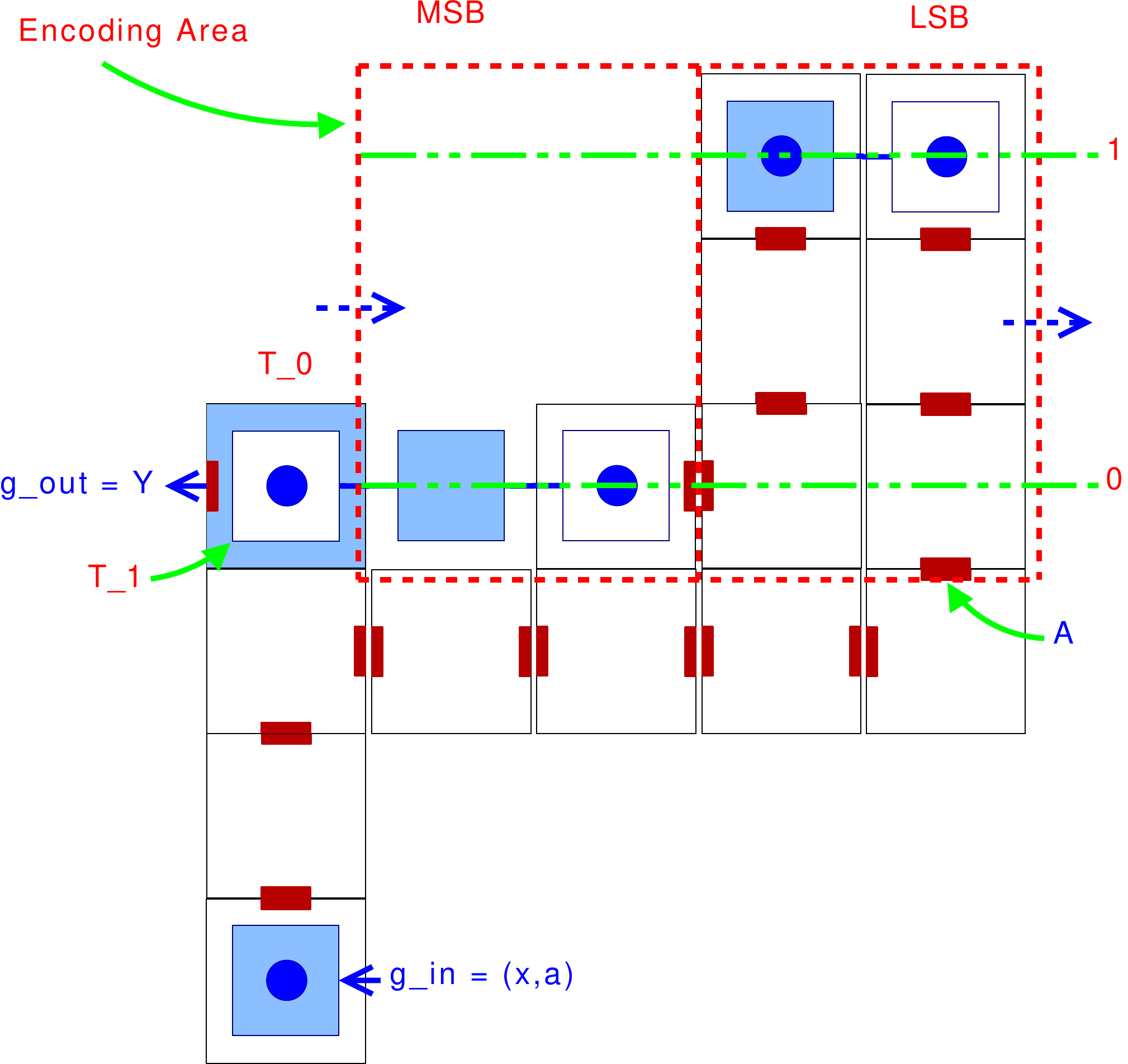}&
\includegraphics[scale=0.23]{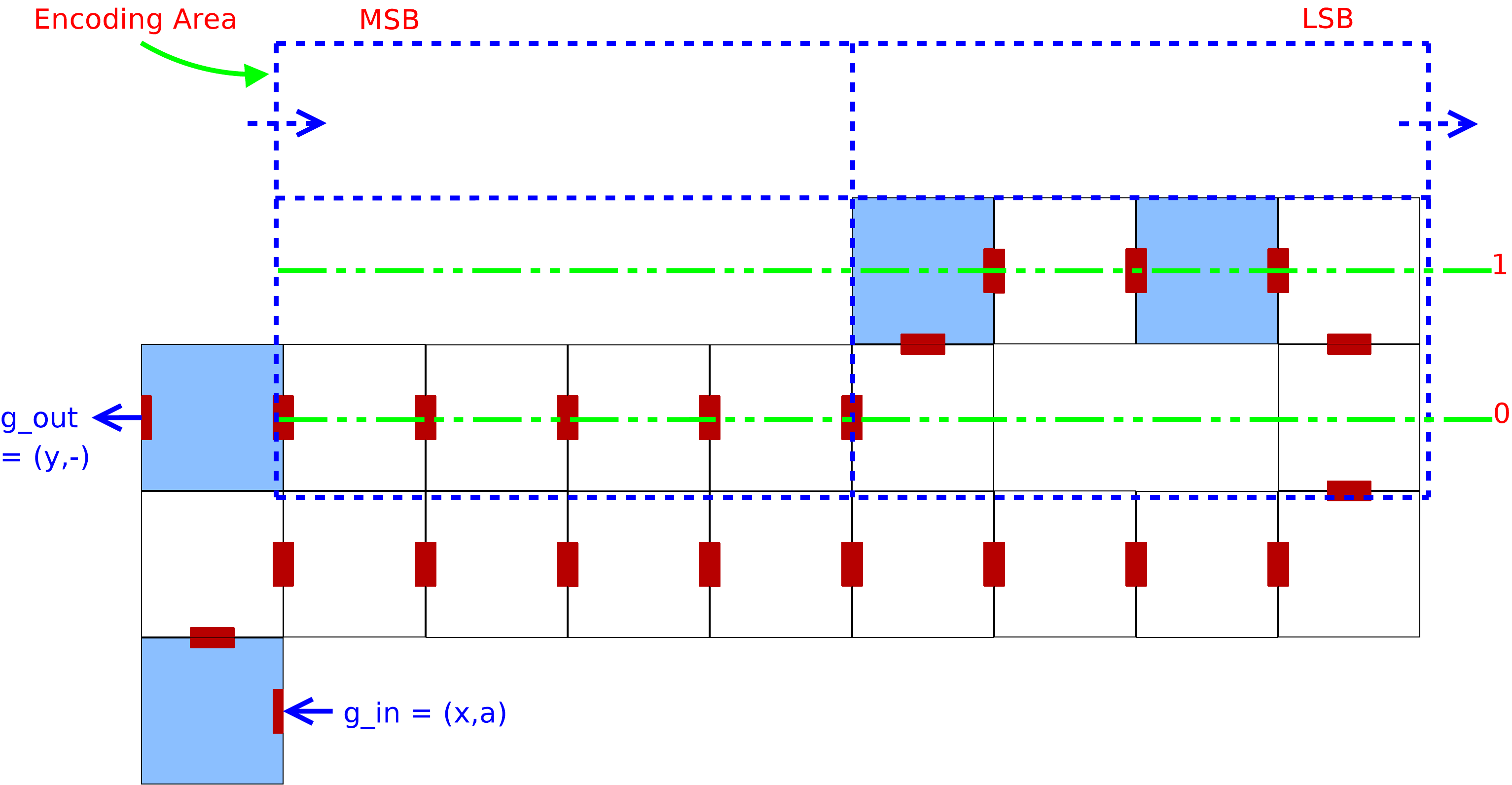}\\

\midrule
\begin{minipage}[b]{0.28\textwidth}\centering
\includegraphics[scale=0.22]{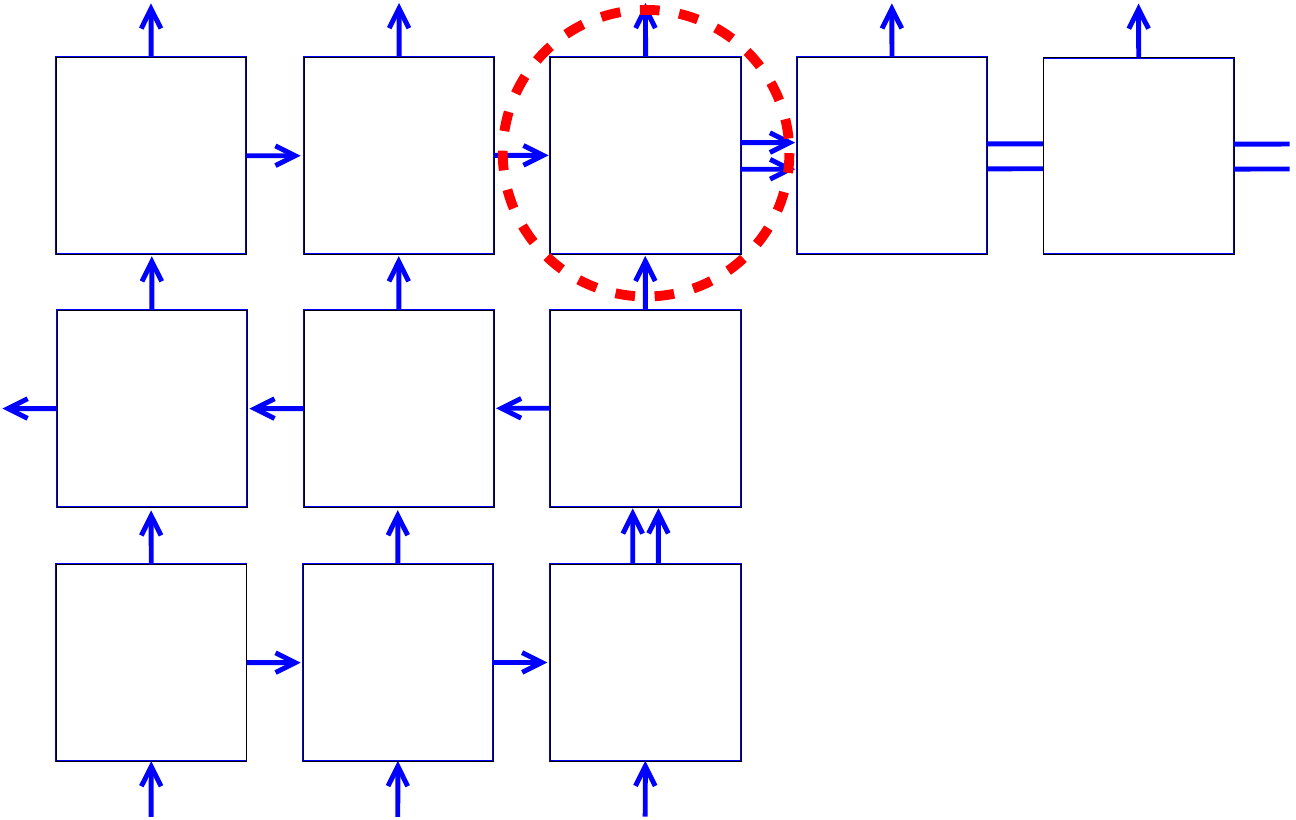}
\includegraphics[scale=0.46]{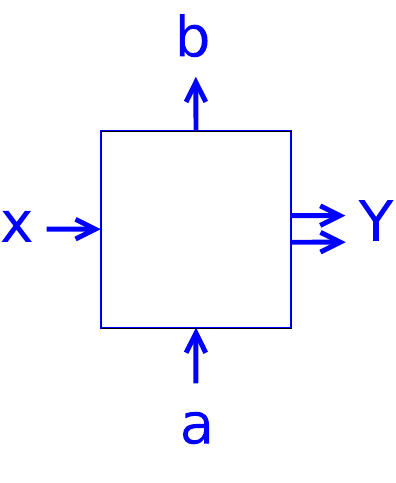}
\hspace{0.04\textwidth}
\small{Direction East, East 2 (DEE2)}
\end{minipage}&
\includegraphics[scale=0.23]{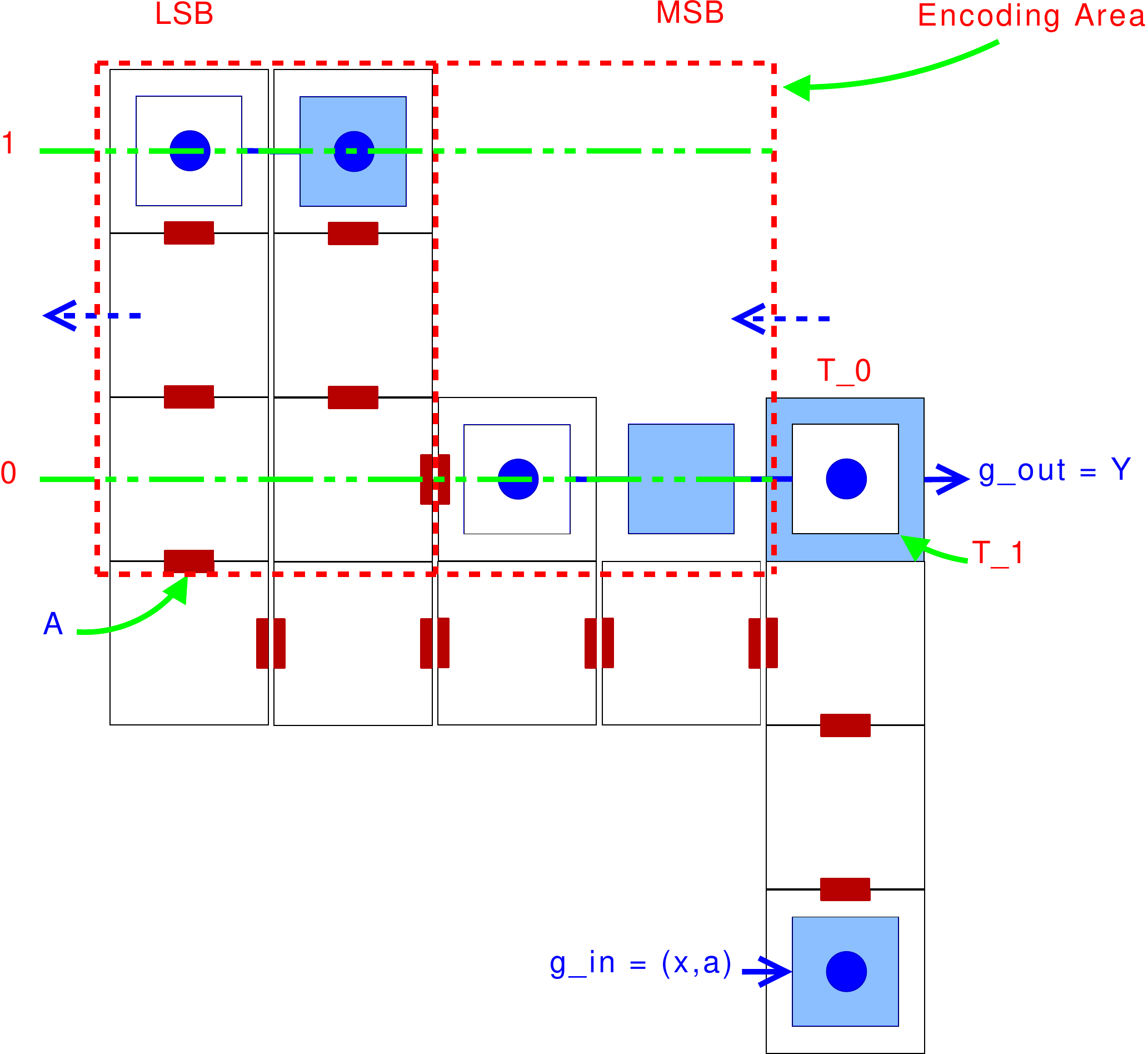}&
\includegraphics[scale=0.23]{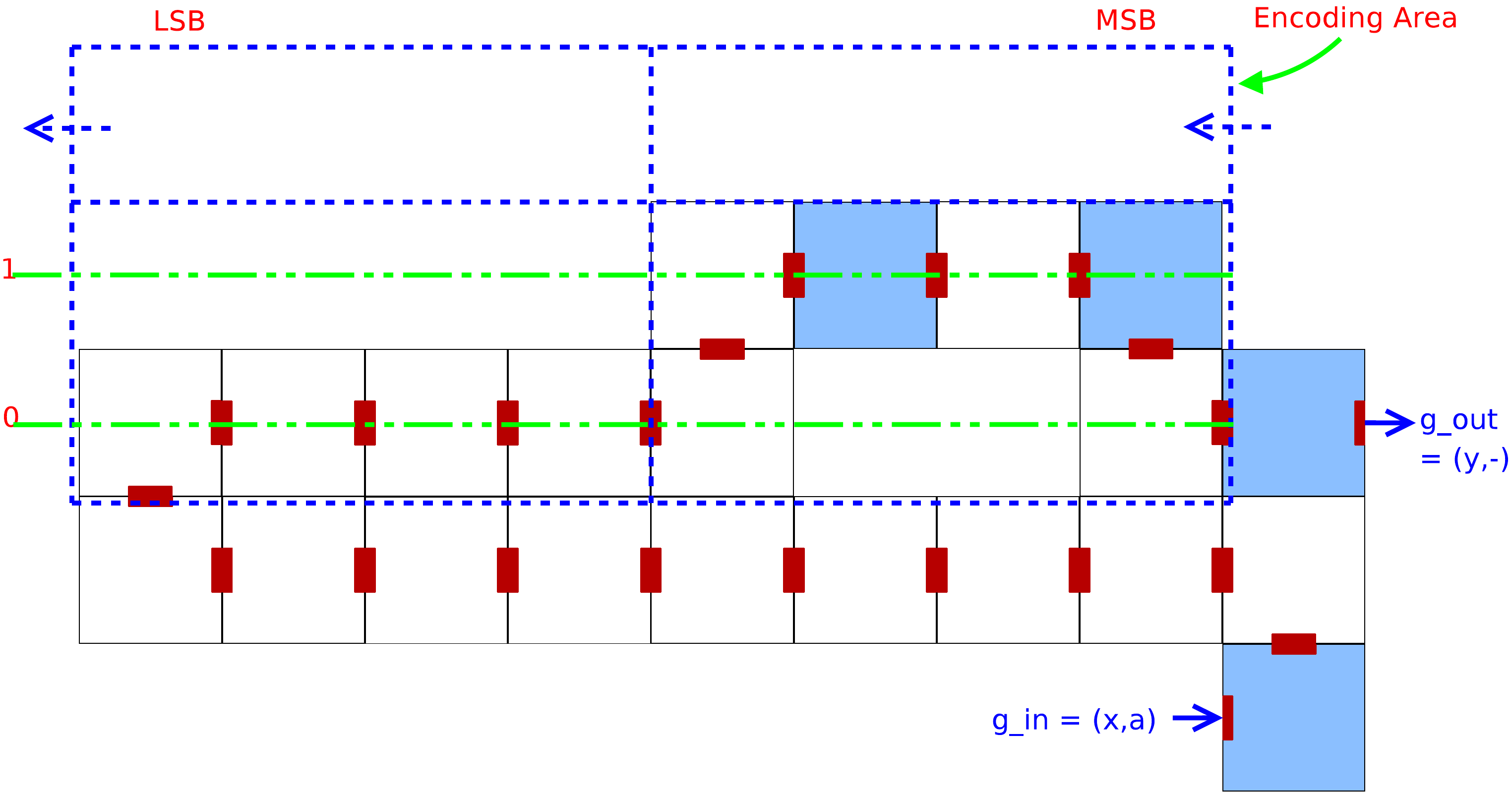}\\

\midrule
\begin{minipage}[b]{0.28\textwidth}\centering
\includegraphics[scale=0.22]{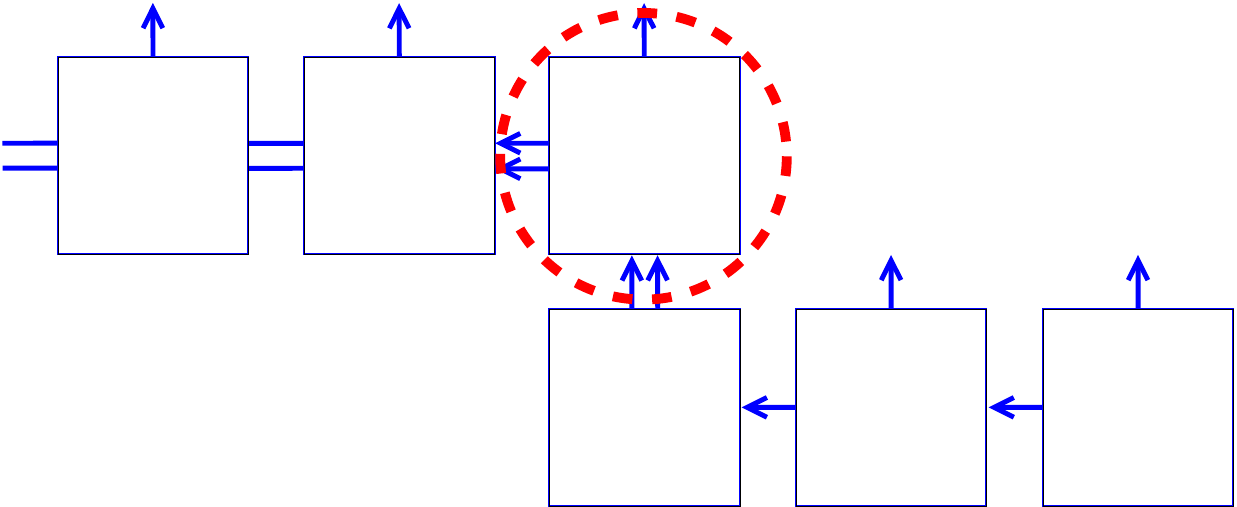}
\includegraphics[scale=0.46]{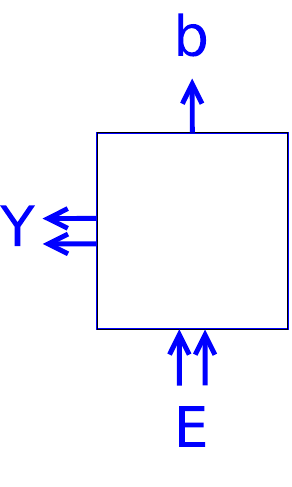}
\hspace{0.04\textwidth}
\small{Turn at West, West 2 (TWW2)}
\end{minipage}&
\includegraphics[scale=0.23]{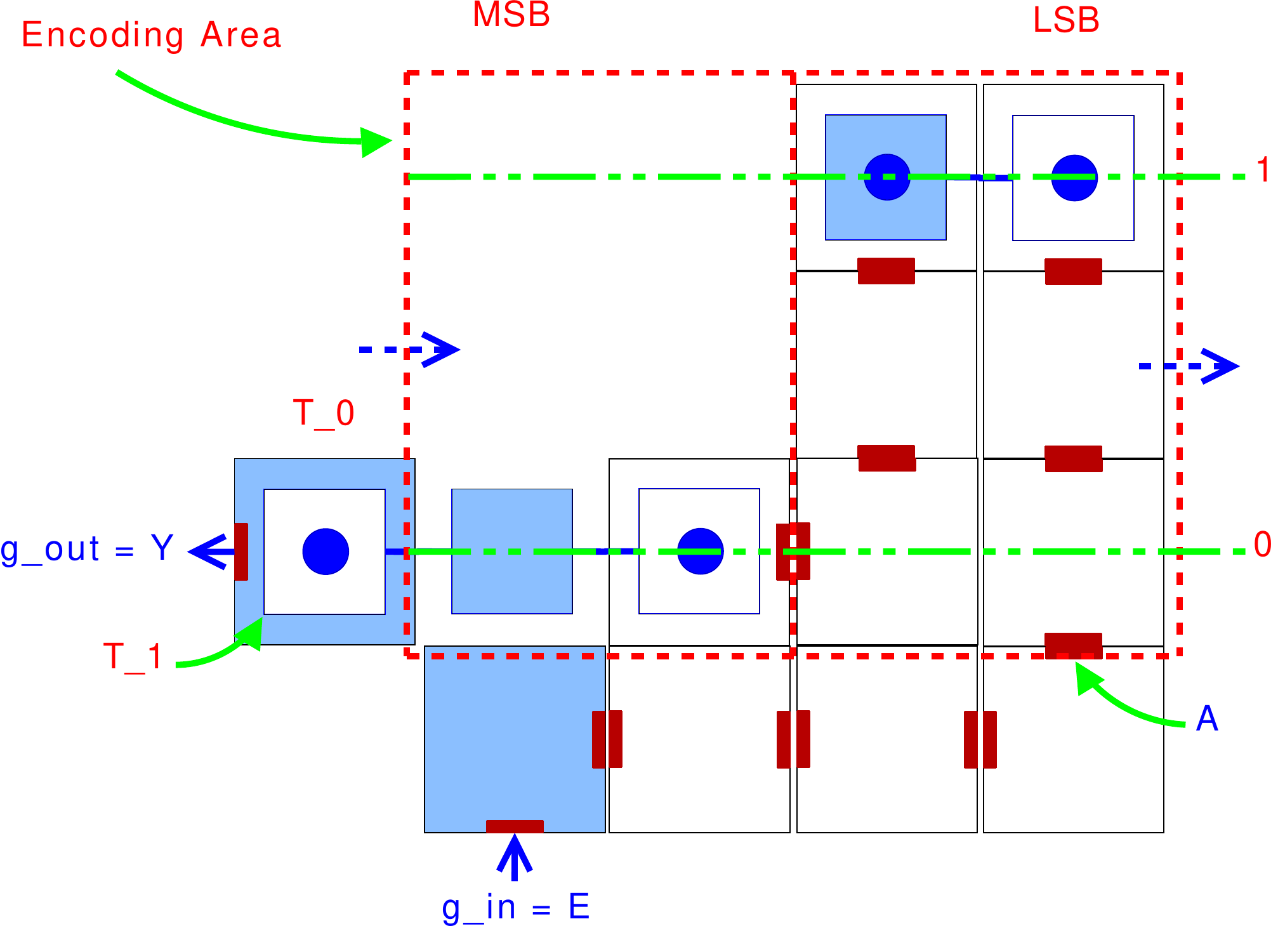}&
\includegraphics[scale=0.23]{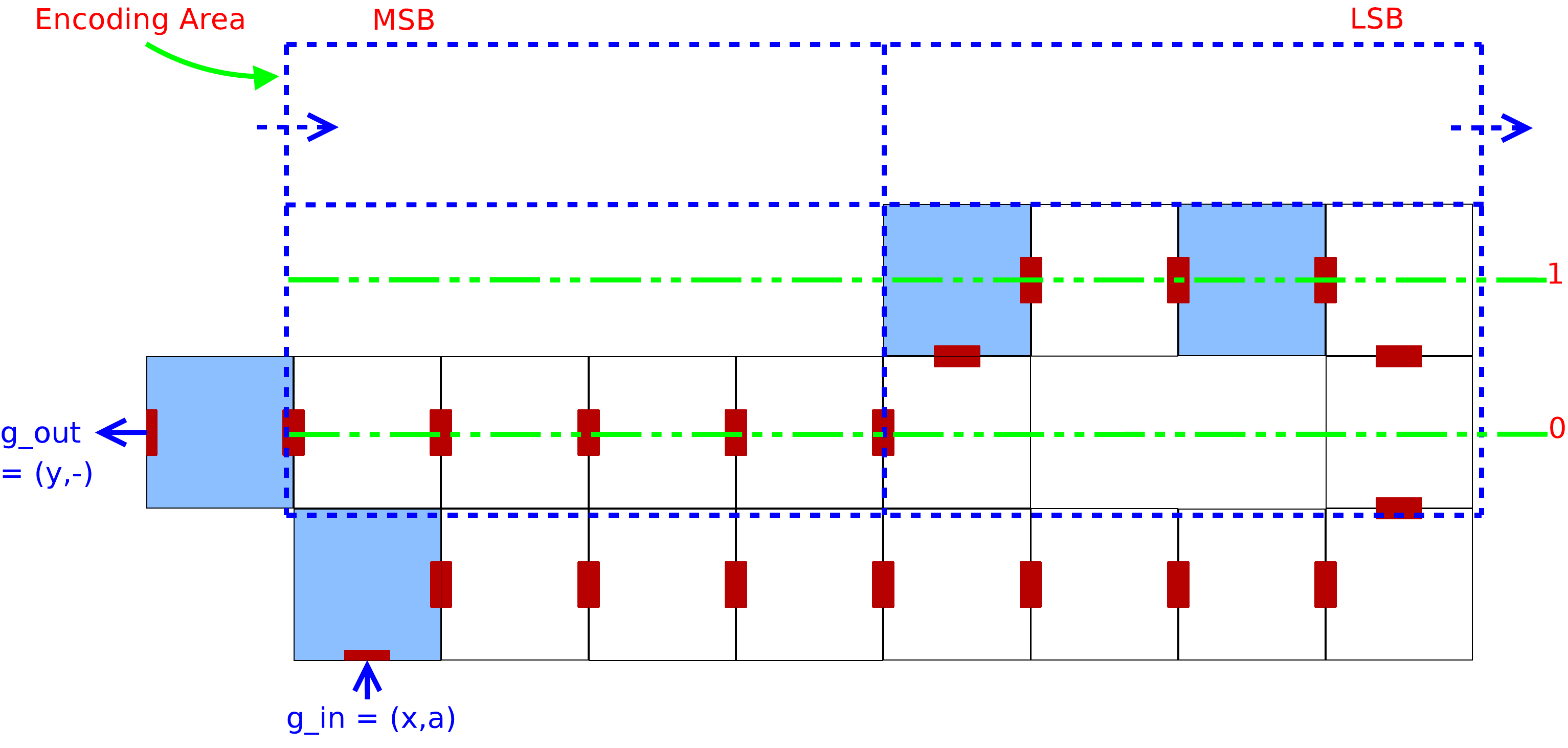}\\

\midrule
\begin{minipage}[b]{0.28\textwidth}\centering
\includegraphics[scale=0.22]{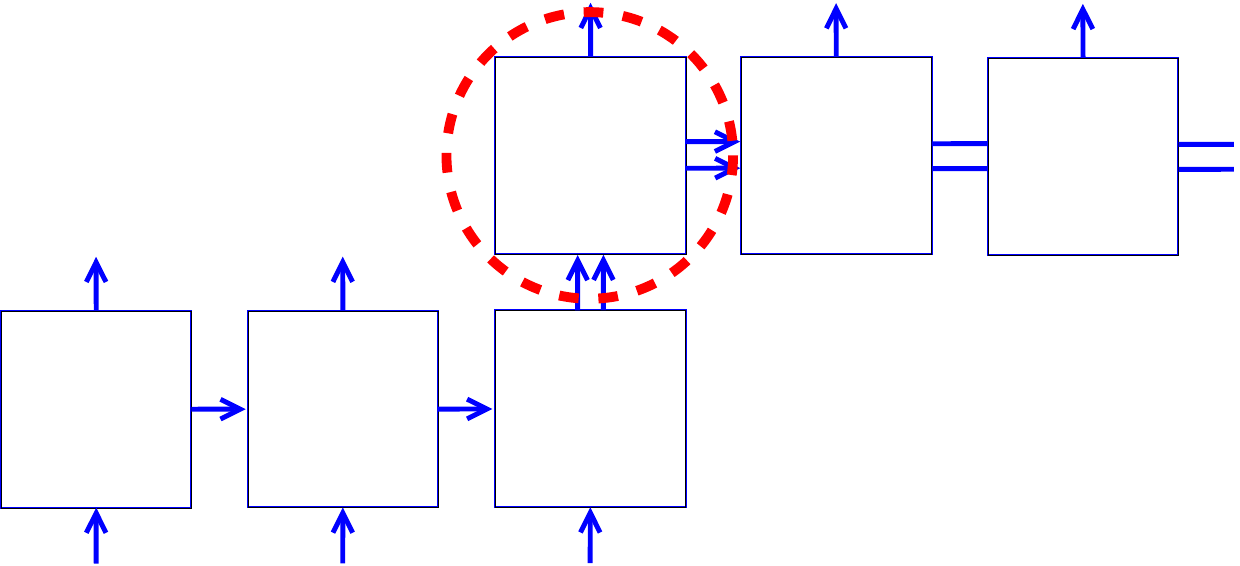}
\includegraphics[scale=0.46]{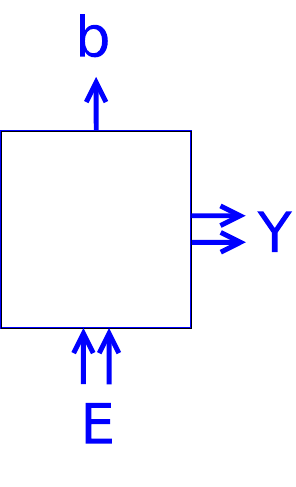}
\hspace{0.04\textwidth}
\small{Turn at East, East 2 (TEE2)}
\end{minipage}&
\includegraphics[scale=0.23]{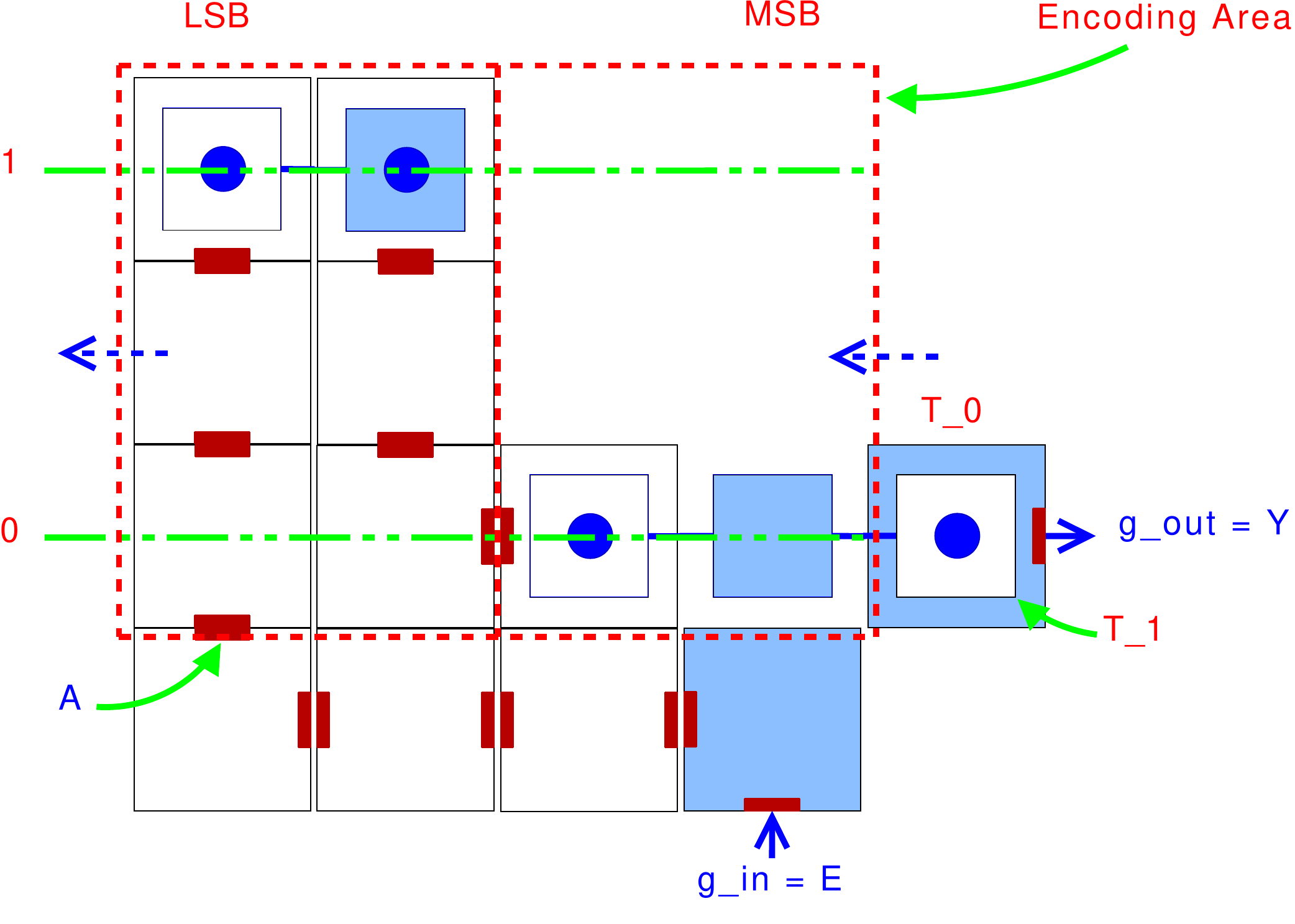}&
\includegraphics[scale=0.23]{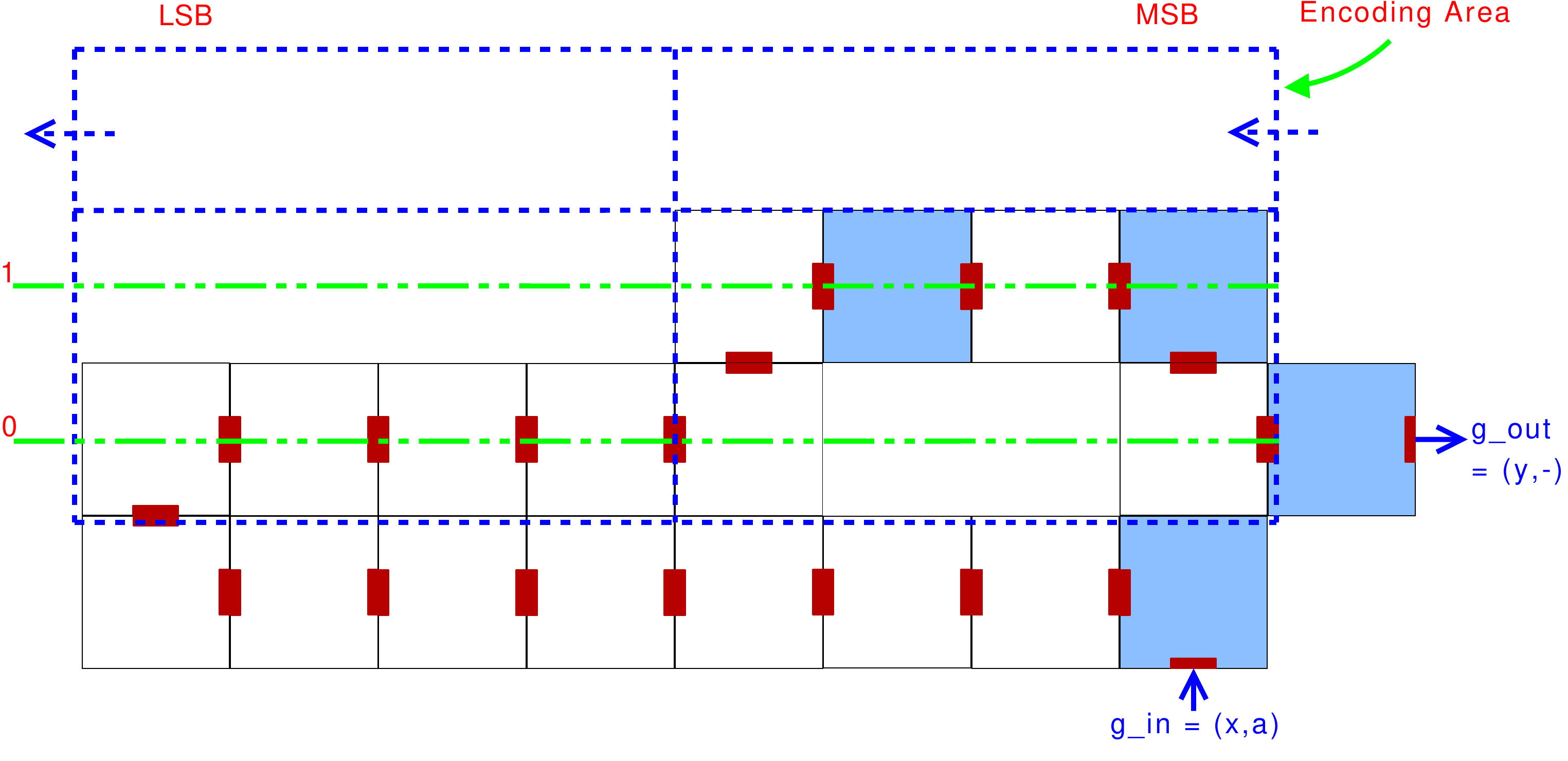}\\

\midrule
\begin{minipage}[b]{0.28\textwidth}\centering
\includegraphics[scale=0.22]{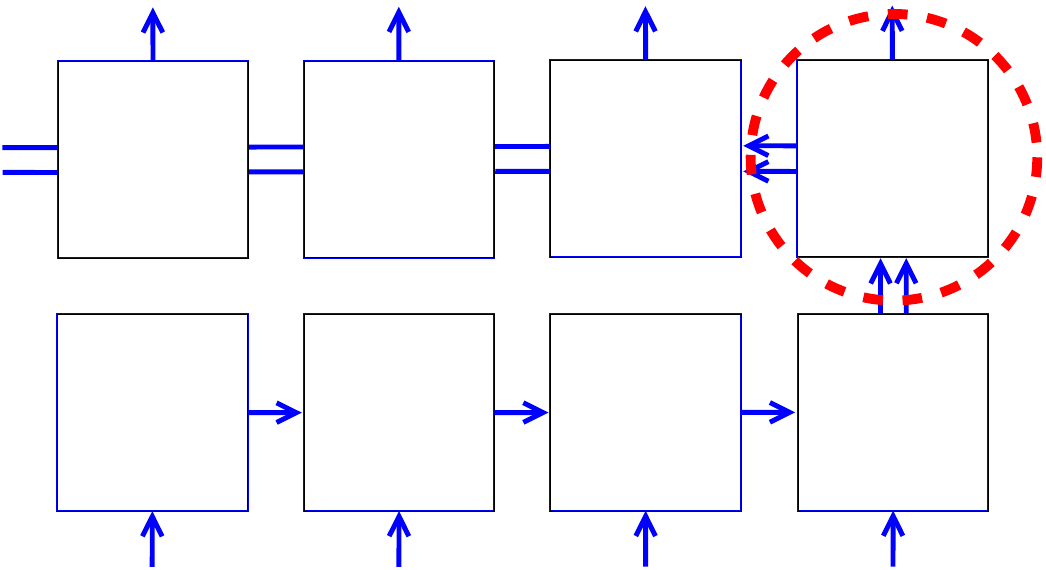}
\includegraphics[scale=0.46]{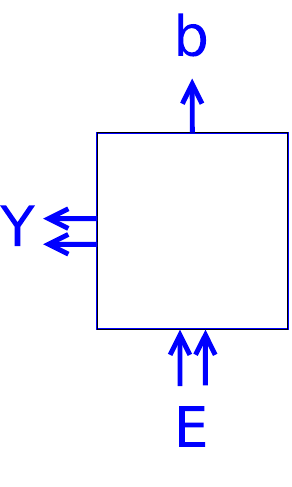}
\hspace{0.04\textwidth}
\small{Turn at East, West 2 (TEW2)}
\end{minipage}&
\includegraphics[scale=0.23]{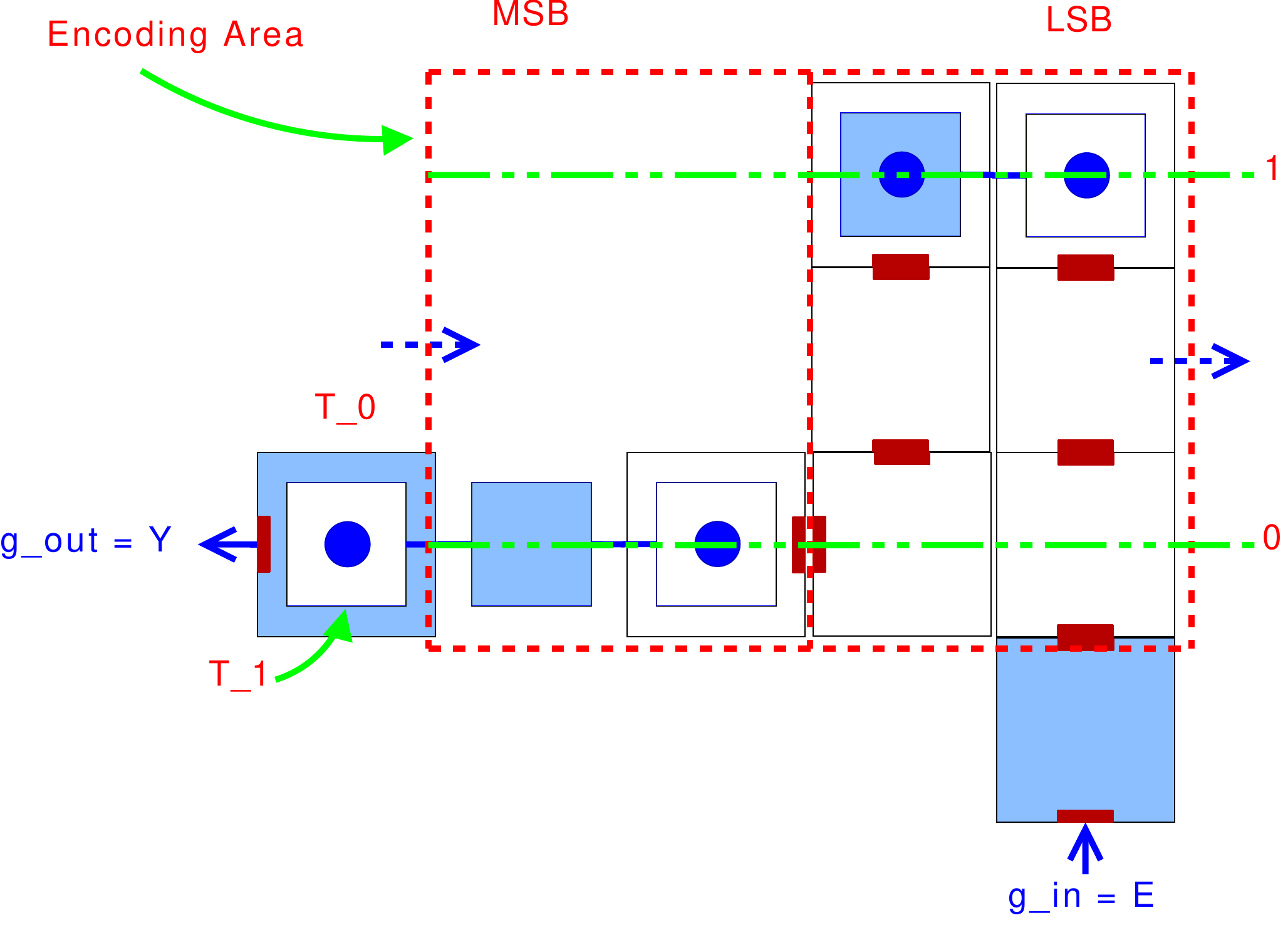}&
\includegraphics[scale=0.23]{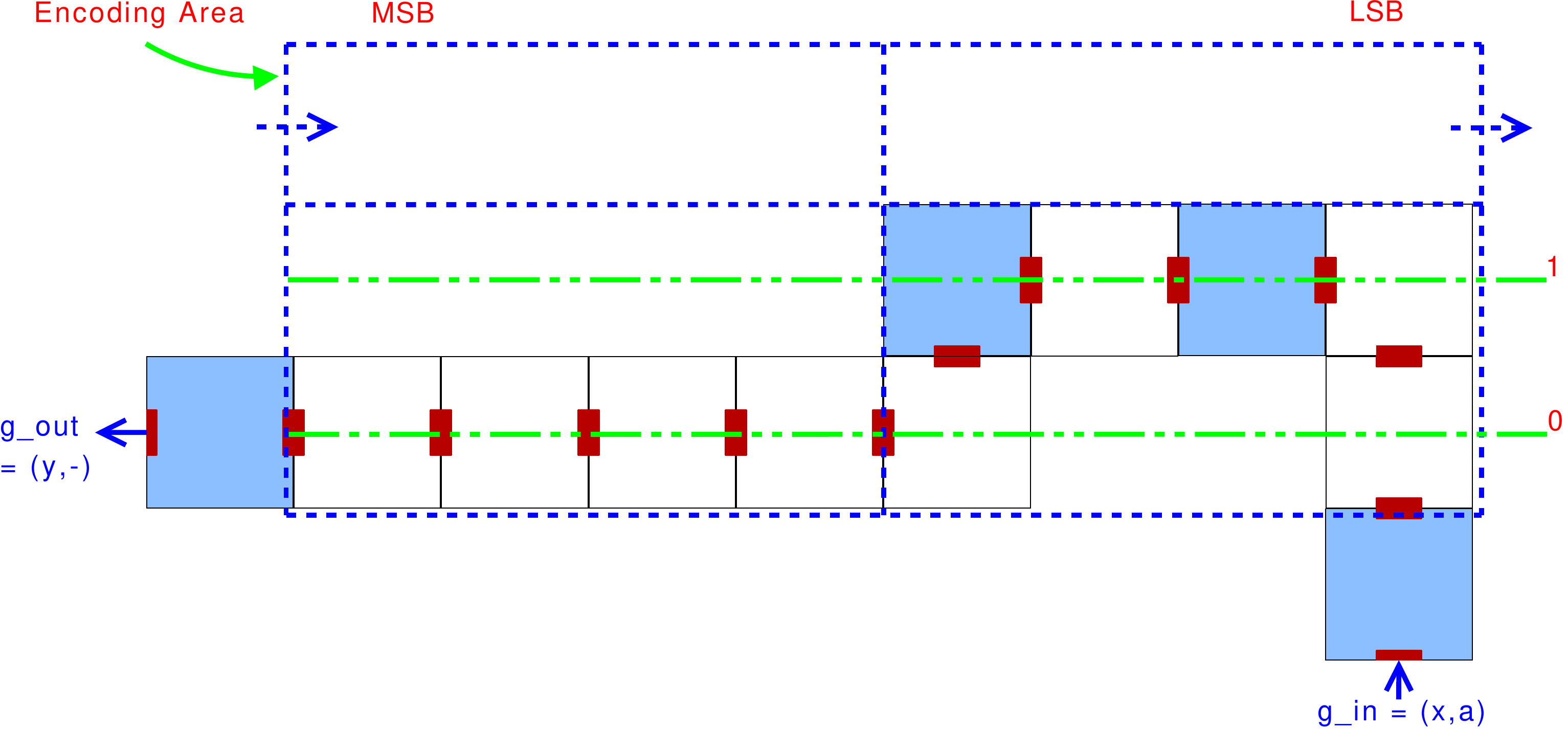}\\

\midrule
\begin{minipage}[b]{0.28\textwidth}\centering
\includegraphics[scale=0.22]{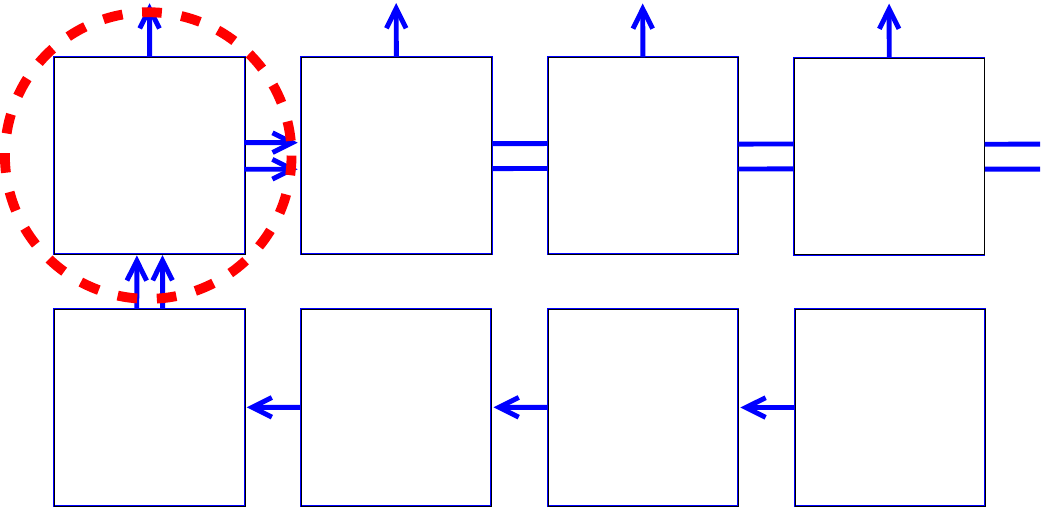}
\includegraphics[scale=0.46]{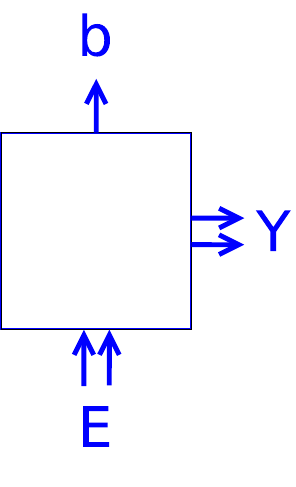}
\hspace{0.04\textwidth}
\small{Turn at West, East 2 (TWE2)}
\end{minipage}&
\includegraphics[scale=0.23]{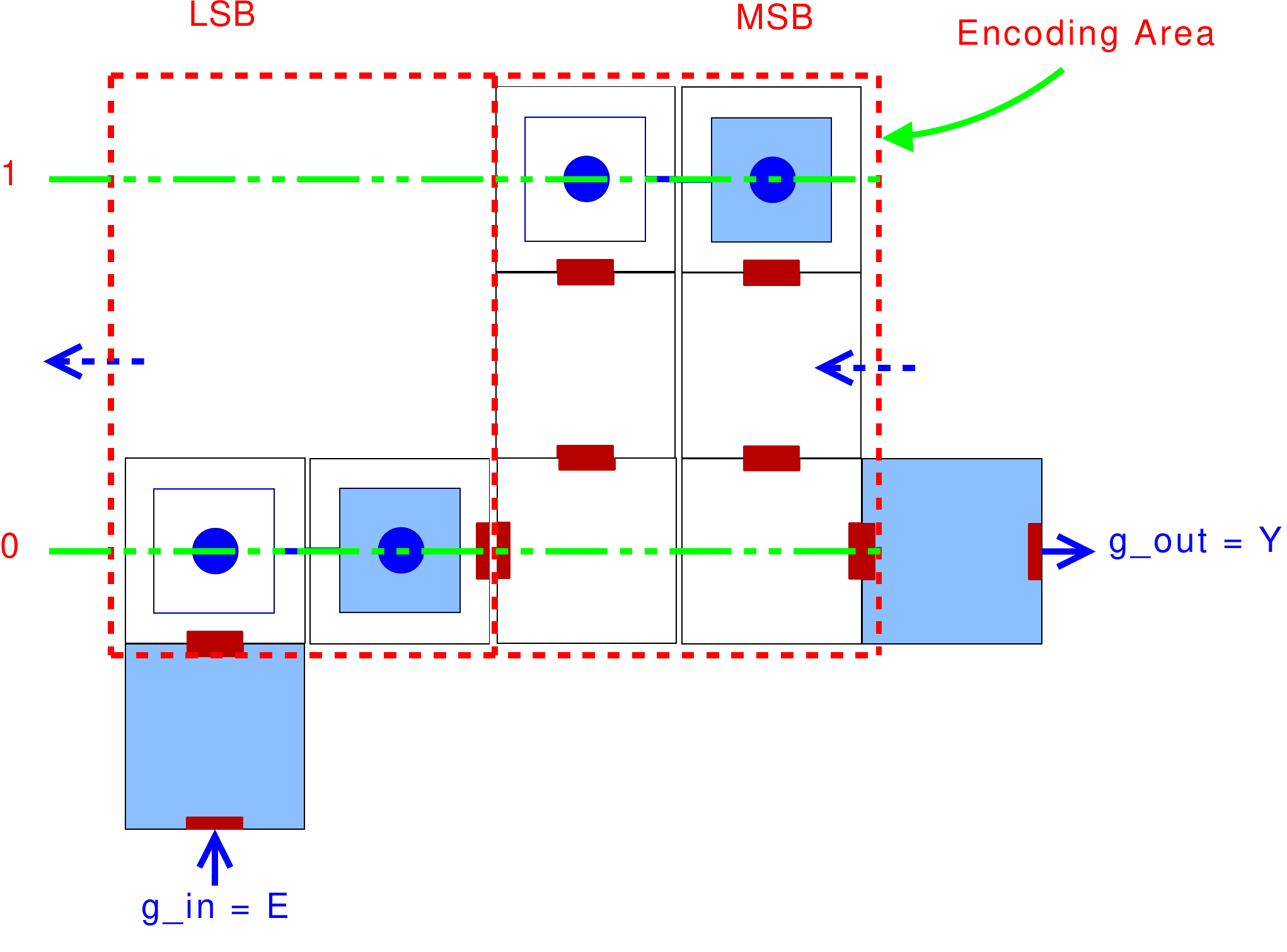}&
\includegraphics[scale=0.23]{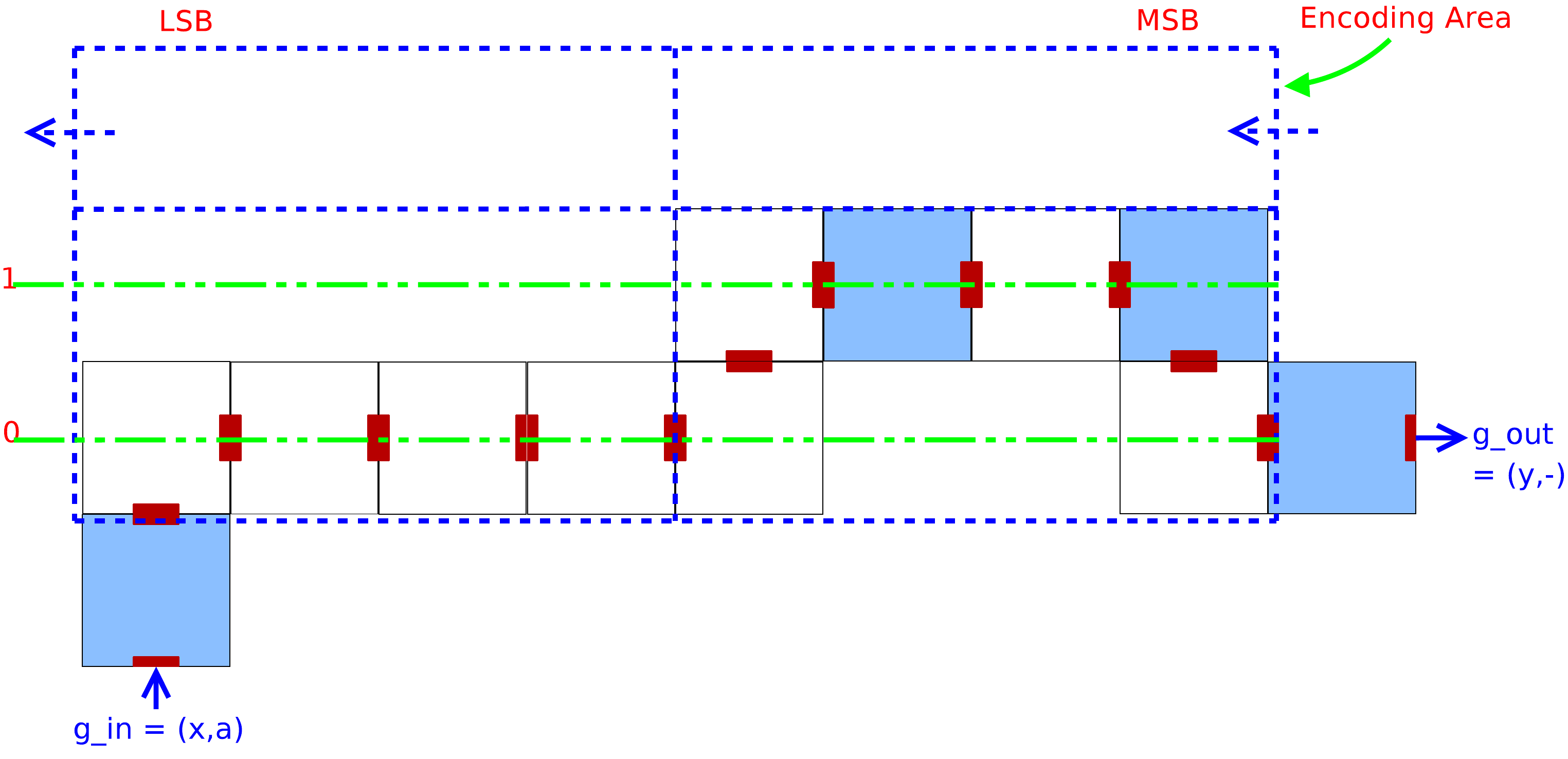}\\

\bottomrule
\end{longtable}

\subsection{Algorithms for Converting Tile Type}

To converting arbitrary zig-zag tile set, the first step is to categorize the tiles
into the sixteen types listed in Table \ref{tab:zigzagmapping}.
Then the zig-zag tiles in 2D can be converted directly to zig-zag tile set in 3D or
probabilistic zig-zag tile set in 2D.
We can recgonize all of the tile types in 2D by observing the positions the tiles
be placed during the growth.
In an other words, a simulator have to be used to get the types of tiles in 2D.



The tiles are categorized into sixteen types, it seems it will cost much effort
to implement the converter by software. As we apply some tricks on the algorithms,
the creating of tile sets converter would be simple.
The tile set type DWW1, DWW2, TEW1, TEW2, TWW2 have similar structure,
the differences between those types are the positions of the input and output glues.
While the tile set type DEE1, DEE2, TWE1, TWE2, TEE2 also have the similar structure,
they are the mirror of the tile set type DWW1, DWW2, TEW1, TEW2, TWW2.
And SWN2 is the mirror of SEN2, DWN2 is mirror of DEN2.

\newcommand{\categorizetiles}{\ensuremath{\mbox{\sc Categorize-Zig-Zag-tiles}}}

\newcommand{\tsconvallcate}{\ensuremath{\mbox{\sc Tileset-convert-all-categories}}}

\newcommand{\genenctowest}{\ensuremath{\mbox{\sc Generate-encode-tile-to-west}}}
\newcommand{\genenctoeast}{\ensuremath{\mbox{\sc Generate-encode-tile-to-east}}}

\newcommand{\gendectowest}{\ensuremath{\mbox{\sc Generate-decode-tile-to-west}}}
\newcommand{\gendectoeast}{\ensuremath{\mbox{\sc Generate-decode-tile-to-east}}}

\newcommand{\genconn}{\ensuremath{\mbox{\sc Generate-connection-tile}}}
\newcommand{\genencfix}{\ensuremath{\mbox{\sc Generate-encode-tile-fixed}}}

\newcommand{\gendectileset}{\ensuremath{\mbox{\sc Generate-decode-tile-set}}}
\newcommand{\genenctileset}{\ensuremath{\mbox{\sc Generate-encode-tile-set}}}

\newcommand{\genenctilesettowest}{\ensuremath{\mbox{\sc Generate-encode-tile-set-to-west}}}
\newcommand{\genenctilesettoeast}{\ensuremath{\mbox{\sc Generate-encode-tile-set-to-east}}}
\newcommand{\genenctilesetothers}{\ensuremath{\mbox{\sc Generate-encode-tile-set-others}}}

\subsubsection{Categorize Zig-Zag Tiles}
A simulator is used to detect the tile types in this algorithm.
The only restriction is that all of the tile types should be occured at least once in resonable steps during simulating.
The algorithm will return right after all of the tile types are detected.
The algorithm is listed in Algorithm \ref{alg:categorizetiles}.

\begin{algorithm}
\caption{$\categorizetiles()$}
\label{alg:categorizetiles}
\linesnumberedhidden
\SetKwData{tilein}{$T_{2d,2t}'$}
\SetKwData{tileout}{$T_{2d,2t}$}
\KwIn{\tilein, Un-categorized tile set at temperature 2 in 2D}
\KwOut{\tileout, Categorized tile set at temperature 2 in 2D}
    \tileout $ \gets \phi $\;
    numTypes $ \gets 0 $ \;

    \While { numTypes $ < \vert \tilein \vert $ } {
        select one of the tile $t' \in T_{2d,2t}'$ which can attach to the supertile at position $pos$ \;
        record this step ($t'$, $pos$) and append to array $steps$ \;
        \showln \If{the type of $t'$ is unknown}{
            adjIdx $ \gets 0$ \;
            numTypes $ \gets $ numTypes + 1 \;
            \For{$i \gets (\vert steps \vert - 1) $ \KwTo $1$}{
                \If{$pos$ is adjacent to $steps[i].pos$ and the glue between them are the same}{
                    adjIdx $ \gets $ i \;
                    adjTile $ \gets $ $steps[i].t$ \;
                    dir $ \gets $ the side of $t'$ that is adjacent to $steps[i].t$\;
                    Goto end of this {\rm for} loop \;
                }
            }
            \showln \If{adjIdx is between $(\vert steps \vert - 1) $ and $1$}{
                \Switch {the value of dir}{
                    \showln \Case {EAST}{
                        \eIf {strength of $t'.g_e > 1$}{
                            \lIf {strength of $t'.g_n > 1$}{
                                The type of $t'$ $ \gets $ SWN2; \tileout $\gets $ \tileout $\cup \{t'\}$ \;
                            }\lElseIf{strength of $t'.g_w > 1$}{
                                The type of $t'$ $ \gets $ FW; \tileout $\gets $ \tileout $\cup \{t'\}$ \;
                            }
                        }{
                            \lIf {strength of $t'.g_n > 1$}{
                                The type of $t'$ $ \gets $ DWN2; \tileout $\gets $ \tileout $\cup \{t'\}$ \;
                            }\lElseIf{strength of $t'.g_w > 1$}{
                                The type of $t'$ $ \gets $ DWW2; \tileout $\gets $ \tileout $\cup \{t'\}$ \;
                            }\lElse{
                                The type of $t'$ $ \gets $ DWW1; \tileout $\gets $ \tileout $\cup \{t'\}$ \;
                            }
                        }
                    }
                    \showln \Case {SOUTH}{
                        \Switch {the type of adjTile}{
                            \uCase {DEN2 or SEN2}{
                                \lIf {strength of $t'.g_e > 1$}{
                                    The type of $t'$ $ \gets $ TEE2; \tileout $\gets $ \tileout $\cup \{t'\}$ \;
                                }\lElseIf{strength of $t'.g_w > 1$}{
                                    The type of $t'$ $ \gets $ TEW2; \tileout $\gets $ \tileout $\cup \{t'\}$ \;
                                }\lElse{
                                    The type of $t'$ $ \gets $ TEW1; \tileout $\gets $ \tileout $\cup \{t'\}$ \;
                                }
                            }
                            \uCase {DWN2 or SWN2}{
                                \lIf {strength of $t'.g_e > 1$}{
                                    The type of $t'$ $ \gets $ TWE2; \tileout $\gets $ \tileout $\cup \{t'\}$ \;
                                }\lElseIf{strength of $t'.g_w > 1$}{
                                    The type of $t'$ $ \gets $ TWW2; \tileout $\gets $ \tileout $\cup \{t'\}$ \;
                                }\lElse{
                                    The type of $t'$ $ \gets $ TWE1; \tileout $\gets $ \tileout $\cup \{t'\}$ \;
                                }
                            }
                        }
                    }
                    \showln \Case {WEST}{
                        \eIf{strength of $t'.g_w > 1$}{
                            \eIf{strength of $t'.g_n > 1$}{
                                The type of $t'$ $ \gets $ SEN2; \tileout $\gets $ \tileout $\cup \{t'\}$ \;
                            }{
                                The type of $t'$ $ \gets $ FE; \tileout $\gets $ \tileout $\cup \{t'\}$ \;
                            }
                        }{
                            \lIf {strength of $t'.g_n > 1$}{
                                The type of $t'$ $ \gets $ DEN2; \tileout $\gets $ \tileout $\cup \{t'\}$ \;
                            }\lElseIf{strength of $t'.g_e > 1$}{
                                The type of $t'$ $ \gets $ DEE2; \tileout $\gets $ \tileout $\cup \{t'\}$ \;
                            }\lElse{
                                The type of $t'$ $ \gets $ DEE1; \tileout $\gets $ \tileout $\cup \{t'\}$ \;
                            }
                        }
                    }
                    \lOther {
                        Error, Ignored \;
                    }
                }
            }
        }
    }
    \KwRet{\tileout }\;
\end{algorithm}

\subsubsection{Zig-Zag in 3D}
The algorithms to create the tile sets in 3D are listed in Algorithm
\ref{alg:tsconvallcate},
\ref{alg:gendectileset},
\ref{alg:gendectowest},
\ref{alg:gendectoeast},
\ref{alg:genenctileset},
\ref{alg:genenctilesettowest},
\ref{alg:genenctilesettoeast},
\ref{alg:genenctilesetothers},
\ref{alg:genenctowest},
\ref{alg:genenctoeast},
\ref{alg:genconn},
\ref{alg:genencfix}.
The number of the binary bits used in the codes is denoted by maxbits.

\begin{algorithm}
\caption{$\tsconvallcate()$}
\label{alg:tsconvallcate}
\linesnumberedhidden
\SetKwData{encset}{$\mathcal{E}$}
\SetKwData{strengthset}{$S_{2d,2t}$}
\SetKwData{tilein}{$T_{2d,2t}$}
\SetKwData{tileout}{$T_{3d,1t}$}

\KwIn{\tilein, Categorized tile set at temperature 2 in 2D}
\KwOut{\tileout, Tile set at temperature 1 in 3D}
    \tcc{$$\tilein=T_{DWW1} \cup ~ T_{DWW2} \cup ~ T_{TEW1} \cup ~ T_{TEW2} \cup ~ T_{TWW2} $$
              $$\cup ~ T_{DEE1} \cup ~ T_{DEE2} \cup ~ T_{TWE1} \cup ~ T_{TWE2} \cup ~ T_{TEE2} $$
              $$\cup ~ T_{FE} \cup ~ T_{FW} \cup ~ T_{SWN2} \cup ~ T_{SEN2} \cup ~ T_{DWN2} \cup ~ T_{DEN2} $$}

    $\tileout \gets \phi $ \tcc*[r]{The tile set in 3D temperature 1}
    $G_{ns} \gets \phi $ \tcc*[r]{The glues at the north and south sides of the tile}
    \strengthset $ \gets $ \{$s_i|s_i$ = Strength of all of the glues of $t_j$, $t_j \in \tilein$ \} \;

    \ForEach{ $ t_i \in \tilein $ }{
        \If{ 1 = $s_{t_i.g_n}$ }{
            $G_{ns} \gets G_{ns} ~ \cup ~ \{ t_i.g_n \}$\;
        }
        \If{ 1 = $s_{t_i.g_s}$ }{
            $G_{ns} \gets G_{ns} ~ \cup ~ \{ t_i.g_s \}$\;
        }
    }

    Encode the glues in set $G_{ns}$ by binary codes $e_i|i \in G_{ns} $\;

    $\encset \gets \{ e_i|i \in G_{ns} \} $ \tcc*[r]{\encset contains all of the code of glue $\in G_{ns}$.}

    maxbits $\gets \lceil \log \vert G_{ns} \vert  \rceil $\;

    $\tileout \gets \tileout \cup ~ \gendectileset (T_{dirleft} \cup T_{dirright}, \strengthset, \encset$, maxbits) \;
    $\tileout \gets \tileout \cup ~ \genenctileset (T_{dirleft} \cup T_{dirright}, \strengthset, \encset$, maxbits) \;

    \KwRet{$\tileout$}\;
\end{algorithm}

\begin{algorithm}
\caption{$\gendectileset()$}
\label{alg:gendectileset}
\linesnumberedhidden
\SetKwData{encset}{$\mathcal{E}$}
\SetKwData{strengthset}{$S_{2d,2t}$}
\SetKwData{tilein}{$T_{2d,2t}$}
\SetKwData{tileout}{$T_{3d,1t}'$}
\KwIn{\tilein, Categorized tile set at temperature 2 in 2D\\
\strengthset, all of the strength of glue in \tilein\\
\encset, all of the code of glue $\in G_{ns}$\\
maxbits, the number of binary bits to encode all of the glues $\in G_{ns}$
}
\KwOut{\tileout, Tile set at temperature 1 in 3D}

$\tileout \gets \phi $\;

$ G_{2w} \gets \phi $ \tcc*[r]{The glue set of input from east to west}
$ G_{2e} \gets \phi $ \tcc*[r]{The glue set of input from west to east}

\ForEach{ $ t_i \in T_{DWW1} \cup ~ T_{DWW2} \cup ~ T_{DWN2} $ }{
    $ G_{2w} \gets G_{2w} \cup ~ \{t_i.g_e\} $\;
}

\ForEach{ $ t_i \in T_{DEE1} \cup ~ T_{DEE2} \cup ~ T_{DEN2} $ }{
    $ G_{2e} \gets G_{2e} \cup ~ \{t_i.g_w\} $\;
}

\ForEach{ $ g_i \in G_{2w} $ }{
    $G_{w,i} \gets \phi $\;
    \ForEach{ $ t_j \in T_{DWW1} \cup ~ T_{DWW2} \cup ~ T_{DWN2} $ }{
        \If{ $t_j.g_e = g_i$ and $s_{t_j.g_s} < 2$}{
            \tcc{$s_{t_j.g_s}$ is the strength of $g_s$ of tile $t_j$, $s_{t_j.g_s} \in \strengthset$}
            $G_{w,i} \gets G_{w,i} \cup ~ \{ t_j.g_s \}$
        }
    }
    $\tileout = \tileout \cup ~ \gendectowest(g_i, G_{w,i}, \encset, maxbits) $\;
}

\ForEach{ $ g_i \in G_{2e} $ }{
    $G_{e,i}= \phi $\;
    \ForEach{ $ t_j \in T_{DEE1} \cup ~ T_{DEE2} \cup ~ T_{DEN2} $ }{
        \If{ $t_j.g_w = g_i$ and $s_{t_j.g_s} < 2$}{
            \tcc{$s_{t_j.g_s}$ is the strength of $g_s$ of tile $t_j$, $s_{t_j.g_s} \in \strengthset$}
            $G_{e,i} \gets G_{e,i} \cup ~ \{ t_j.g_s \}$
        }
    }
    $\tileout \gets \tileout \cup ~ \gendectoeast(g_i, G_{e,i}, \encset, maxbits) $\;
}

\KwRet{\tileout}\;

\end{algorithm}

\begin{algorithm}
\caption{$\gendectowest()$}
\label{alg:gendectowest}
\linesnumberedhidden
\SetKwData{encset}{$\mathcal{E}$}
\SetKwData{tileout}{$T_{3d,1t}'$}
\KwIn{
$g_{in}$, the input glue of the current tile set \\
$G_{w,i}$, the glue set to be the output glues of current tile set\\
\encset, all of the code of glue $\in G_{ns}$\\
maxbits, the number of binary bits to encode all of the glues $\in G_{ns}$
}
\KwOut{\tileout, Tile set at temperature 1 in 3D}
    $ \tileout \gets \phi $\;
    \tcc{Contruct a binary tree according to the encoding code of each items in $G_{w,i}$.}
    \ForEach{ $g_i \in G_{w,i}$ }{
        curnode $\gets $ root\;
        \ForEach{ bit of $e_{g_i}$ from MSB to LSB }{
            \If{ bit = 1 }{
                \If { curnode have no right child } {
                    create right child of the curnode\;
                }
                curnode $\gets $ curnode.right\_child\;
            }\Else{
                \If { curnode have no left child } {
                    create left child of the curnode\;
                }
                curnode $\gets $ curnode.left\_child\;
            }
        }
    }
    Traversal the tree by pre-order algorithm: Part of the tiles set will be generated and saved to \tileout in each visitation. The input glue of the tile set is $g_{in}$.
    See Figure \ref{fig:dec3dleft}\;
    \KwRet{\tileout}\;
\end{algorithm}

\begin{algorithm}
\caption{$\gendectoeast()$}
\label{alg:gendectoeast}
\linesnumberedhidden
\SetKwData{encset}{$\mathcal{E}$}
\SetKwData{tileout}{$T_{3d,1t}'$}
\KwIn{
$g_{in}$, the input glue of the current tile set \\
$G_{e,i}$, the glue set to be the output glues of current tile set\\
\encset, all of the code of glue $\in G_{ns}$\\
maxbits, the number of binary bits to encode all of the glues $\in G_{ns}$
}
\KwOut{\tileout, Tile set at temperature 1 in 3D}
    $ \tileout \gets \phi $\;
    \tcc{Contruct a binary tree according to the encoding of each items in $G_{e,i}$.}
    \ForEach{ $g_i \in G_{e,i}$ }{
        curnode $\gets $ root\;
        \ForEach{ bit of $e_{g_i}$ from MSB to LSB }{
            \If{ bit = 1 }{
                \If { curnode have no right child } {
                    create right child of the curnode\;
                }
                curnode $\gets $ curnode.right\_child\;
            }\Else{
                \If { curnode have no left child } {
                    create left child of the curnode\;
                }
                curnode $\gets $ curnode.left\_child\;
            }
        }
    }
    Traversal the tree by pre-order algorithm: Part of the tiles set will be generated and saved to \tileout in each visitation. The input glue of the tile set is $g_{in}$.
    See Figure \ref{fig:dec3d2right}\;
    \KwRet{\tileout}\;
\end{algorithm}

\begin{figure}[h]\centering
 \includegraphics[height=0.3\textheight]{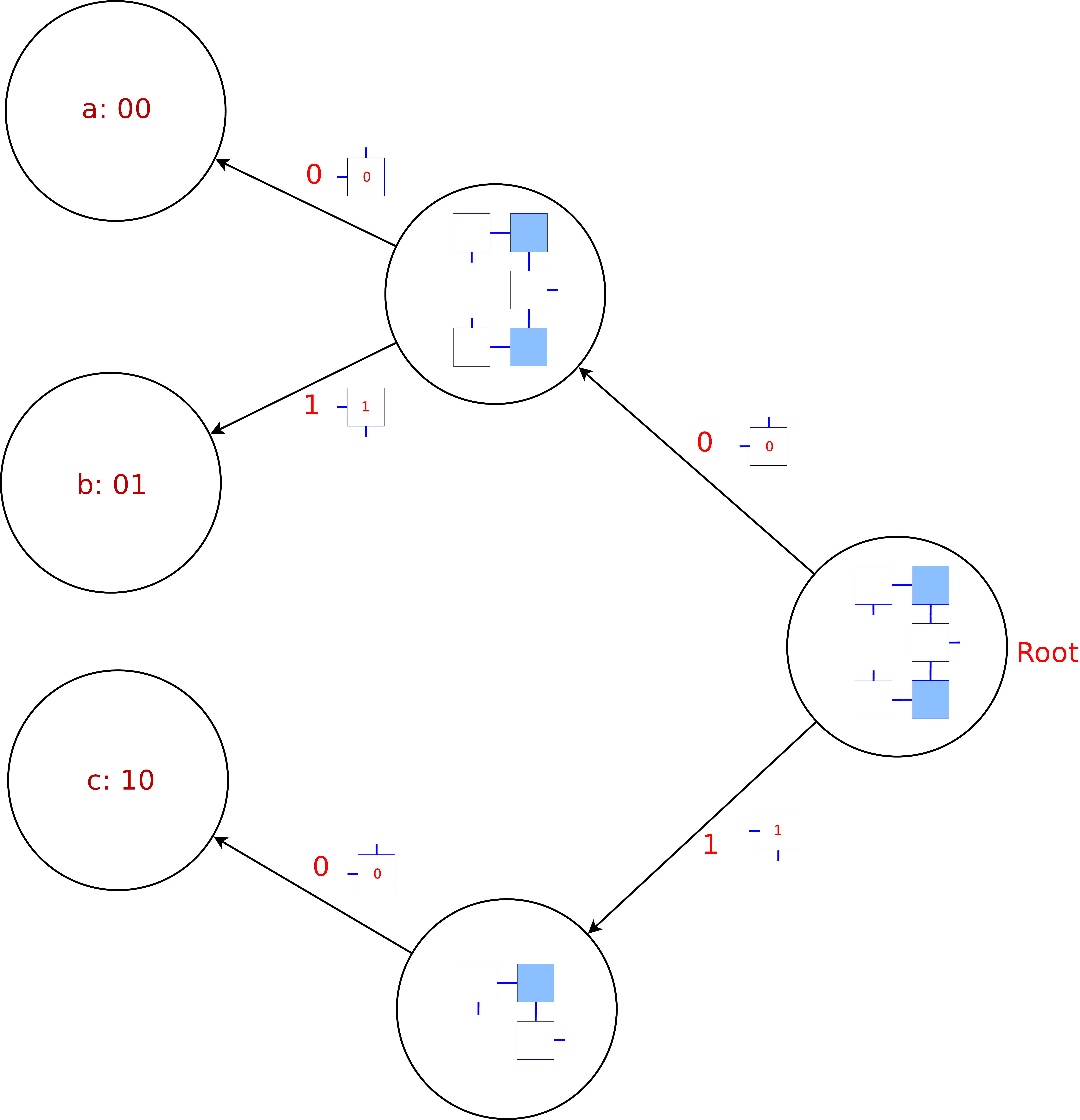}
 \caption{The logic binary tree for constructing decoding tile set (direction left).}\label{fig:dec3dleft}
\end{figure}

\begin{figure}[h]\centering
 \includegraphics[height=0.3\textheight]{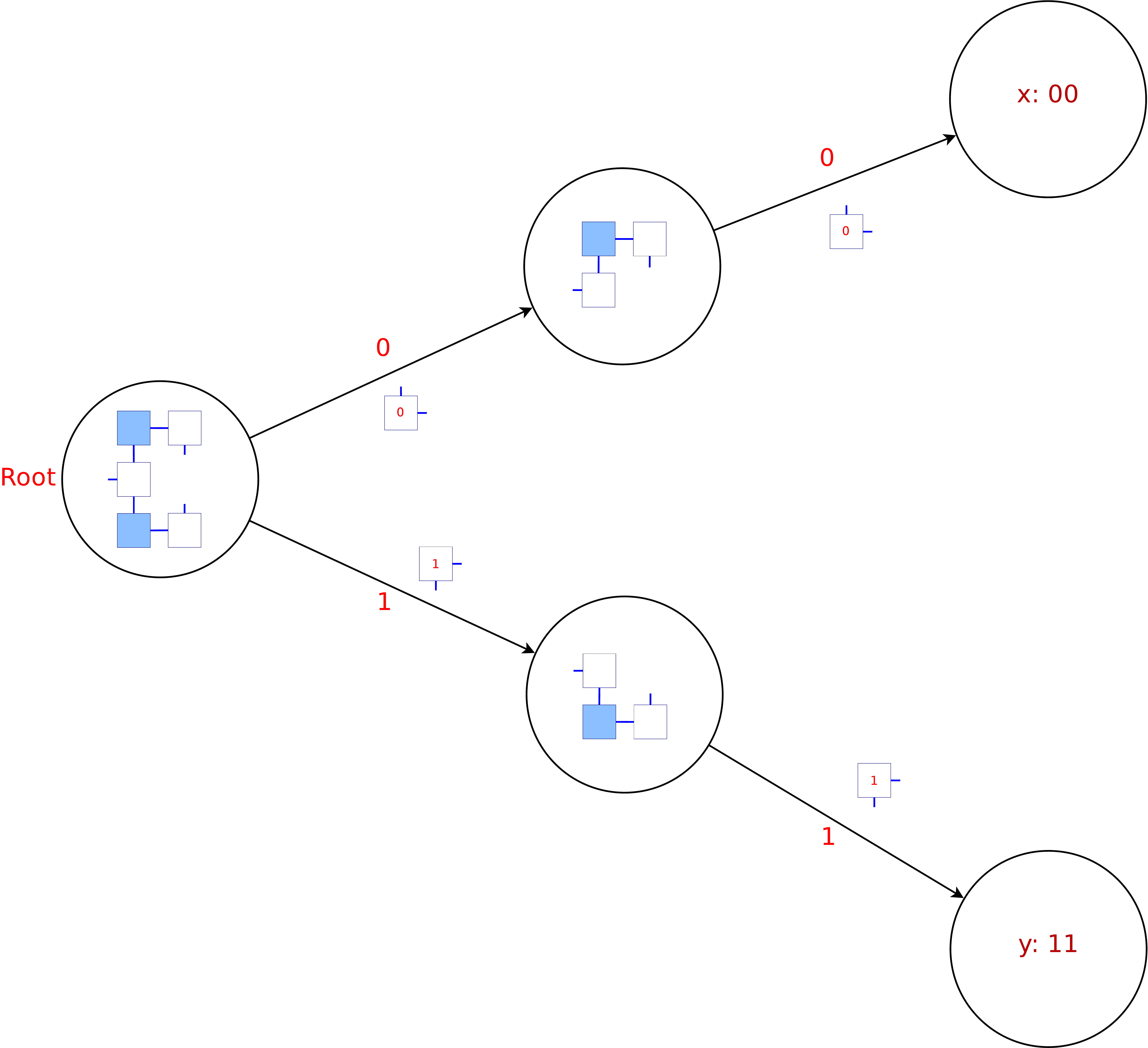}
 \caption{The logic binary tree for constructing decoding tile set (direction right).}\label{fig:dec3d2right}
\end{figure}

\begin{algorithm}
\caption{$\genenctileset()$}
\label{alg:genenctileset}
\linesnumberedhidden
\SetKwData{encset}{$\mathcal{E}$}
\SetKwData{tilein}{$T_{2d,2t}$}
\SetKwData{tileout}{$T_{3d,1t}'$}
\KwIn{\tilein, Categorized tile set at temperature 2 in 2D\\
\encset, all of the code of glue $\in G_{ns}$\\
maxbits, the number of binary bits to encode all of the glues $\in G_{ns}$
}
\KwOut{\tileout, Tile set at temperature 1 in 3D}

    $\tileout = \phi $\;

    $\tileout \gets \tileout \cup ~ \genenctilesettoeast(\tilein, \encset, maxbits) $\;
    $\tileout \gets \tileout \cup ~ \genenctilesettowest(\tilein, \encset, maxbits) $\;
    $\tileout \gets \tileout \cup ~ \genenctilesetothers(\tilein, \encset, maxbits) $\;

\KwRet{\tileout}\;
\end{algorithm}

\begin{algorithm}
\caption{$\genenctilesettowest()$}
\label{alg:genenctilesettowest}
\linesnumberedhidden
\SetKwData{encset}{$\mathcal{E}$}
\SetKwData{tilein}{$T_{2d,2t}$}
\SetKwData{tileout}{$T_{3d,1t}'$}
\KwIn{\tilein, Categorized tile set at temperature 2 in 2D\\
\encset, all of the code of glue $\in G_{ns}$\\
maxbits, the number of binary bits to encode all of the glues $\in G_{ns}$
}
\KwOut{\tileout, Tile set at temperature 1 in 3D}

    $\tileout = \phi $\;

    \ForEach{ $ t_i \in T_{DWW1} $ }{
        $g_{in} \gets (t_i.g_e, t_i.g_s)$\;
        $g_{out} \gets (t_i.g_w,-)$\;
        $\tileout \gets \tileout \cup ~ \genenctowest (DWW1, g_{in}, g_{out}, e_{t_i.g_n}, maxbits)$ \tcc*[r]{$e_{t_i.g_n}$ is the encoding code of glue $g_n$ of tile $t_i$ }
    }

    \ForEach{ $ t_i \in T_{TEW1} $ }{
        $g_{in} \gets t_i.g_s$\;
        $g_{out} \gets (t_i.g_w,-)$\;
        $\tileout \gets \tileout \cup ~ \genenctowest (TEW1, g_{in}, g_{out}, e_{t_i.g_n}, maxbits)$\;
    }

    \ForEach{ $ t_i \in T_{DWW2} $ }{
        $g_{in} \gets (t_i.g_e, t_i.g_s)$\;
        $g_{out} \gets t_i.g_w$\;
        $\tileout \gets \tileout \cup ~ \genenctowest (DWW2, g_{in}, g_{out}, e_{t_i.g_n}, maxbits)$\;
    }

    \ForEach{ $ t_i \in T_{TWW2} $ }{
        $g_{in} \gets t_i.g_s$\;
        $g_{out} \gets t_i.g_w$\;
        $\tileout \gets \tileout \cup ~ \genenctowest (TWW2, g_{in}, g_{out}, e_{t_i.g_n}, maxbits)$\;
    }

    \ForEach{ $ t_i \in T_{TEW2} $ }{
        $g_{in} \gets t_i.g_s$\;
        $g_{out} \gets t_i.g_w$\;
        $\tileout \gets \tileout \cup ~ \genenctowest (TEW2, g_{in}, g_{out}, e_{t_i.g_n}, maxbits)$\;
    }

    \KwRet{\tileout}\;
\end{algorithm}

\begin{algorithm}
\caption{$\genenctilesettoeast()$}
\label{alg:genenctilesettoeast}
\linesnumberedhidden
\SetKwData{encset}{$\mathcal{E}$}
\SetKwData{tilein}{$T_{2d,2t}$}
\SetKwData{tileout}{$T_{3d,1t}'$}
\KwIn{\tilein, Categorized tile set at temperature 2 in 2D\\
\encset, all of the code of glue $\in G_{ns}$\\
maxbits, the number of binary bits to encode all of the glues $\in G_{ns}$
}
\KwOut{\tileout, Tile set at temperature 1 in 3D}

    $\tileout = \phi $\;

    \ForEach{ $ t_i \in T_{DEE1} $ }{
        $g_{in} \gets (t_i.g_w, t_i.g_s)$\;
        $g_{out} \gets (t_i.g_e,-)$\;
        $\tileout \gets \tileout \cup ~ \genenctoeast (DEE1, g_{in}, g_{out}, e_{t_i.g_n}, maxbits)$ \tcc*[r]{$e_{t_i.g_n}$ is the encoding code of glue $g_n$ of tile $t_i$ }
    }

    \ForEach{ $ t_i \in T_{TWE1} $ }{
        $g_{in} \gets t_i.g_s$\;
        $g_{out} \gets (t_i.g_e,-)$\;
        $\tileout \gets \tileout \cup ~ \genenctoeast (TWE1, g_{in}, g_{out}, e_{t_i.g_n}, maxbits)$\;
    }

    \ForEach{ $ t_i \in T_{DEE2} $ }{
        $g_{in} \gets (t_i.g_w, t_i.g_s)$\;
        $g_{out} \gets t_i.g_e$\;
        $\tileout \gets \tileout \cup ~ \genenctoeast (DEE2, g_{in}, g_{out}, e_{t_i.g_n}, maxbits)$\;
    }

    \ForEach{ $ t_i \in T_{TEE2} $ }{
        $g_{in} \gets t_i.g_s$\;
        $g_{out} \gets t_i.g_e$\;
        $\tileout \gets \tileout \cup ~ \genenctoeast (TEE2, g_{in}, g_{out}, e_{t_i.g_n}, maxbits)$\;
    }

    \ForEach{ $ t_i \in T_{TWE2} $ }{
        $g_{in} \gets t_i.g_s$\;
        $g_{out} \gets t_i.g_e$\;
        $\tileout \gets \tileout \cup ~ \genenctoeast (TWE2, g_{in}, g_{out}, e_{t_i.g_n}, maxbits)$\;
    }

    \KwRet{\tileout}\;
\end{algorithm}

\begin{algorithm}
\caption{$\genenctilesetothers()$}
\label{alg:genenctilesetothers}
\linesnumberedhidden
\SetKwData{encset}{$\mathcal{E}$}
\SetKwData{tilein}{$T_{2d,2t}$}
\SetKwData{tileout}{$T_{3d,1t}'$}
\KwIn{\tilein, Categorized tile set at temperature 2 in 2D\\
\encset, all of the code of glue $\in G_{ns}$\\
maxbits, the number of binary bits to encode all of the glues $\in G_{ns}$
}
\KwOut{\tileout, Tile set at temperature 1 in 3D}

    $\tileout = \phi $\;

    \ForEach{ $ t_i \in T_{DWN2} $ }{
        $g_{in} \gets (t_i.g_e, t_i.g_s)$\;
        $g_{out} \gets t_i.g_n$\;
        $\tileout \gets \tileout \cup ~ \genconn (DWN2, g_{in}, g_{out}, maxbits)$\;
    }

    \ForEach{ $ t_i \in T_{DEN2} $ }{
        $g_{in} \gets (t_i.g_w, t_i.g_s)$\;
        $g_{out} \gets t_i.g_n$\;
        $\tileout \gets \tileout \cup ~ \genconn (DEN2, g_{in}, g_{out}, maxbits)$\;
    }

    \ForEach{ $ t_i \in T_{SWN2} $ }{
        $g_{in} \gets t_i.g_e$\;
        $g_{out} \gets t_i.g_n$\;
        $\tileout \gets \tileout \cup ~ \genconn (SWN2, g_{in}, g_{out}, maxbits)$\;
    }

    \ForEach{ $ t_i \in T_{SEN2} $ }{
        $g_{in} \gets t_i.g_w$\;
        $g_{out} \gets t_i.g_n$\;
        $\tileout \gets \tileout \cup ~ \genconn (SEN2, g_{in}, g_{out}, maxbits)$\;
    }

    \ForEach{ $ t_i \in T_{FW} $ }{
        $g_{in} \gets t_i.g_e$\;
        $g_{out} \gets t_i.g_w$\;
        $\tileout \gets \tileout \cup ~ \genencfix (FW, g_{in}, g_{out}, e_{t_i.g_n}, maxbits)$ \tcc*[r]{$e_{t_i.g_n}$ is the encoding code of glue $g_n$ of tile $t_i$ }
    }

    \ForEach{ $ t_i \in T_{FE} $ }{
        $g_{in} \gets t_i.g_w$\;
        $g_{out} \gets t_i.g_e$\;
        $\tileout \gets \tileout \cup ~ \genencfix (FE, g_{in}, g_{out}, e_{t_i.g_n}, maxbits)$\;
    }

    \KwRet{\tileout}\;
\end{algorithm}

\begin{algorithm}
\caption{$\genenctowest()$}
\label{alg:genenctowest}
\linesnumberedhidden
\SetKwData{tileout}{$T_{3d,1t}'$}
\KwIn{tiletype, the type of the tile\\
$g_{in}$, the input glue\\
$g_{out}$, the output glue\\
$e_n$, the code of the north glue\\
maxbits, the number of binary bits to encode all of the glues $\in G_{ns}$
}
\KwOut{\tileout, Tile set at temperature 1 in 3D}

    $ \tileout \gets \phi $\;

    \tcc{Generate {\bf \emph{distinct}} tiles as showed in Table \ref{tab:zigzagmapping}.}

    \If{tiletype = DWW1}{
        Generate the tiles $t_0', t_1', t_2', t_3', t_4', t_5'$ and put it to \tileout\;
        \tcc{The $t_0', t_1', t_2', t_3', t_4', t_5'$ are denoted separately by t\_0, t\_1, t\_2, t\_3, t\_4, and t\_5 in the figure of DWW1 in the Table \ref{tab:zigzagmapping}. The glue is $g_{out}$ at the west of the tile $t_0'$.}
    }\ElseIf{tiletype = DWW2}{
        Generate the tiles $t_0', t_1'$ and put it to \tileout\;
    }\ElseIf{tiletype = TEW1}{
        Generate the tiles $t_0', t_1', t_2', t_3', t_4', t_5'$ and put it to \tileout\;
    }\ElseIf{tiletype = TEW2 or tiletype = TWW2}{
        Generate the tiles $t_0', t_1'$ and put it to \tileout\;
    }

    Generate the tiles in the \emph{Encoding Area} and put it to \tileout\;
    \tcc{The positions of the tiles are depend on the $e_n$;
    Each bit of the $e_n$ are encoded by two tiles in the plane $z=1$;
    The tiles will place at the position '1'(the dotted line denoted by '1' in the figures if the bit is $1$,
    while the tiles will place at position '0' if the bit is $0$;
    The encoding of the most significant bit (MSB) of the $e_n$ is place at the left side of the \emph{Encoding Area},
    and the least siginificant bit(LSB) of the $e_n$ is place at the right side of the \emph{Encoding Area}.
    All of the encoded tiles are connected by the tiles in the plane $z=0$.}

    \If{tiletype = DWW1 or tiletype = DWW2 or tiletype = TWW2}{
        Generate the tiles between glue A and glue $g_{in}$, put those tiles to \tileout\;
    }\ElseIf{tiletype = TEW1 or tiletype = TEW2}{
        Generate the tile with glue $g_{in}$, put it to \tileout\;
    }

    \KwRet{\tileout}\;
\end{algorithm}

\begin{algorithm}
\caption{$\genenctoeast()$}
\label{alg:genenctoeast}
\linesnumberedhidden
\SetKwData{tileout}{$T_{3d,1t}'$}
\KwIn{tiletype, the type of the tile\\
$g_{in}$, the input glue\\
$g_{out}$, the output glue\\
$e_n$, the code of the north glue\\
maxbits, the number of binary bits to encode all of the glues $\in G_{ns}$
}
\KwOut{\tileout, Tile set at temperature 1 in 3D}

    $ \tileout \gets \phi $\;

    \If{tiletype = DEE1}{
        Generate the tiles $t_0', t_1', t_2', t_3', t_4', t_5'$ and put it to \tileout\;
        \tcc{The $t_0', t_1', t_2', t_3', t_4', t_5'$ are denoted separately by t\_0, t\_1, t\_2, t\_3, t\_4, and t\_5 in the figures of Table \ref{tab:zigzagmapping}. The glue is $g_{out}$ at the east of the tile $t_0'$.}
    }\ElseIf{tiletype = DEE2}{
        Generate the tiles $t_0', t_1'$ and put it to \tileout\;
    }\ElseIf{tiletype = TWE1}{
        Generate the tiles $t_0', t_1', t_2', t_3', t_4', t_5'$ and put it to \tileout\;
    }\ElseIf{tiletype = TWE2 or tiletype = TEE2}{
        Generate the tiles $t_0', t_1'$ and put it to \tileout\;
    }

    Generate the tiles in the \emph{Encoding Area} and put it to \tileout\;

    \If{tiletype = DEE1 or tiletype = DEE2 or tiletype = TEE2}{
        Generate the tiles between glue A and glue $g_{in}$, put those tiles to \tileout\;
    }\ElseIf{tiletype = TWE1 or TWE2}{
        Generate the tile with glue $g_{in}$, put it to \tileout\;
    }

    \KwRet{\tileout}\;
\end{algorithm}

\begin{algorithm}
\caption{$\genconn()$}
\label{alg:genconn}
\linesnumberedhidden
\SetKwData{tileout}{$T_{3d,1t}'$}
\KwIn{tiletype, the type of the tile\\
$g_{in}$, the input glue\\
$g_{out}$, the output glue\\
$e_n$, the code of the north glue\\
maxbits, the number of binary bits to encode all of the glues $\in G_{ns}$
}
\KwOut{\tileout, Tile set at temperature 1 in 3D}

    $ \tileout \gets \phi $\;

    \If{tiletype = SWN2 or tiletype = SEN2}{
        The number of tiles to be generated at the bottom of dotted box in the figure is $(maxbits \times 2)$.
    }\ElseIf{tiletype = DWN2 or tiletype = DEN2}{
        Generate the tiles showed as the figures in the Table \ref{tab:zigzagmapping}, with the input glue $g_{in}$ and output glue $g_{out}$. Put all of the tiles to \tileout.
    }

    \KwRet{\tileout}\;
\end{algorithm}

\begin{algorithm}
\caption{$\genencfix()$}
\label{alg:genencfix}
\linesnumberedhidden
\SetKwData{tileout}{$T_{3d,1t}'$}
\KwIn{tiletype, the type of the tile\\
$g_{in}$, the input glue\\
$g_{out}$, the output glue\\
$e_n$, the code of the north glue\\
maxbits, the number of binary bits to encode all of the glues $\in G_{ns}$
}
\KwOut{\tileout, Tile set at temperature 1 in 3D}

    $ \tileout \gets \phi $\;

    Generate the tiles with the input glue $g_{in}$\;
    Generate the tiles in the \emph{Encoding Area} and put it to \tileout\;
    Generate the tiles with the output glue $g_{out}$\;

    \KwRet{\tileout}\;
\end{algorithm}

\subsubsection{Probabilistic Zig-Zag in 2D}
The algorithms for converting the zig-zag from temperature $\tau=2$ to temperature $\tau=1$ are similar to that in 3D.
Using the same algorithms to encode all of the glues at the north or south of tiles with strength 1
(See Algorithm \ref{alg:tsconvallcate}, \ref{alg:gendectileset}).

The decoding tile sets are different from the zig-zag in 3D.
The parameter $K$ is introduced in the probabilistic zig-zag tile set.
Figure \ref{fig:dec2d1bit2left}, \ref{fig:dec2d1bit2right} shows the tile set for detecting one bit of the code
by using $K (=4)$ groups of the detect tile set.
Figure \ref{fig:dec2d2left} shows one of the complete decoding tile sets
which have similar function as showed in Figure \ref{fig:dec3dleft}.

The \emph{Encoding Area} is a bit different from that in 3D.
The length for each bit of the code in the \emph{Encoding Area} is depend on the parameter $K$.
The length of the \emph{Encoding Area} will be $2K\times maxbits$.
The mapping between the zig-zag and probabilistic zig-zag for each of the tile types is showed in Table \ref{tab:zigzagmapping},
the implemention use the similar algorithms as that used in 3D.

The success ratio of constructing zig-zag structure depends on the parameter $K$, but we noticed that it also depends on the number of zero in encoding code, because there exist false positive in detecting the zero bits of the encoding codes. We can select the codes which contain many one bits for the encoding code to reduce the error ratio.

The algorithms for probabilistic zig-zag in 2D are omited.


\begin{figure}[h]\centering
 \includegraphics[width=0.75\textwidth]{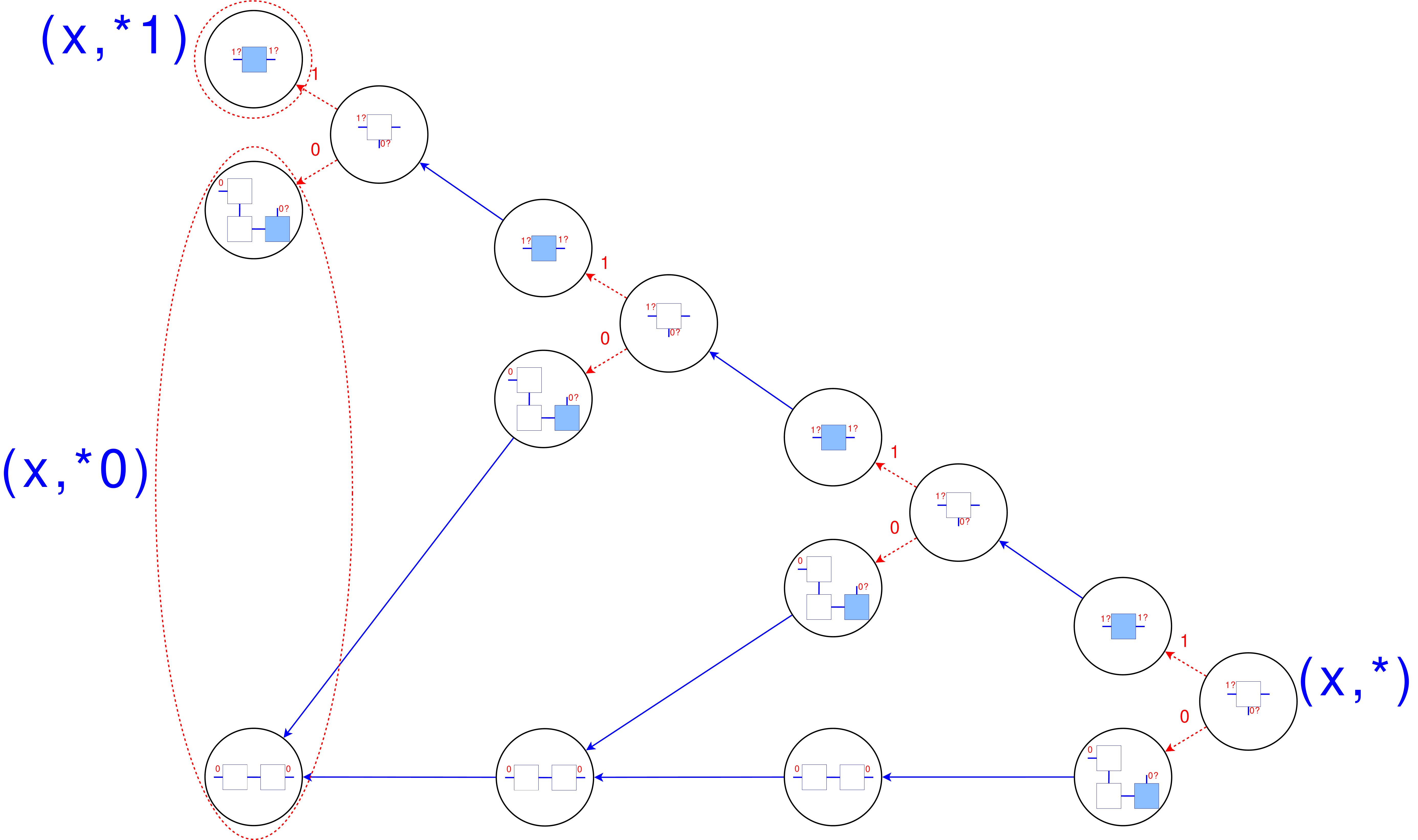}
 \caption{The tile set to decode one bit of the code (Direction left, K=4).}\label{fig:dec2d1bit2left}
\end{figure}

\begin{figure}[h]\centering
 \includegraphics[width=0.75\textwidth]{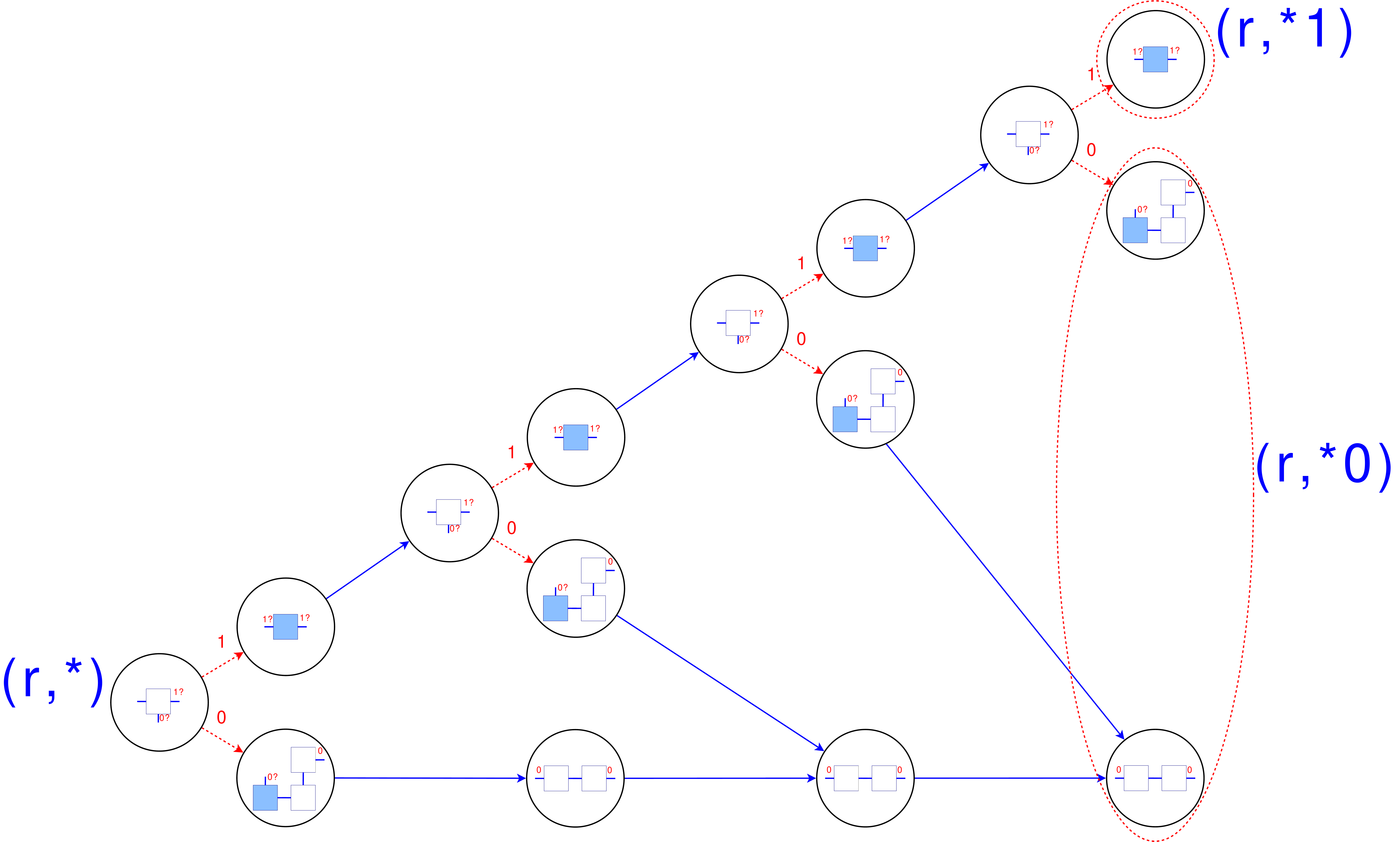}
 \caption{The tile set to decode one bit of the code (Direction right, K=4).}\label{fig:dec2d1bit2right}
\end{figure}

\begin{figure}[h]\centering
  \includegraphics[width=\textheight , angle=90]{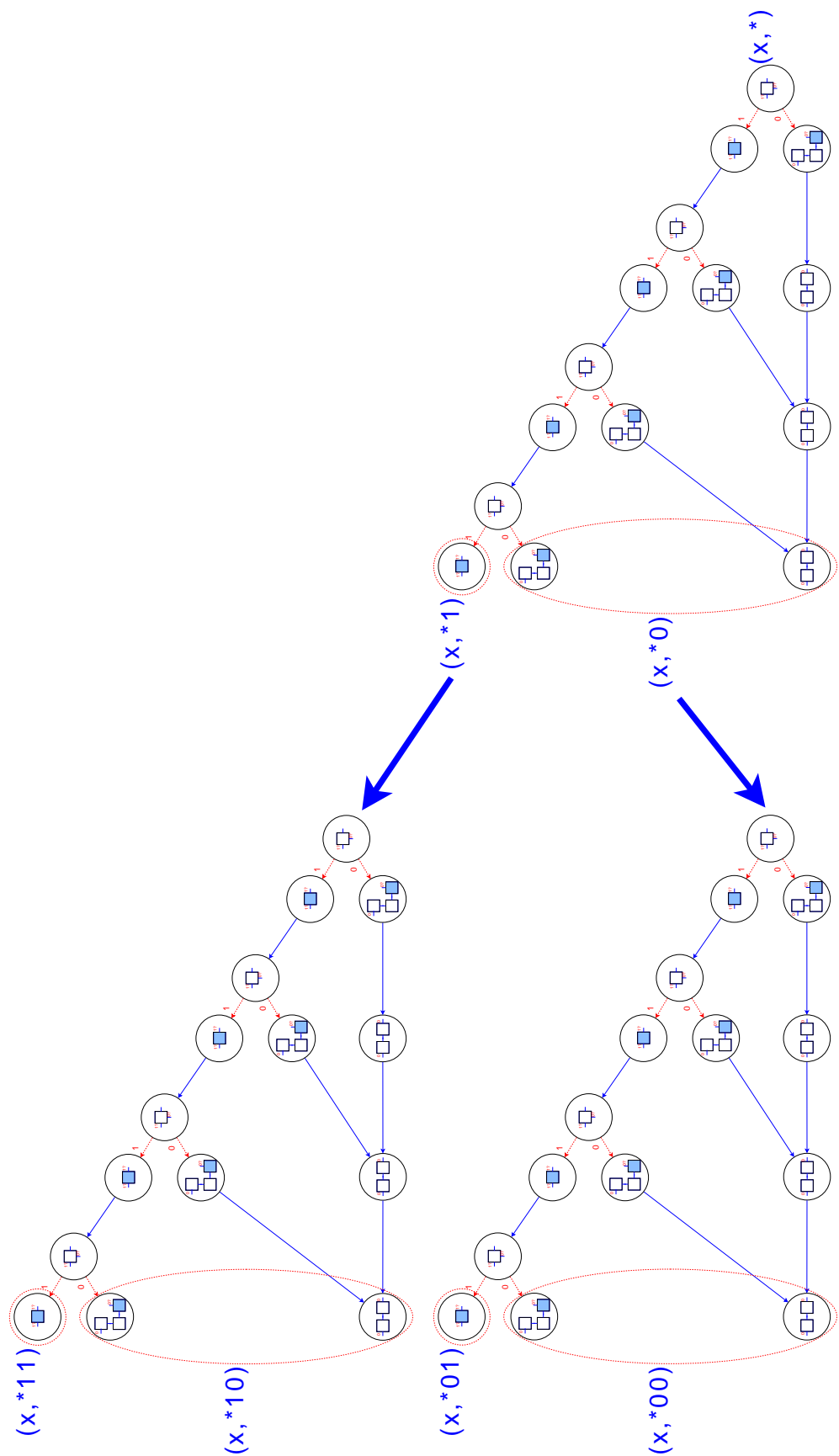}
 \caption{The logic binary tree for constructing one of decoding tile set (direction left).}\label{fig:dec2d2left}
\end{figure}